\documentclass[reprint,superscriptaddress,amsmath,amssymb,aps,pra]{revtex4-2}

\usepackage{cmap} 
\usepackage[utf8]{inputenc}
\usepackage[english]{babel}
\usepackage[T1]{fontenc}
\usepackage{appendix}
\usepackage{braket}
\usepackage{url, amsfonts, tikz, physics, listings, xcolor, graphicx, float}
\usepackage[shortlabels]{enumitem}
\usepackage{dsfont}
\usepackage{amsmath, amssymb, amstext, amscd, amsthm, makeidx, graphicx, url, mathrsfs, mathtools, longdivision, polynom, bbm, complexity}
\usepackage{nicefrac}
\usepackage[linesnumbered,ruled,vlined]{algorithm2e}
\usepackage{xcolor}
\usepackage[normalem]{ulem}

\definecolor{blueviolet}{rgb}{0.2, 0.2, 0.6}
\definecolor{webgreen}{rgb}{0,.5,0}
\definecolor{webbrown}{rgb}{.6,0,0}
\usepackage[pdftex,
  bookmarks=true,
  colorlinks=true, 
  urlcolor=webbrown,
  linkcolor=blueviolet, 
  citecolor=webgreen,
  pdfstartpage=1,
  pdfstartview={FitH},  
  bookmarksopen=false
  ]{hyperref}
\usepackage{cleveref}

\newcommand{\dtr}{d_{\mathrm{tr}}}

\newtheorem{lemma}{Lemma}

\newtheorem{definition}{Definition}
\newtheorem{remark}{Remark}

\begin{document}

\title{Learning to erase quantum states: \\
thermodynamic implications of quantum learning theory}

\author{Haimeng Zhao}
\email{haimengzhao@icloud.com}
\affiliation{\mbox{Institute for Quantum Information and Matter, California Institute of Technology, Pasadena, CA 91125, USA}}

\author{Yuzhen Zhang}
\email{yuzhenzhang@ucsb.edu}
\affiliation{Department of Physics, University of California, Santa Barbara, CA 93106, USA}

\author{John Preskill}
\email{preskill@caltech.edu}
\affiliation{\mbox{Institute for Quantum Information and Matter, California Institute of Technology, Pasadena, CA 91125, USA}}
\affiliation{AWS Center for Quantum Computing, Pasadena, CA 91125, USA}

\begin{abstract}
The energy cost of erasing quantum states depends on our knowledge of the states.
We show that learning algorithms can acquire such knowledge to erase many copies of an unknown state at the optimal energy cost.
This is proved by showing that learning can be made fully reversible and has no fundamental energy cost itself.
With simple counting arguments, we relate the energy cost of erasing quantum states to their complexity, entanglement, and magic.
We further show that the constructed erasure protocol is computationally efficient when learning is efficient.
Conversely, under standard cryptographic assumptions, we prove that the optimal energy cost cannot be achieved efficiently in general.
These results also enable efficient work extraction based on learning.
Together, our results establish a concrete connection between quantum learning theory and thermodynamics, highlighting the physical significance of learning processes and enabling provably-efficient learning-based protocols for thermodynamic tasks.
\end{abstract}

\maketitle

Do abstract learning processes have tangible physical consequences?
For example, does the (in)ability to learn impact the amount of physical resources required to perform certain tasks?

Learning is the process of acquiring information, which can be used to execute actions \cite{jaynesInformationTheoryStatistical1957a,jaynesInformationTheoryStatistical1957,jaynesProbabilityTheoryLogic2003,landauerInformationPhysical1991}.
Landauer's principle \cite{landauerIrreversibilityHeatGeneration1961,bennettLogicalReversibilityComputation1973,bennettThermodynamicsComputationReview1982,reebImprovedLandauerPrinciple2014,dahlstenInadequacyNeumannEntropy2011,feynmanFeynmanLecturesComputation2023} illustrates how the energy cost of erasing a physical system depends on our knowledge of the system.
In particular, given a system $S$ in a state $\rho$ with degenerate Hamiltonian $\mathcal{H}=0$, the amount of work required to transform it to a reference pure state $\ket{0}$ in an environment at temperature $T$ is $W\geq H(S)k_B T\ln 2$, where $k_B$ is the Boltzmann constant and $H(S)$ is the entropy of the state.
Entropy measures our ignorance of the state, and additional knowledge of the system reduces the optimal work cost to $H(S|M)k_B T \ln 2$, where $H(S|M)$ is the entropy of the system conditioned on the memory $M$ that records the knowledge \cite{rioThermodynamicMeaningNegative2011}. 
However, this statement does not take into account the potential cost of acquiring the knowledge recorded in the memory.

To formalize this process of ``learning to erase'', we consider the simple scenario in which a source repeatedly produces an unknown $n$-qubit state $\ket{\psi_x}$ from a class of $m$ possible states $\mathcal{C}=\{\ket{\psi_x}\}_{x=1}^m$.
At first, due to our ignorance of the state, we have to pay work to erase it.
As we collect more copies of the state $\ket{\psi_x}$, we progressively learn enough information to identify the state.
Once we have learned the state, we can erase additional copies without further work by reversing the state preparation unitary, $U_x^\dagger\ket{\psi_x}=\ket{0}$ 
\footnote{
By learning, we mean taking multiple copies of $\ket{\psi_x}$ as input and outputting a circuit description of $U_x$, the unitary that prepares $\ket{\psi_x}$. 
The state unpreparation unitary $U_x^\dagger$ can be easily derived from the circuit description of $U_x$.
This is very different from constructing $U_x^\dagger$ from quantum queries to $U_x$, which is generally hard \cite{yoshida2023reversing}.
See \Cref{sec:learning} for a formal definition.
}

In this paper, we rigorously study this process using tools recently developed in quantum learning theory \cite{zhaoLearningQuantumStates2024,huangLearningShallowQuantum2024,landauLearningQuantumStates2024,leoneLearningTdopedStabilizer2024,grewalEfficientLearningQuantum2024,cramerEfficientQuantumState2010,arunachalamOptimalAlgorithmsLearning2023,schuster2025random,anshuSurveyComplexityLearning2024} and explore its thermodynamic implications.
We explain a general method that lifts any learning algorithm to a reversible one, enabling it to acquire knowledge from multiple copies of the unknown state and erase many additional copies at the optimal energy cost, saturating Landauer's limit (Figure \ref{fig:overview}(a)).
We further pinpoint the quantitative relation between the energy cost of erasing a class of quantum states and their complexity, as measured by circuit depth, entanglement, or magic (Figure \ref{fig:overview}(b)).

After establishing the information-theoretical optimality of learning-to-erase protocols, we study their computational efficiency --- whether they can be implemented in a number of elementary operations that scales polynomially with the system size $n$ and number of copies $N$.
Computational efficiency is crucial in that if a protocol requires exponential time to execute, it would take longer than the age of the universe to erase only a few hundreds of qubits.
We show that whenever the learning algorithm and state preparation are efficient, the constructed erasure protocol is also efficient.
We illustrate this with a wide range of physically relevant classes of states, including shallow-circuit states, doped stabilizer states, matrix product states, and low-degree phase states.
Conversely, under standard cryptographic assumptions, we prove that there exist classes of low-complexity states that are hard to learn and also hard to erase. 
For these states, Landauer's principle allows $N$ copies to be erased by expending an amount of work independent of $N$, but for any efficient erasure protocol the cost scales nearly linearly with $N$.
This huge gap between the work cost of computationally-bounded agents versus what is information-theoretically possible is a genuine quantum many-body phenomenon and a much stronger no-go result than the third law of thermodynamics \cite{wilmingThirdLawThermodynamics2017,masanesGeneralDerivationQuantification2017,scharlauQuantumHornsLemma2018,tarantoLandauerNernstWhat2023}.
It also implies that the existing general purpose compress-to-erase methods for achieving Landauer's limit \cite{rioThermodynamicMeaningNegative2011,pleschEfficientCompressionQuantum2010,yangEfficientQuantumCompression2016,hayashiQuantumInformationTheory2017,bennett2006universal} require resources superpolynomial in $n$
\footnote{Note that the time efficiency claims in these previous works refer to a runtime of $\poly(N, D)$, where $D=2^n$ is the Hilbert space dimension of a single copy that grows exponentially with the system size.}.

Finally, we show that our learning-to-erase protocols can be used to extract work \cite{bennettThermodynamicsComputationReview1982,alickiThermodynamicsQuantumInformational2004,faistMinimalWorkCost2015} in Methods, where similar optimality and (in)efficiency results hold.
In particular, the no-go result implies that in general one can extract only a small fraction of all the possible work in polynomial time.
In contrast, given efficiently learnable states, we can efficiently extract the optimal amount of work by learning.

\begin{figure*}[ht]
    \centering
    \includegraphics[width=\linewidth]{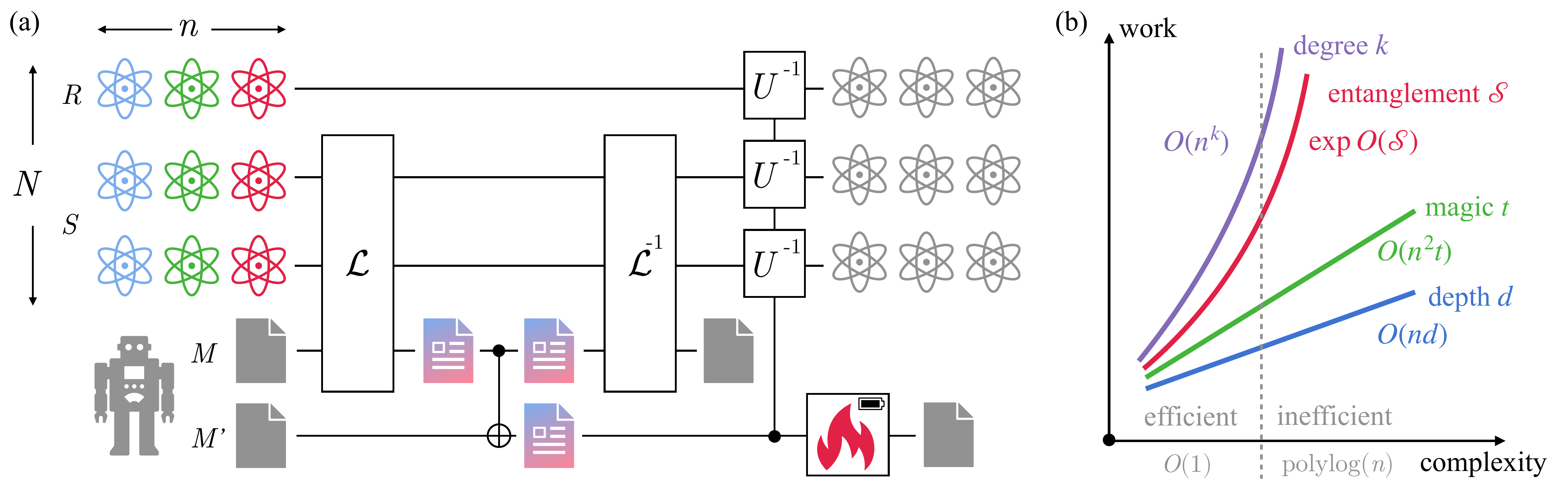}
    \caption{(a) Reversible learning algorithm $\mathcal{L}$ can acquire knowledge of the unknown $n$-qubit state and erase $N$ copies of it at the optimal work cost, saturating Landauer's limit.
    Here, $S$ stores the copies used by the learning algorithm, $R$ stores the rest, $M$ is the memory of the learning algorithm, and $M'$ is an auxiliary memory.
    (b) The work cost of erasing physically relevant classes of states grows with the complexity of the states, as measured by circuit depth $d$, magic $t$, entanglement entropy $\mathcal{S}$, and degree $k$. When the complexity is bounded by a constant, Landauer's limit can be achieved efficiently by learning. In contrast, when the complexity grows poly-logarithmically, no polynomial-time quantum algorithm can erase the states without paying a nearly maximal amount of work. 
    }
    \label{fig:overview}
\end{figure*}

\section{General Formulation}
We first consider the general case where the class of possible states $\mathcal{C}=\{\ket{\psi_x}\}_{x=1}^m$ is arbitrary and known (along with circuit descriptions of their state preparation unitaries $\{U_x\}_{x=1}^m$).
We assume that these states are distinct, in the sense that they are pairwise apart by a constant trace distance: $\forall x\neq x', \dtr(\ket{\psi_x}, \ket{\psi_{x'}})=\frac12 \|\ket{\psi_x}\bra{\psi_x} - \ket{\psi_{x'}}\bra{\psi_{x'}}\|_1\geq \epsilon = \Theta(1)$ \footnote{In this work, we use the standard notations for asymptotics. For two positive functions $f(n)$ and $g(n)$, $f(n) = O(g(n))$ if $\exists n_0, C>0$ such that $\forall n>n_0, f(n)\leq Cg(n)$. $f(n)=\Omega(g(n))$ if $g(n)=O(f(n))$. $f(n)=\Theta(g(n))$ if $f(n)=O(g(n))$ and $f(n)=\Omega(g(n))$.}.
The state we aim to erase can be described by $\rho = \sum_{x=1}^m p_x (\ket{\psi_x}\bra{\psi_x})^{\otimes N}$, where $\{p_x\}$ are the prior probabilities. 
We consider the state-agnostic setting, where we do not have knowledge of $\{p_x\}$ and thus the state $\rho$ that we want to erase.
For example, the distribution $\{p_x\}$ can be supported on any single state $\ket{\psi_x}$, or it can be uniform.
Our goal is to design a protocol that can erase $\rho$ with high success probability no matter what the distribution $\{p_x\}$ is.
We are allowed to use ancilla qubits as long as they are restored to the initial state, so that one cannot hide entropy or work cost in the ancilla qubits \cite{faistMinimalWorkCost2015}.
To do so, we describe a learning algorithm that identifies $x$, lift it to a reversible algorithm, and then use it to erase the state.
Throughout this work, we define learning algorithms as quantum algorithms that take multiple copies of an unknown state as input and output an approximate circuit description of that state with high probability.
See \Cref{sec:prelim} for a formal description.

We adopt a learning algorithm \cite{zhaoLearningQuantumStates2024} based on hypothesis selection \cite{badescuImprovedQuantumData2024} and Clifford classical shadows \cite{huangPredictingManyProperties2020}.
Given $s$ copies of the state $\ket{\psi_x}$, we apply a random $n$-qubit Clifford gate on each copy and measure in the computational basis.
The measurement statistics allow us to classically calculate unbiased estimators for the expectation values of $\binom{m}{2} = O(m^2)$ Helstrom measurements \cite{helstromQuantumDetectionEstimation1969} between all pairs of possible states in $\mathcal{C}$.
We can then use quantum hypothesis selection to find out the candidate state in $\mathcal{C}$ that is closest to $\ket{\psi_x}$ in trace distance.
When we have $s=O(\log (m)/\epsilon^2)=O(\log m)$ samples, we are guaranteed to find the correct $x$ with high probability $p_{\mathrm{succ}}$.
A detailed proof can be found in \cite[Proposition 1]{zhaoLearningQuantumStates2024}.

The above learning algorithm involves randomness generation and irreversible measurements that may incur extra work cost.
To avoid this, we describe a general method that lifts a learning algorithm to a reversible one.
Formally speaking, a learning algorithm $\mathcal{L}_0$ takes the system $\ket{\psi_x}^{\otimes s}_S$ and a classical memory register $\ket{0}_M$ as input and writes the learning outcome into $\ket{x}_M$ with high probability.
It is implemented using a set of quantum and classical gates, random number generators, and projective measurements.
We first replace every quantum and classical gates by their reversible counterparts \cite{bennettThermodynamicsComputationReview1982,fredkinConservativeLogic1982}.
Then we replace projective measurements by their coherent version: to measure a set of projectors $\{P_a\}$, we apply an unitary $U$ on the system and an ancilla $A$ that records the measurement outcome.
$U$ completes the isometry $\sum_a P_a \otimes \ket{a}\bra{0}$ such that $ U\ket{\psi}_{S}\ket{0}_A=\sum_a P_a\ket{\psi}_{S}\ket{a}_A$ for any $\ket{\psi}_S$. 
When classical random numbers are needed, we coherently measure $\ket{+}$ states in the computational basis.
All operations that depend on random numbers or measurement outcomes can be realized using reversible gates controlled by the ancilla.
In this way, we have constructed a reversible learning algorithm $\mathcal{L}$ (a unitary) that satisfies
\begin{equation}
    \mathcal{L}\ket{\psi_x}^{\otimes s}_S\ket{0}_M\ket{0}_A=\sum_{x'=1}^m c_{x'|x}\ket{x'}_M\ket{\mathrm{junk}_{x'}}_{S,A},
\end{equation}
where the coefficients $c_{x'|x} \in\mathbb{C}$ are determined by the probability of predicting $x'$ given the state $\ket{\psi_x}$.
In particular, $|c_{x|x}|^2\geq p_{\mathrm{succ}}$ is close to one for any $x$.

Now we construct an erasure protocol $\mathcal{E}(\mathcal{L})$ using the reversible learning algorithm $\mathcal{L}$.
The construction is summarized in Figure \ref{fig:overview}(a).
Given $N\geq s$ copies of $\ket{\psi_x}$, the protocol proceeds by first executing $\mathcal{L}$ on $s$ copies.
We call the $s$ copies $S$ and the rest $R$.
Then we introduce an additional classical memory $M'$ and use CNOT gates to copy the content of $M$ into $M'$.
Next, we apply the inverse $\mathcal{L}^{-1}$ to $S, M, A$ to uncompute \cite{nielsenQuantumComputationQuantum2010} and get rid of the junk. 
We erase each copy in $S$ and $R$ by applying the inverse state preparation unitary conditioned on the learning outcome recorded in $M'$.
Finally, we erase $M'$ bit by bit using standard classical Landauer erasure \cite{bennettThermodynamicsComputationReview1982}, costing energy $W=(\log_2 m)k_BT\ln 2$.
We prove the correctness of this erasure protocol in \Cref{sec:correctness} and provide a proof sketch in Methods.

We note that the only step that costs energy in this erasure protocol is the step of erasing learning outcomes, which costs $W=(\log_2 m)k_BT\ln 2$ joules of work, independent of $N$.
This confirms our intuition that once we have learned the state, we can erase many more copies without further work.
This general method of constructing reversible learning algorithms illustrates that learning as a physical process does not itself have a fundamental energy cost \cite{vsafranek2023work,xuereb2024resources}.
In principle, the energy cost only occurs when the learning agent needs to erase and recycle its memory.
In practice, there may be an extra energy cost if the learning algorithm is not implemented in a fully coherent and reversible way.

The sample complexity $s$ of the learning algorithm plays an interesting role in the erasure protocol.
The above construction works as long as $N\geq s$ and we have sufficient quantum working memory (ancilla supply) so that we have enough information to learn and space to make everything reversible.
Extra work may be needed when these assumptions are not met.
This shows that the usual $N\to \infty$ limit taken in thermodynamics is not necessary and can be significantly relaxed to $N\geq s$.
Assuming that the learning algorithm and the state preparation unitary have time complexity $T_{\mathrm{learn}}$ and $T_{\mathrm{prep}}$, the total time complexity of $\mathcal{E}(\mathcal{L})$ is $O(T_{\mathrm{learn}}+\log m+NT_{\mathrm{prep}})$.
Therefore, whenever the learning and state preparation steps are efficient, the erasure protocol is also efficient.

We show that the above work cost is optimal by calculating the lower bound given by Landauer's principle.
Since we are only given one copy of $\rho$ (which contains $N$ copies of $\ket{\psi_x}$), the Landauer's limit is given by the one-shot max-entropy $W\geq H_{\mathrm{max}}(\rho)k_BT\ln 2$, where $H_{\mathrm{max}}(\rho) = \log_2 (\mathrm{rank}(\rho))$ is the max-entropy of $\rho$ \cite{dahlstenInadequacyNeumannEntropy2011,rioThermodynamicMeaningNegative2011,faistMinimalWorkCost2015,rennerSmoothRenyiEntropy2004,konigOperationalMeaningMin2009}. 
To calculate the rank of $\rho$, we construct the Gram matrix $G_{xx'}=\braket{\psi_x}{\psi_{x'}}^N$.
Assuming $N> 2\log(m-1)/\log(1/(1-\epsilon^2))=\Theta(\log m)$ (which is naturally satisfied when $N\geq s$ and we always assume this from now on), we have $\sum_{x'\neq x} |G_{xx'}| \leq \sum_{x'\neq x}(1-\epsilon^2)^{N/2}< \sum_{x'\neq x}(m-1)^{-1} =1 =G_{xx}, \forall x$.
This means that $G$ is diagonally dominant and thus non-singular, implying that $\{\ket{\psi_x}^{\otimes N}\}_{x=1}^m$ are linearly independent.
Therefore, when we maximize over all possible $\{p_x\}$ in the state-agnostic setting, we have $\mathrm{rank}(\rho) = m$ and Landauer's limit reads $(\log_2 m) k_BT\ln 2$.
This coincides with the work cost of $\mathcal{E}(\mathcal{L})$, proving that the learning-to-erase protocol is information-theoretically optimal.

We remark that the above construction works for any valid learning algorithm, not just for the one described above.
In fact, a simple compress-to-erase protocol commonly used in the literature \cite{dahlstenInadequacyNeumannEntropy2011,rioThermodynamicMeaningNegative2011} can already achieve the optimal energy cost (see \Cref{sec:prelim}).
However, this compress-to-erase protocol has exponential runtime, and in fact we will show that this computational inefficiency is inevitable.
In contrast, we will see below that learning-to-erase protocols are provably-efficient and energy-optimal for physically relevant classes of states.
Moreover, in regimes where our protocols are inefficient, our computational hardness results show that no efficient protocols can achieve the optimal energy cost.

\section{Work and complexity}
We now apply the general formulation to several physically relevant classes of states, and illustrate how the work cost of erasure grows with the complexity of the states in the target class.
For these classes, we give time efficient erasure protocols built from efficient learning algorithms.
The results are summarized in Figure \ref{fig:overview}(b).

Our first example is shallow-circuit states $\mathcal{C}=\{\ket{\psi} = U\ket{0^n}\}$, where $U$ can be an arbitrary quantum circuit with depth $d$ on a $O(1)$-dimensional lattice built from a finite-size two-qubit gate set $G$ \cite{huangLearningShallowQuantum2024,landauLearningQuantumStates2024,kimLearningStatePreparation2024}.
This class naturally describes the output state of finite-time evolution of physical systems \cite{poulinQuantumSimulationTimedependent2011} and noisy intermediate-scale quantum computation \cite{preskillQuantumComputingNISQ2018}, offering quantum advantage over classical computation \cite{terhalAdaptiveQuantumComputation2004,gaoQuantumSupremacySimulating2017,giovannettiArchitecturesQuantumRandom2008,haferkampClosingGapsQuantum2020,hangleiterComputationalAdvantageQuantum2023,bravyiQuantumAdvantageShallow2018}.
The number of possible states in this class is $m=|G|^{\Theta(nd)}$.
Thus the required work cost is 
\begin{equation}
    W=\Theta(nd)k_BT\ln 2,
\end{equation}
growing linearly with circuit depth $d$, when $N$ is sufficiently large.
In Methods, we generalize this to the case where each gate is non-discrete and can be any two-qubit unitary.
The same work cost scaling applies up to logarithmic factors.
For constant $d$, this optimal work cost is achieved efficiently using $\mathcal{E}(\mathcal{L})$ with a polynomial-time learning algorithm $\mathcal{L}$ that learns the local inversions of the shallow circuit and sews them together (the discrete gate set version of \cite{landauLearningQuantumStates2024}).
In particular, we have $T_{\mathrm{learn}}=\poly(n)\cdot 2^{d^{O(1)}}$ and $T_{\mathrm{prep}}=O(nd)$, and the total time complexity of erasure is $O(\poly(n)\cdot 2^{d^{O(1)}}+ndN)$.

A second example is $t$-doped stabilizer states $\mathcal{C}=\{\ket{\psi} = U\ket{0^n}\}$, where $U$ is composed of an arbitrary number of Clifford gates and at most $t$ $T$-gates ($T=\mathrm{diag}(1, e^{i\pi/4})$) \cite{grewalEfficientLearningQuantum2024,leoneLearningTdopedStabilizer2024}.
Here, $t$ characterizes the amount of magic (or non-stabilizerness) in the state \cite{liuManyBodyQuantumMagic2022}, which is closely related to universality \cite{shiBothToffoliControlledNOT2002,bravyiUniversalQuantumComputation2005}, classical simulability \cite{aaronsonImprovedSimulationStabilizer2008,bravyiImprovedClassicalSimulation2016,bravyiSimulationQuantumCircuits2019}, and error correction \cite{knillFaultTolerantPostselectedQuantum2004,knillTheoryQuantumError2000} of quantum computation.
The number of possible states in the class is $m=(2^{\Theta(n^2)})^{t+1} \times n^t = \exp\Theta(n^2(t+1))$ \cite{aaronsonImprovedSimulationStabilizer2008} because there are 
$2^{\Theta(n^2)}$ possible Clifford unitaries in between each pair of $T$-gates.
When $N$ is large, the work cost for erasure is 
\begin{equation}
    W=\Theta(n^2(t+1))k_BT\ln 2,
\end{equation} 
growing linearly with the magic $t$.
In particular, the work cost for erasing stabilizer states is $\Theta(n^2)k_B T \ln 2$.
For $t=O(\log n)$, this is achieved efficiently by an efficient learning algorithm for $t$-doped stabilizer states, which identifies the state exactly by bell sampling \cite{montanaroLearningStabilizerStates2017} and estimating the stabilizer group (see \cite{leoneLearningTdopedStabilizer2024}).
We have $T_{\mathrm{learn}}=\poly(n, 2^t)$ and $T_{\mathrm{prep}}=O(n^2t)$; thus the total time complexity of erasure is $O(\poly(n, 2^t)+n^2tN)$.

Our third example is matrix product states (MPS) $\mathcal{C}=\{\ket{\psi} = \sum_{x\in \{0, 1\}^n} A_1^{x_1}\cdots A_n^{x_n} \ket{x}\}$, where $A_i^x\in \mathbb{C}^{\chi_i\times\chi_{i+1}}$ with bond dimension $\chi_1=\chi_{n+1}=1$ and $\chi_i\leq 2^{\mathcal{S}}$ \cite{cramerEfficientQuantumState2010}.
Here, $\mathcal{S}$ is the maximal entanglement entropy across every cut of the system.
MPS representations with constant $\mathcal{S}$ exist for states with bounded entanglement, including ground states of one-dimensional gapped local Hamiltonians \cite{hastingsAreaLawOne2018} and symmetry-protected topological phases \cite{zengQuantumInformationMeets2019}, as well as important states like AKLT \cite{affleckRigorousResultsValencebond1987}, GHZ, W and cluster states that are useful for quantum metrology \cite{degenQuantumSensing2017} and universal quantum computation \cite{jozsaIntroductionMeasurementBased2005}.
In Methods, we show that the work cost for erasure 
\begin{equation}
    W=\mathrm{exp}(\tilde{\Theta}(\mathcal{S}))k_BT\ln 2
\end{equation} 
grows exponentially with the entanglement entropy $\mathcal{S}$.
For constant $\mathcal{S}$, this is achieved efficiently using an efficient learning algorithm for MPS that sequentially learns $n$ $O(\mathcal{S})$-qubit unitaries to disentangle the MPS (see \cite{cramerEfficientQuantumState2010}).
We have $T_{\mathrm{learn}}=\poly(n, 2^{\mathcal{S}})$ and $T_{\mathrm{prep}}=O(n4^{\mathcal{S}})$, and the total time complexity of erasure is $O(\poly(n, 2^{\mathcal{S}})+n4^\mathcal{S}N)$.

The last example is low-degree phase states $\mathcal{C} = \{\ket{\psi} = 2^{-n/2}\sum_{x\in \{0, 1\}^n}(-1)^{f(x)}\ket{x}\}$, where $f: \{0, 1\}^n\to \{0, 1\}$ is an arbitrary Boolean function with degree at most $k=O(1)$ \cite{arunachalamOptimalAlgorithmsLearning2023}.
They correspond to hypergraph states that are generated by $k$-body multi-controlled Z gates on $\ket{+}^{\otimes n}$.
The degree $k$ is closely related to the level of Clifford hierarchy \cite{hu2021climbing} and the amount of quantum contextuality \cite{raussendorf2013contextuality}.
They are useful resources in measurement-based quantum computation \cite{raussendorfOneWayQuantumComputer2001,jozsaIntroductionMeasurementBased2005} and appear in quantum advantage experiments \cite{bremnerClassicalSimulationCommuting2010,bremnerAchievingQuantumSupremacy2017} and oracle query algorithms \cite{nielsenQuantumComputationQuantum2010}. By counting polynomial Boolean functions, we see that the number of such states is $m=\prod_{j=0}^k2^{\binom{n}{j}}=\exp \Theta(n^k)$ and the required work cost is 
\begin{equation}
    W=\Theta(n^k)k_BT\ln 2,
\end{equation}
growing exponentially with degree $k$.
For constant $k$, the erasure is efficient using an efficient learning algorithm that measures single-qubit $X$ and $Z$ observables and learns the gradient of $f$ with respect to each argument (see \cite[Theorem 3]{arunachalamOptimalAlgorithmsLearning2023}).
Here, $T_{\mathrm{learn}}=O(n^{3k-2})$ and $T_{\mathrm{prep}}=O(kn^k)$ \cite{gidneyConstructingLargeControlled2015}, and the total time complexity of erasure is $O(n^{3k-2}+k n^k N)$.

In Figure \ref{fig:overview}(b), we summarize the above results and visualize how the optimal work cost of erasure grows with the four different complexity measures: circuit depth $d$, magic $t$, entanglement entropy $\mathcal{S}$, and degree $k$.
In particular, the work cost scales linearly with depth ($\sim nd$) and magic ($\sim n^2t$), where the growth with magic has a larger slope.
Meanwhile, the work cost scales exponentially with entanglement entropy ($\sim e^\mathcal{S}$) and degree ($\sim n^k$), where the growth with degree is faster.
The computational complexity of erasure grows exponentially with all four measures of complexity. 
When the complexity is bounded by a constant, the optimal work cost can be achieved efficiently by learning.

\section{Computational Hardness}
For the above special classes of states, we have constructed efficient erasure protocols with optimal work cost.
However, we show that this is not possible in general.
We consider a particular class of states, pseudorandom states $\{\ket{\psi_x}\}$ \cite{jiPseudorandomQuantumStates2018,maHowConstructRandom2024}, which cannot be efficiently distinguished from Haar random states with non-negligible ($1/\poly(n)$) probability.
Concretely, under the standard cryptographic conjecture that there exist one-way functions secure against any sub-exponential time quantum adversary (e.g., based on the learning with errors problems \cite{regevLatticesLearningErrors2009}), pseudorandom states can be constructed with $d=\polylog(n)$ circuit depth \cite{schuster2025random}.
Let $N=\poly(n)$.
Because the entropy of $\rho=\sum_x p_x(\ket{\psi_x}\bra{\psi_x})^{\otimes N}$ depends on $m$ but not $N$, Landauer's principle asserts that the $N$ copies of pseudorandom states can be erased with low work cost $\Theta(nd)k_BT\ln 2=\Theta(n\polylog(n))k_BT\ln 2$, independent of $N$. 
On the other hand, for the Haar ensemble, the Landauer cost of erasure is 
\begin{equation}
    W_{\mathrm{Haar}}= \log_2\binom{N+2^n-1}{N}k_BT\ln 2,
\end{equation}
because the density operator realized by the Haar ensemble has full rank over the totally symmetric subspace for $N$ copies of a $2^n$-dimensional Hilbert space, which has dimension $\binom{N+2^n-1}{N}$.
In the following, we prove that any polynomial-time quantum algorithm that can erase a pseudorandom state $\rho=\sum_x p_x(\ket{\psi_x}\bra{\psi_x})^{\otimes N}$ with success probability close to one must require work cost $W_{\mathrm{Haar}}$.
The intuition is simple: if we can erase pseudorandom states with less work than $W_{\mathrm{Haar}}$, then we can distinguish them from Haar random states by measuring the work cost, violating the definition of pseudorandom states.

To prove the statement, we suppose for the sake of contradiction that there exists an efficient erasure protocol $\mathcal{E}$ that erases $\rho$ with work cost less than $W_{\mathrm{Haar}}$.
We construct a distinguisher $\mathcal{D}$ that distinguishes pseudorandom states from Haar random states.
The distinguisher proceeds by executing $\mathcal{E}$ on the given states and conducting two tests: (1) determine if the erasure is successful by measuring the projector $\ket{0}\bra{0}$; and (2) test if the work cost is less than $W_{\mathrm{Haar}}$ by measuring the energy change in the energy source \footnote{Here, we assume that the energy source is classical and can be measured deterministically. If the energy source is quantum, we can always measure the work cost to $1/\poly(n)$ accuracy in polynomial time. Thus the work cost lower bound in the no-go result will change to $W_{\mathrm{Haar}}-\mathrm{negl}(n)$ where $\mathrm{negl}(n)$ is a function that decays faster than any polynomial of $n$.
In \Cref{sec:pseudorandom}, we provide the proof details of the fully quantum case}.
$\mathcal{D}$ outputs $0$ if both tests pass, and $1$ otherwise.

When given pseudorandom states, we have $\Pr[\mathcal{D}((\ket{\psi_x}\bra{\psi_x})^{\otimes N})=0] = \Pr[\mathcal{D}(\rho)=0]$ close to one, where the probability is over $x$ and the randomness in $\mathcal{D}$; here we used linearity of quantum channels and the guarantee that $\mathcal{E}$ can erase pseudorandom states.
When given Haar random states $\rho_{\mathrm{Haar}} = \int \mathrm{d}\ket{\psi}(\ket{\psi}\bra{\psi})^{\otimes N}$, the probability of outputting $0$ is equal to $\Pr[\mathcal{D}(\ket{\psi}\bra{\psi})=0]=\Pr[\mathcal{D}(\rho_{\mathrm{Haar}})=0]$, the probability of one-shot successful erasure of $\rho_{\mathrm{Haar}}$ with work cost less than $W_{\mathrm{Haar}}$.
Because the erasure protocol $\mathcal{E}$ is assumed to be efficient,
this probability must also be close to one by the definition of pseudorandom states.
But this contradicts Landauer's principle, which asserts that it is not possible to perform one-shot erasure with work cost less than $H_{\max}(\rho_{\mathrm{Haar}})k_B T\ln 2 = \log_2(\mathrm{rank}(\rho_{\mathrm{Haar}}))k_BT\ln2 =W_{\mathrm{Haar}}$.
This completes the proof of the no-go result.
We provide more details of the proof in \Cref{sec:pseudorandom}.

For a quantitative understanding, we note that $W_{\mathrm{Haar}}/(k_BT\ln2)\geq nN(1-(\log_2 N)/n)$, which goes to the maximal value $nN$ as $n$ grows large with $N=\poly(n)$.
But Landauer's limit for this class of states is only $\Theta(nd) = n\polylog(n)$, independent of $N$.
This means that Landauer's limit cannot be achieved in polynomial time, and asymptotically maximal work must be paid.

We remark that this is a much stronger no-go result than the third law of thermodynamics \footnote{
We consider the unattainability formulation of the third law of thermodynamics, which applies not only to the cooling of thermal states but also the erasure/purification of arbitrary states. See e.g., \cite{masanesGeneralDerivationQuantification2017}.
}, which asserts that perfect erasure as a form of cooling requires diverging resources (time, size of the heat bath, etc.) \cite{wilmingThirdLawThermodynamics2017,masanesGeneralDerivationQuantification2017,scharlauQuantumHornsLemma2018,tarantoLandauerNernstWhat2023}.
The third law does not rule out the possibility of achieving Landauer's limit in polynomial time by paying $1/\poly(n)$ extra work and using an infinite size heat bath. 
In contrast, our result shows that even with an infinite size heat bath, any procedure that can erase a particular class of low-complexity states in polynomial time must pay a nearly maximal amount of work.
This is a many-body phenomenon due to the complexity of large systems.

\section{Discussion}
We have established a concrete connection between quantum learning theory and thermodynamics.
Our results illustrate that seemingly abstract learning processes have tangible physical consequences that determine the energy cost or gain in thermodynamic tasks.
This allows us to study properties of learning itself from a thermodynamic perspective.
In particular, we find that learning has no fundamental energy cost, if implemented in a fully coherent fashion.

This connection also enables provably-efficient and energy-optimal thermodynamic protocols based on learning, improving upon existing compression-based protocols that have exponential time complexity \cite{rioThermodynamicMeaningNegative2011,pleschEfficientCompressionQuantum2010,yangEfficientQuantumCompression2016,hayashiQuantumInformationTheory2017,bennett2006universal}.
This contributes to a growing literature on state-agnostic thermodynamic and compression protocols \cite{yangCompressionQuantumShallowcircuit2025,watanabe2024black,xuereb2024resources,chakraborty2024sample}.
In particular, we circumvent the potential energy cost of measurements and revive the use of learning algorithms for thermodynamic tasks \cite{xuereb2024resources}.
In a way, learning is compression --- compressing quantum states into their classical descriptions --- and efficient learning yields efficient compression \cite{yangCompressionQuantumShallowcircuit2025}.
These efficient erasure protocols may find applications in the initialization of distributed quantum metrology systems and in recycling qubits for repeated use in quantum experiments.
They also provide a viable way to experimentally demonstrate Landauer's limit \cite{berutExperimentalVerificationLandauers2012,gaudenziQuantumLandauerErasure2018,scandiMinimallyDissipativeInformation2022} and materialize Maxwell's demon as a learning agent \cite{koskiExperimentalRealizationSzilard2014,vidrighinPhotonicMaxwellsDemon2016,camatiExperimentalRectificationEntropy2016,cottetObservingQuantumMaxwell2017} in the quantum many-body regime.
These efficient energy-optimal protocols may be a key component in building future large-scale, energy-efficient quantum computers \cite{auffevesQuantumTechnologiesNeed2022}.

Our results also exemplify the drastic change of physical resources required by certain tasks when we are in the many-body regime.
Ideas that date back to the early days of information theory, like Landauer's principle, may require reexamination via the lens of complexity \cite{munsonComplexityconstrainedQuantumThermodynamics2024,leoneEntanglementTheoryLimited2025}.

Our work opens up many interesting future directions.
First, our setting can be extended to more realistic scenarios, where we want to erase many copies of mixed states $\rho = \sum_x p_x \rho_x^{\otimes N}$.
A natural choice could be thermal states of unknown Hamiltonians that describe physical systems in thermal equilibrium.
Efficient learning and erasing may be possible in this case \cite{bakshi2024learning,chen2025learning}, where we expect the optimal work cost to scale with $N$: $\sum_x p_x NH(\rho_x) + H(\{p_x\})$, because after we have identified $\rho_x$, we still need to pay work to erase each copy because it is mixed.
Second, our results may be extended to continuous variable systems such as fermionic or bosonic systems and inspire efficient thermodynamic protocols there \cite{meleEfficientLearningQuantum2025,meleLearningQuantumStates2024}.
Third, learning may yield efficient protocols for other thermodynamic tasks beyond erasure and work extraction.

Our computational hardness results may find cryptological applications related to energy storage, manipulation, and transfer.
For example, one may envision an encrypted battery, where any computationally-bounded agent without the secret key cannot extract much work from the battery even when having access to the physical substrate of the battery; the full energy is only efficiently extractable when one has the secret key.
Intriguingly, limitations on erasing information in or extracting work from black holes may apply if the evolution of these systems is modeled accurately by pseudorandom unitaries \cite{haydenBlackHolesMirrors2007,boulandComputationalPseudorandomnessWormhole2020,akersHolographicPseudoentanglementComplexity2024}.
Finally, it may be instructive to study other physical properties of learning processes beyond energy cost, such as their consequences for information dynamics \cite{swingleMeasuringScramblingQuantum2016}, phase transitions \cite{zdeborovaStatisticalPhysicsInference2016}, and noise robustness \cite{anshuSurveyComplexityLearning2024}.

\section*{Note added}
During the revision of this manuscript, we became aware of related independent works by Watanabe and Takagi \cite{watanabe2025universal} and Lumbreras et al. \cite{lumbreras2025quantum}, released subsequent to the initial posting of our preprint.
Ref.~\cite{watanabe2025universal} addresses the information-theoretic achievability of state-agnostic work extraction when infinitely many copies are given ($N\to \infty$).
Ref.~\cite{lumbreras2025quantum} applies a known regret bound from online learning of quantum states to online work extraction from qubits ($n=1$).
Both works start from the premise that standard tomography destroys many copies of the states and uses measurements that may incur additional energy cost.
In contrast, our work shows that one can circumvent these costs by making tomography fully coherent and reversible, thereby achieving the optimal energy cost by learning.

Our work also studies the computational complexity aspect of erasure and work extraction, providing provably efficient protocols for structured states and proving no-go theorems for high-complexity states in the many-body regime ($n \gg 1$).
This aspect is not considered in \cite{watanabe2025universal,lumbreras2025quantum}.
In particular, Ref.~\cite{lumbreras2025quantum} is restricted to single-qubit pure states ($n=1$).
Ref.~\cite{watanabe2025universal} studies mixed states in general systems and has an excess energy cost per copy that scales as $\sim 2^{O(n)}/N$, which approaches zero only when the number of copies $N$ is exponential in $n$.
In fact, our no-go theorem proves that no general-purpose protocol can have runtime polynomial in $n$ while dissipating energy sublinear in $N$, as doing so would break post-quantum cryptography.
As a corollary, the $\polylog(N)$ dissipation scaling derived in \cite{lumbreras2025quantum} for qubits ($n=1$) cannot be generalized to computationally efficient protocols for many-body systems.

In addition, our work relates the energy cost of erasing quantum states to their structural properties, including circuit depth, entanglement, magic, and Boolean function degree.
This relation between different resources has not been addressed before.
We provide a detailed discussion of other related works in \Cref{sec:rel-work}.

\section{Methods}
\subsection{Proof of correctness}
Here we provide the proof sketch for the correctness of the erasure protocol $\mathcal{E}(\mathcal{L})$.
A more detailed proof can be found in \Cref{sec:correctness}.
The key is to calculate the overlap between $\ket{x}_{M'}\ket{\psi_x}^{\otimes s}_{S}\ket{0}_{M,A}$ and the intermediate state of $\mathcal{E}(\mathcal{L})$ after uncomputation (apart from $R$):
\begin{equation}
    \begin{split}
        &(\bra{x}_{M'}\bra{\psi_x}^{\otimes s}_S \bra{0}_{M,A}) \Bigl(\mathcal{L}^\dagger \sum_{x'}c_{x'|x}\ket{x'}_{M'}\ket{x'}_M\ket{\mathrm{junk}_{x'}}_{S,A}\Bigr) \\
        &=\sum_{x'}c_{x'|x}\braket{x}{x'}_{M'}\Bigl(\bra{\psi_x}^{\otimes s}_S \bra{0}_{M,A} \mathcal{L}^\dagger \Bigr) \! \left(\ket{x'}_M\ket{\mathrm{junk}_{x'}}_{S,A}\right) \\&=c_{x|x}\Bigl(\sum_{x''}c^*_{x''|x}\bra{x''}_M\bra{\mathrm{junk}_{x''}}_{S,A}\Bigr)\ket{x}_M\ket{\mathrm{junk}_x}_{S,A} \\
        &=|c_{x|x}|^2\geq p_{\mathrm{succ}},
    \end{split}
\end{equation}
which means that the trace distance between the two states is $\sqrt{1-|c_{x|x}|^4}\leq \sqrt{1-p_{\mathrm{succ}}^2}$.
Due to the data processing inequality satisfied by the trace distance \cite{wildeQuantumInformationTheory2013}, after executing the rest of the erasure protocol (a fixed quantum channel), we have $\dtr(\mathcal{E}(\mathcal{L})(\ket{\psi_x}^{\otimes N}\ket{0}), \ket{0})\leq \sqrt{1-p_{\mathrm{succ}}^2}$.
When the input is the mixture $\rho=\sum_x p_x(\ket{\psi_x}\bra{\psi_x})^{\otimes N}$, we invoke the triangle inequality and have trace distance
$\frac12\|\mathcal{E}(\mathcal{L})(\rho \otimes \ket{0}\bra{0})- \ket{0}\bra{0}\|_1\leq \sum_x p_x \dtr(\mathcal{E}(\mathcal{L})(\ket{\psi_x}^{\otimes N}\ket{0}), \ket{0})\leq \sqrt{1-p_{\mathrm{succ}}^2}$, which is close to zero.
This concludes the proof of the correctness. 

\subsection{Continuous classes of states}
It is straightforward to generalize the results to erase a continuous class of states $\mathcal{C}$.
We can use the erasure protocol for a discrete class in conjunction with a minimal $\epsilon$-covering net of $\mathcal{C}$, similar to \cite[Theorem 1]{zhaoLearningQuantumStates2024}.
The resulting erasure protocol will be able to approximately erase every copy to $\epsilon$ error in trace distance.
The work cost will be $(\log_2 m(\mathcal{C}, \dtr, \epsilon))k_B T\ln 2$ where $m(\mathcal{C}, \dtr, \epsilon)$ is the size of the covering net.
Meanwhile, this protocol can perfectly erase states drawn from a maximal $\Theta(\epsilon)$-packing net of $\mathcal{C}$.
By Landauer's principle, this requires work $(\log_2 \tilde{m}(\mathcal{C}, \dtr, \Theta(\epsilon))k_BT\ln2$ when $N$ is large, where $\tilde{m}(\mathcal{C}, \dtr, \Theta(\epsilon))$ is the size of the packing net.
Since $m$ and $\tilde{m}$ are equivalent up to a constant factor change in $\epsilon$ \cite{vershyninHighDimensionalProbabilityIntroduction2018}, the erasure protocol is asymptotically optimal.

As an example, we bound the size of the $\epsilon$-covering and packing nets of matrix product states.
Every MPS with entanglement entropy $\mathcal{S}$ can be generated by a sequential quantum circuit with $n$ gates each acting on $O(\mathcal{S})$ qubits \cite{schonSequentialGenerationEntangled2005}.
Utilizing existing covering number bounds on circuits with bounded gate complexity and fixed circuit structure \cite[Theorem 8]{zhaoLearningQuantumStates2024}, we have $\log m(\mathcal{C}, \dtr, \epsilon)\leq O(n4^{O(\mathcal{S})}\log(n4^{O(\mathcal{S})}/\epsilon))=\exp \tilde{O}(\mathcal{S})$, where $\tilde{O}$ omits $\log n, \log\log(1/\epsilon)$ factors.
To lower bound the packing number $\tilde{m}$, we note that an MPS with entanglement entropy $\mathcal{S}$ can represent any $\Theta(\mathcal{S})$-qubit state tensored with zeroes on the remaining qubits.
The set of general $\Theta(\mathcal{S})$-qubit states have log-packing number $\exp\tilde{\Omega}(\mathcal{S})$.
Thus we have $\log \tilde{m}(\mathcal{C}, \dtr, \epsilon)\geq \exp\tilde{\Omega}(\mathcal{S})$. 
This means that the work cost of erasure is $W=\mathrm{exp}(\tilde{\Theta}(\mathcal{S}))k_BT\ln 2$.
Similarly, one can bound the covering and packing numbers of depth-$d$ shallow circuit states with arbitrary two-qubit gates as $\exp\tilde{\Theta}(nd)$ \cite{zhaoLearningQuantumStates2024}, which implies the work cost of erasure $W = \tilde{\Theta}(nd)k_BT\ln 2$ when each gate is not assumed to be discrete.

\subsection{The classical case}
It helps to compare with the classical version of the problem.
Classically, the state is a mixture of $m\leq 2^n$ possible computational basis states copied $N$ times.
The entropy is $\log_2m\leq n$, which does not increase with $N$, even when the bit-strings are generated by pseudorandom functions.
We can always erase $N$ copies with work cost at most $nk_BT\ln 2$.
Therefore, it is a genuine quantum feature to have a work cost $W_{\mathrm{Haar}}$ that scales with $N$.
This originates from the fact that there are many more states when we allow superposition.
The optimal classical erasure protocol is straightforward: we condition on the first copy of the bit-string, apply transversal CNOT gates to erase the other copies, and finally erase the first copy.
From a learning perspective, this protocol learns the bitstring by looking at the first copy, and the knowledge of this one copy is sufficient for erasing the rest.
In contrast, in the fully quantum case, many copies are needed to learn a state.

\subsection{Work extraction}
The thermodynamic implications of quantum learning theory extend beyond erasure.
In the task of work extraction, it is information-theoretically possible to extract $W=(nN-H(\rho))k_BT\ln 2$ work from the state $\rho$ \cite{bennettThermodynamicsComputationReview1982,feynmanFeynmanLecturesComputation2023,alickiThermodynamicsQuantumInformational2004,faistMinimalWorkCost2015}.
When learning is efficient, our learning-to-erase protocols give efficient work extraction protocols with the maximal work yield: we simply erase the state first, and then extract work from each refreshed qubit.
In contrast, there exist ensembles such that no polynomial-time quantum algorithm can extract more work than $nNk_BT\ln 2-W_{\mathrm{Haar}}\leq N\log_2 (N)k_BT \ln 2$.
Otherwise, one could construct an efficient erasure protocol with less than $W_{\mathrm{Haar}}$ work cost for pseudorandom states by first extracting work and then erasing each qubit (which costs $nNk_BT\ln 2$ work and $O(nN)=\poly(n)$ time).
Therefore, for general states, one can extract in polynomial time only a fraction $\log_2(N)/n=O(\log (n)/n)$ of the available work, which becomes negligibly small for $n$ large.

\section*{Acknowledgments}
We thank Soonwon Choi, Bartek Czech, Jens Eisert, Philippe Faist, Yingfei Gu, Hsin-Yuan Huang, Jarrod R. McClean, Jinzhao Wang, and Yuxiang Yang for helpful discussions. 
We gratefully acknowledge support from the U.S. Department of Energy, Office of Science, National Quantum Information Science Research Centers, Quantum Systems Accelerator, and the National Science Foundation (PHY-2317110). 
The Institute for Quantum Information and Matter is an NSF Physics Frontiers Center.




\bibliography{ref}

\begin{thebibliography}{116}%
\makeatletter
\providecommand \@ifxundefined [1]{%
 \@ifx{#1\undefined}
}%
\providecommand \@ifnum [1]{%
 \ifnum #1\expandafter \@firstoftwo
 \else \expandafter \@secondoftwo
 \fi
}%
\providecommand \@ifx [1]{%
 \ifx #1\expandafter \@firstoftwo
 \else \expandafter \@secondoftwo
 \fi
}%
\providecommand \natexlab [1]{#1}%
\providecommand \enquote  [1]{``#1''}%
\providecommand \bibnamefont  [1]{#1}%
\providecommand \bibfnamefont [1]{#1}%
\providecommand \citenamefont [1]{#1}%
\providecommand \href@noop [0]{\@secondoftwo}%
\providecommand \href [0]{\begingroup \@sanitize@url \@href}%
\providecommand \@href[1]{\@@startlink{#1}\@@href}%
\providecommand \@@href[1]{\endgroup#1\@@endlink}%
\providecommand \@sanitize@url [0]{\catcode `\\12\catcode `\$12\catcode `\&12\catcode `\#12\catcode `\^12\catcode `\_12\catcode `\%12\relax}%
\providecommand \@@startlink[1]{}%
\providecommand \@@endlink[0]{}%
\providecommand \url  [0]{\begingroup\@sanitize@url \@url }%
\providecommand \@url [1]{\endgroup\@href {#1}{\urlprefix }}%
\providecommand \urlprefix  [0]{URL }%
\providecommand \Eprint [0]{\href }%
\providecommand \doibase [0]{https://doi.org/}%
\providecommand \selectlanguage [0]{\@gobble}%
\providecommand \bibinfo  [0]{\@secondoftwo}%
\providecommand \bibfield  [0]{\@secondoftwo}%
\providecommand \translation [1]{[#1]}%
\providecommand \BibitemOpen [0]{}%
\providecommand \bibitemStop [0]{}%
\providecommand \bibitemNoStop [0]{.\EOS\space}%
\providecommand \EOS [0]{\spacefactor3000\relax}%
\providecommand \BibitemShut  [1]{\csname bibitem#1\endcsname}%
\let\auto@bib@innerbib\@empty
\bibitem [{\citenamefont {Jaynes}(1957{\natexlab{a}})}]{jaynesInformationTheoryStatistical1957a}%
  \BibitemOpen
  \bibfield  {author} {\bibinfo {author} {\bibfnamefont {E.~T.}\ \bibnamefont {Jaynes}},\ }\bibfield  {title} {\bibinfo {title} {Information theory and statistical mechanics},\ }\href {https://doi.org/10.1103/PhysRev.106.620} {\bibfield  {journal} {\bibinfo  {journal} {Physical Review}\ }\textbf {\bibinfo {volume} {106}},\ \bibinfo {pages} {620} (\bibinfo {year} {1957}{\natexlab{a}})}\BibitemShut {NoStop}%
\bibitem [{\citenamefont {Jaynes}(1957{\natexlab{b}})}]{jaynesInformationTheoryStatistical1957}%
  \BibitemOpen
  \bibfield  {author} {\bibinfo {author} {\bibfnamefont {E.~T.}\ \bibnamefont {Jaynes}},\ }\bibfield  {title} {\bibinfo {title} {Information theory and statistical mechanics. {{II}}},\ }\href {https://doi.org/10.1103/PhysRev.108.171} {\bibfield  {journal} {\bibinfo  {journal} {Physical Review}\ }\textbf {\bibinfo {volume} {108}},\ \bibinfo {pages} {171} (\bibinfo {year} {1957}{\natexlab{b}})}\BibitemShut {NoStop}%
\bibitem [{\citenamefont {Jaynes}(2003)}]{jaynesProbabilityTheoryLogic2003}%
  \BibitemOpen
  \bibfield  {author} {\bibinfo {author} {\bibfnamefont {E.~T.}\ \bibnamefont {Jaynes}},\ }\href {https://doi.org/10.1017/CBO9780511790423} {\emph {\bibinfo {title} {Probability {{Theory}}: {{The Logic}} of {{Science}}}}},\ edited by\ \bibinfo {editor} {\bibfnamefont {G.~L.}\ \bibnamefont {Bretthorst}}\ (\bibinfo  {publisher} {Cambridge University Press},\ \bibinfo {address} {Cambridge},\ \bibinfo {year} {2003})\BibitemShut {NoStop}%
\bibitem [{\citenamefont {Landauer}(1991)}]{landauerInformationPhysical1991}%
  \BibitemOpen
  \bibfield  {author} {\bibinfo {author} {\bibfnamefont {R.}~\bibnamefont {Landauer}},\ }\bibfield  {title} {\bibinfo {title} {Information is {{Physical}}},\ }\href {https://doi.org/10.1063/1.881299} {\bibfield  {journal} {\bibinfo  {journal} {Physics Today}\ }\textbf {\bibinfo {volume} {44}},\ \bibinfo {pages} {23} (\bibinfo {year} {1991})}\BibitemShut {NoStop}%
\bibitem [{\citenamefont {Landauer}(1961)}]{landauerIrreversibilityHeatGeneration1961}%
  \BibitemOpen
  \bibfield  {author} {\bibinfo {author} {\bibfnamefont {R.}~\bibnamefont {Landauer}},\ }\bibfield  {title} {\bibinfo {title} {Irreversibility and {{Heat Generation}} in the {{Computing Process}}},\ }\href {https://doi.org/10.1147/rd.53.0183} {\bibfield  {journal} {\bibinfo  {journal} {IBM Journal of Research and Development}\ }\textbf {\bibinfo {volume} {5}},\ \bibinfo {pages} {183} (\bibinfo {year} {1961})}\BibitemShut {NoStop}%
\bibitem [{\citenamefont {Bennett}(1973)}]{bennettLogicalReversibilityComputation1973}%
  \BibitemOpen
  \bibfield  {author} {\bibinfo {author} {\bibfnamefont {C.~H.}\ \bibnamefont {Bennett}},\ }\bibfield  {title} {\bibinfo {title} {Logical {{Reversibility}} of {{Computation}}},\ }\href {https://doi.org/10.1147/rd.176.0525} {\bibfield  {journal} {\bibinfo  {journal} {IBM Journal of Research and Development}\ }\textbf {\bibinfo {volume} {17}},\ \bibinfo {pages} {525} (\bibinfo {year} {1973})}\BibitemShut {NoStop}%
\bibitem [{\citenamefont {Bennett}(1982)}]{bennettThermodynamicsComputationReview1982}%
  \BibitemOpen
  \bibfield  {author} {\bibinfo {author} {\bibfnamefont {C.~H.}\ \bibnamefont {Bennett}},\ }\bibfield  {title} {\bibinfo {title} {The thermodynamics of computation---a review},\ }\href {https://doi.org/10.1007/BF02084158} {\bibfield  {journal} {\bibinfo  {journal} {International Journal of Theoretical Physics}\ }\textbf {\bibinfo {volume} {21}},\ \bibinfo {pages} {905} (\bibinfo {year} {1982})}\BibitemShut {NoStop}%
\bibitem [{\citenamefont {Reeb}\ and\ \citenamefont {Wolf}(2014)}]{reebImprovedLandauerPrinciple2014}%
  \BibitemOpen
  \bibfield  {author} {\bibinfo {author} {\bibfnamefont {D.}~\bibnamefont {Reeb}}\ and\ \bibinfo {author} {\bibfnamefont {M.~M.}\ \bibnamefont {Wolf}},\ }\bibfield  {title} {\bibinfo {title} {An improved {{Landauer}} principle with finite-size corrections},\ }\href {https://doi.org/10.1088/1367-2630/16/10/103011} {\bibfield  {journal} {\bibinfo  {journal} {New Journal of Physics}\ }\textbf {\bibinfo {volume} {16}},\ \bibinfo {pages} {103011} (\bibinfo {year} {2014})}\BibitemShut {NoStop}%
\bibitem [{\citenamefont {Dahlsten}\ \emph {et~al.}(2011)\citenamefont {Dahlsten}, \citenamefont {Renner}, \citenamefont {Rieper},\ and\ \citenamefont {Vedral}}]{dahlstenInadequacyNeumannEntropy2011}%
  \BibitemOpen
  \bibfield  {author} {\bibinfo {author} {\bibfnamefont {O.~C.~O.}\ \bibnamefont {Dahlsten}}, \bibinfo {author} {\bibfnamefont {R.}~\bibnamefont {Renner}}, \bibinfo {author} {\bibfnamefont {E.}~\bibnamefont {Rieper}},\ and\ \bibinfo {author} {\bibfnamefont {V.}~\bibnamefont {Vedral}},\ }\bibfield  {title} {\bibinfo {title} {Inadequacy of von {{Neumann}} entropy for characterizing extractable work},\ }\href {https://doi.org/10.1088/1367-2630/13/5/053015} {\bibfield  {journal} {\bibinfo  {journal} {New Journal of Physics}\ }\textbf {\bibinfo {volume} {13}},\ \bibinfo {pages} {053015} (\bibinfo {year} {2011})}\BibitemShut {NoStop}%
\bibitem [{\citenamefont {Feynman}(2023)}]{feynmanFeynmanLecturesComputation2023}%
  \BibitemOpen
  \bibfield  {author} {\bibinfo {author} {\bibfnamefont {R.~P.}\ \bibnamefont {Feynman}},\ }\href {https://doi.org/10.1201/9781003358817} {\emph {\bibinfo {title} {Feynman {{Lectures}} on {{Computation}}: {{Anniversary Edition}}}}},\ \bibinfo {edition} {2nd}\ ed.,\ edited by\ \bibinfo {editor} {\bibfnamefont {T.}~\bibnamefont {Hey}}\ (\bibinfo  {publisher} {CRC Press},\ \bibinfo {address} {Boca Raton},\ \bibinfo {year} {2023})\BibitemShut {NoStop}%
\bibitem [{\citenamefont {del Rio}\ \emph {et~al.}(2011)\citenamefont {del Rio}, \citenamefont {Aberg}, \citenamefont {Renner}, \citenamefont {Dahlsten},\ and\ \citenamefont {Vedral}}]{rioThermodynamicMeaningNegative2011}%
  \BibitemOpen
  \bibfield  {author} {\bibinfo {author} {\bibfnamefont {L.}~\bibnamefont {del Rio}}, \bibinfo {author} {\bibfnamefont {J.}~\bibnamefont {Aberg}}, \bibinfo {author} {\bibfnamefont {R.}~\bibnamefont {Renner}}, \bibinfo {author} {\bibfnamefont {O.}~\bibnamefont {Dahlsten}},\ and\ \bibinfo {author} {\bibfnamefont {V.}~\bibnamefont {Vedral}},\ }\bibfield  {title} {\bibinfo {title} {The thermodynamic meaning of negative entropy},\ }\href {https://doi.org/10.1038/nature10123} {\bibfield  {journal} {\bibinfo  {journal} {Nature}\ }\textbf {\bibinfo {volume} {474}},\ \bibinfo {pages} {61} (\bibinfo {year} {2011})},\ \Eprint {https://arxiv.org/abs/1009.1630} {arXiv:1009.1630 [quant-ph]} \BibitemShut {NoStop}%
\bibitem [{Note1()}]{Note1}%
  \BibitemOpen
  \bibinfo {note} {By learning, we mean taking multiple copies of $\ket {\psi _x}$ as input and outputting a circuit description of $U_x$, the unitary that prepares $\ket {\psi _x}$. The state unpreparation unitary $U_x^\dagger $ can be easily derived from the circuit description of $U_x$. This is very different from constructing $U_x^\dagger $ from quantum queries to $U_x$, which is generally hard \cite {yoshida2023reversing}. See \protect \Cref {sec:learning} for a formal definition.}\BibitemShut {Stop}%
\bibitem [{\citenamefont {Zhao}\ \emph {et~al.}(2024)\citenamefont {Zhao}, \citenamefont {Lewis}, \citenamefont {Kannan}, \citenamefont {Quek}, \citenamefont {Huang},\ and\ \citenamefont {Caro}}]{zhaoLearningQuantumStates2024}%
  \BibitemOpen
  \bibfield  {author} {\bibinfo {author} {\bibfnamefont {H.}~\bibnamefont {Zhao}}, \bibinfo {author} {\bibfnamefont {L.}~\bibnamefont {Lewis}}, \bibinfo {author} {\bibfnamefont {I.}~\bibnamefont {Kannan}}, \bibinfo {author} {\bibfnamefont {Y.}~\bibnamefont {Quek}}, \bibinfo {author} {\bibfnamefont {H.-Y.}\ \bibnamefont {Huang}},\ and\ \bibinfo {author} {\bibfnamefont {M.~C.}\ \bibnamefont {Caro}},\ }\bibfield  {title} {\bibinfo {title} {Learning quantum states and unitaries of bounded gate complexity},\ }\href {https://doi.org/10.1103/PRXQuantum.5.040306} {\bibfield  {journal} {\bibinfo  {journal} {PRX Quantum}\ }\textbf {\bibinfo {volume} {5}},\ \bibinfo {pages} {40306} (\bibinfo {year} {2024})}\BibitemShut {NoStop}%
\bibitem [{\citenamefont {Huang}\ \emph {et~al.}(2024)\citenamefont {Huang}, \citenamefont {Liu}, \citenamefont {Broughton}, \citenamefont {Kim}, \citenamefont {Anshu}, \citenamefont {Landau},\ and\ \citenamefont {McClean}}]{huangLearningShallowQuantum2024}%
  \BibitemOpen
  \bibfield  {author} {\bibinfo {author} {\bibfnamefont {H.-Y.}\ \bibnamefont {Huang}}, \bibinfo {author} {\bibfnamefont {Y.}~\bibnamefont {Liu}}, \bibinfo {author} {\bibfnamefont {M.}~\bibnamefont {Broughton}}, \bibinfo {author} {\bibfnamefont {I.}~\bibnamefont {Kim}}, \bibinfo {author} {\bibfnamefont {A.}~\bibnamefont {Anshu}}, \bibinfo {author} {\bibfnamefont {Z.}~\bibnamefont {Landau}},\ and\ \bibinfo {author} {\bibfnamefont {J.~R.}\ \bibnamefont {McClean}},\ }\bibfield  {title} {\bibinfo {title} {Learning shallow quantum circuits},\ }in\ \href {https://doi.org/10.1145/3618260.3649722} {\emph {\bibinfo {booktitle} {Proceedings of the 56th {{Annual ACM Symposium}} on {{Theory}} of {{Computing}}}}}\ (\bibinfo {year} {2024})\ pp.\ \bibinfo {pages} {1343--1351},\ \Eprint {https://arxiv.org/abs/2401.10095} {arXiv:2401.10095 [quant-ph]} \BibitemShut {NoStop}%
\bibitem [{\citenamefont {Landau}\ and\ \citenamefont {Liu}(2024)}]{landauLearningQuantumStates2024}%
  \BibitemOpen
  \bibfield  {author} {\bibinfo {author} {\bibfnamefont {Z.}~\bibnamefont {Landau}}\ and\ \bibinfo {author} {\bibfnamefont {Y.}~\bibnamefont {Liu}},\ }\href {https://doi.org/10.48550/arXiv.2410.23618} {\bibinfo {title} {Learning quantum states prepared by shallow circuits in polynomial time}} (\bibinfo {year} {2024}),\ \Eprint {https://arxiv.org/abs/2410.23618} {arXiv:2410.23618 [quant-ph]} \BibitemShut {NoStop}%
\bibitem [{\citenamefont {Leone}\ \emph {et~al.}(2024)\citenamefont {Leone}, \citenamefont {Oliviero},\ and\ \citenamefont {Hamma}}]{leoneLearningTdopedStabilizer2024}%
  \BibitemOpen
  \bibfield  {author} {\bibinfo {author} {\bibfnamefont {L.}~\bibnamefont {Leone}}, \bibinfo {author} {\bibfnamefont {S.~F.~E.}\ \bibnamefont {Oliviero}},\ and\ \bibinfo {author} {\bibfnamefont {A.}~\bibnamefont {Hamma}},\ }\bibfield  {title} {\bibinfo {title} {Learning t-doped stabilizer states},\ }\href {https://doi.org/10.22331/q-2024-05-27-1361} {\bibfield  {journal} {\bibinfo  {journal} {Quantum}\ }\textbf {\bibinfo {volume} {8}},\ \bibinfo {pages} {1361} (\bibinfo {year} {2024})},\ \Eprint {https://arxiv.org/abs/2305.15398} {arXiv:2305.15398 [quant-ph]} \BibitemShut {NoStop}%
\bibitem [{\citenamefont {Grewal}\ \emph {et~al.}(2024)\citenamefont {Grewal}, \citenamefont {Iyer}, \citenamefont {Kretschmer},\ and\ \citenamefont {Liang}}]{grewalEfficientLearningQuantum2024}%
  \BibitemOpen
  \bibfield  {author} {\bibinfo {author} {\bibfnamefont {S.}~\bibnamefont {Grewal}}, \bibinfo {author} {\bibfnamefont {V.}~\bibnamefont {Iyer}}, \bibinfo {author} {\bibfnamefont {W.}~\bibnamefont {Kretschmer}},\ and\ \bibinfo {author} {\bibfnamefont {D.}~\bibnamefont {Liang}},\ }\href {http://arxiv.org/abs/2305.13409} {\bibinfo {title} {Efficient learning of quantum states prepared with few non-clifford gates}} (\bibinfo {year} {2024}),\ \Eprint {https://arxiv.org/abs/2305.13409} {arXiv:2305.13409 [quant-ph]} \BibitemShut {NoStop}%
\bibitem [{\citenamefont {Cramer}\ \emph {et~al.}(2010)\citenamefont {Cramer}, \citenamefont {Plenio}, \citenamefont {Flammia}, \citenamefont {Somma}, \citenamefont {Gross}, \citenamefont {Bartlett}, \citenamefont {{Landon-Cardinal}}, \citenamefont {Poulin},\ and\ \citenamefont {Liu}}]{cramerEfficientQuantumState2010}%
  \BibitemOpen
  \bibfield  {author} {\bibinfo {author} {\bibfnamefont {M.}~\bibnamefont {Cramer}}, \bibinfo {author} {\bibfnamefont {M.~B.}\ \bibnamefont {Plenio}}, \bibinfo {author} {\bibfnamefont {S.~T.}\ \bibnamefont {Flammia}}, \bibinfo {author} {\bibfnamefont {R.}~\bibnamefont {Somma}}, \bibinfo {author} {\bibfnamefont {D.}~\bibnamefont {Gross}}, \bibinfo {author} {\bibfnamefont {S.~D.}\ \bibnamefont {Bartlett}}, \bibinfo {author} {\bibfnamefont {O.}~\bibnamefont {{Landon-Cardinal}}}, \bibinfo {author} {\bibfnamefont {D.}~\bibnamefont {Poulin}},\ and\ \bibinfo {author} {\bibfnamefont {Y.-K.}\ \bibnamefont {Liu}},\ }\bibfield  {title} {\bibinfo {title} {Efficient quantum state tomography},\ }\href {https://doi.org/10.1038/ncomms1147} {\bibfield  {journal} {\bibinfo  {journal} {Nature Communications}\ }\textbf {\bibinfo {volume} {1}},\ \bibinfo {pages} {149} (\bibinfo {year} {2010})}\BibitemShut {NoStop}%
\bibitem [{\citenamefont {Arunachalam}\ \emph {et~al.}(2023)\citenamefont {Arunachalam}, \citenamefont {Bravyi}, \citenamefont {Dutt},\ and\ \citenamefont {Yoder}}]{arunachalamOptimalAlgorithmsLearning2023}%
  \BibitemOpen
  \bibfield  {author} {\bibinfo {author} {\bibfnamefont {S.}~\bibnamefont {Arunachalam}}, \bibinfo {author} {\bibfnamefont {S.}~\bibnamefont {Bravyi}}, \bibinfo {author} {\bibfnamefont {A.}~\bibnamefont {Dutt}},\ and\ \bibinfo {author} {\bibfnamefont {T.~J.}\ \bibnamefont {Yoder}},\ }\href {https://doi.org/10.48550/arXiv.2208.07851} {\bibinfo {title} {Optimal algorithms for learning quantum phase states}} (\bibinfo {year} {2023}),\ \Eprint {https://arxiv.org/abs/2208.07851} {arXiv:2208.07851 [quant-ph]} \BibitemShut {NoStop}%
\bibitem [{\citenamefont {Schuster}\ \emph {et~al.}(2025)\citenamefont {Schuster}, \citenamefont {Haferkamp},\ and\ \citenamefont {Huang}}]{schuster2025random}%
  \BibitemOpen
  \bibfield  {author} {\bibinfo {author} {\bibfnamefont {T.}~\bibnamefont {Schuster}}, \bibinfo {author} {\bibfnamefont {J.}~\bibnamefont {Haferkamp}},\ and\ \bibinfo {author} {\bibfnamefont {H.-Y.}\ \bibnamefont {Huang}},\ }\bibfield  {title} {\bibinfo {title} {Random unitaries in extremely low depth},\ }\href@noop {} {\bibfield  {journal} {\bibinfo  {journal} {Science}\ }\textbf {\bibinfo {volume} {389}},\ \bibinfo {pages} {92} (\bibinfo {year} {2025})}\BibitemShut {NoStop}%
\bibitem [{\citenamefont {Anshu}\ and\ \citenamefont {Arunachalam}(2024)}]{anshuSurveyComplexityLearning2024}%
  \BibitemOpen
  \bibfield  {author} {\bibinfo {author} {\bibfnamefont {A.}~\bibnamefont {Anshu}}\ and\ \bibinfo {author} {\bibfnamefont {S.}~\bibnamefont {Arunachalam}},\ }\bibfield  {title} {\bibinfo {title} {A survey on the complexity of learning quantum states},\ }\href {https://doi.org/10.1038/s42254-023-00662-4} {\bibfield  {journal} {\bibinfo  {journal} {Nature Reviews Physics}\ }\textbf {\bibinfo {volume} {6}},\ \bibinfo {pages} {59} (\bibinfo {year} {2024})}\BibitemShut {NoStop}%
\bibitem [{\citenamefont {Wilming}\ and\ \citenamefont {Gallego}(2017)}]{wilmingThirdLawThermodynamics2017}%
  \BibitemOpen
  \bibfield  {author} {\bibinfo {author} {\bibfnamefont {H.}~\bibnamefont {Wilming}}\ and\ \bibinfo {author} {\bibfnamefont {R.}~\bibnamefont {Gallego}},\ }\bibfield  {title} {\bibinfo {title} {Third {{Law}} of {{Thermodynamics}} as a {{Single Inequality}}},\ }\href {https://doi.org/10.1103/PhysRevX.7.041033} {\bibfield  {journal} {\bibinfo  {journal} {Physical Review X}\ }\textbf {\bibinfo {volume} {7}},\ \bibinfo {pages} {041033} (\bibinfo {year} {2017})}\BibitemShut {NoStop}%
\bibitem [{\citenamefont {Masanes}\ and\ \citenamefont {Oppenheim}(2017)}]{masanesGeneralDerivationQuantification2017}%
  \BibitemOpen
  \bibfield  {author} {\bibinfo {author} {\bibfnamefont {L.}~\bibnamefont {Masanes}}\ and\ \bibinfo {author} {\bibfnamefont {J.}~\bibnamefont {Oppenheim}},\ }\bibfield  {title} {\bibinfo {title} {A general derivation and quantification of the third law of thermodynamics},\ }\href {https://doi.org/10.1038/ncomms14538} {\bibfield  {journal} {\bibinfo  {journal} {Nature Communications}\ }\textbf {\bibinfo {volume} {8}},\ \bibinfo {pages} {14538} (\bibinfo {year} {2017})}\BibitemShut {NoStop}%
\bibitem [{\citenamefont {Scharlau}\ and\ \citenamefont {Mueller}(2018)}]{scharlauQuantumHornsLemma2018}%
  \BibitemOpen
  \bibfield  {author} {\bibinfo {author} {\bibfnamefont {J.}~\bibnamefont {Scharlau}}\ and\ \bibinfo {author} {\bibfnamefont {M.~P.}\ \bibnamefont {Mueller}},\ }\bibfield  {title} {\bibinfo {title} {Quantum {{Horn}}'s lemma, finite heat baths, and the third law of thermodynamics},\ }\href {https://doi.org/10.22331/q-2018-02-22-54} {\bibfield  {journal} {\bibinfo  {journal} {Quantum}\ }\textbf {\bibinfo {volume} {2}},\ \bibinfo {pages} {54} (\bibinfo {year} {2018})}\BibitemShut {NoStop}%
\bibitem [{\citenamefont {Taranto}\ \emph {et~al.}(2023)\citenamefont {Taranto}, \citenamefont {Bakhshinezhad}, \citenamefont {Bluhm}, \citenamefont {Silva}, \citenamefont {Friis}, \citenamefont {Lock}, \citenamefont {Vitagliano}, \citenamefont {Binder}, \citenamefont {Debarba}, \citenamefont {Schwarzhans}, \citenamefont {Clivaz},\ and\ \citenamefont {Huber}}]{tarantoLandauerNernstWhat2023}%
  \BibitemOpen
  \bibfield  {author} {\bibinfo {author} {\bibfnamefont {P.}~\bibnamefont {Taranto}}, \bibinfo {author} {\bibfnamefont {F.}~\bibnamefont {Bakhshinezhad}}, \bibinfo {author} {\bibfnamefont {A.}~\bibnamefont {Bluhm}}, \bibinfo {author} {\bibfnamefont {R.}~\bibnamefont {Silva}}, \bibinfo {author} {\bibfnamefont {N.}~\bibnamefont {Friis}}, \bibinfo {author} {\bibfnamefont {M.~P.}\ \bibnamefont {Lock}}, \bibinfo {author} {\bibfnamefont {G.}~\bibnamefont {Vitagliano}}, \bibinfo {author} {\bibfnamefont {F.~C.}\ \bibnamefont {Binder}}, \bibinfo {author} {\bibfnamefont {T.}~\bibnamefont {Debarba}}, \bibinfo {author} {\bibfnamefont {E.}~\bibnamefont {Schwarzhans}}, \bibinfo {author} {\bibfnamefont {F.}~\bibnamefont {Clivaz}},\ and\ \bibinfo {author} {\bibfnamefont {M.}~\bibnamefont {Huber}},\ }\bibfield  {title} {\bibinfo {title} {Landauer {{Versus Nernst}}: {{What}} is the {{True Cost}} of {{Cooling}} a {{Quantum System}}?},\ }\href {https://doi.org/10.1103/PRXQuantum.4.010332} {\bibfield  {journal} {\bibinfo
  {journal} {PRX Quantum}\ }\textbf {\bibinfo {volume} {4}},\ \bibinfo {pages} {010332} (\bibinfo {year} {2023})}\BibitemShut {NoStop}%
\bibitem [{\citenamefont {Plesch}\ and\ \citenamefont {Bu{\v z}ek}(2010)}]{pleschEfficientCompressionQuantum2010}%
  \BibitemOpen
  \bibfield  {author} {\bibinfo {author} {\bibfnamefont {M.}~\bibnamefont {Plesch}}\ and\ \bibinfo {author} {\bibfnamefont {V.}~\bibnamefont {Bu{\v z}ek}},\ }\bibfield  {title} {\bibinfo {title} {Efficient compression of quantum information},\ }\href {https://doi.org/10.1103/PhysRevA.81.032317} {\bibfield  {journal} {\bibinfo  {journal} {Physical Review A}\ }\textbf {\bibinfo {volume} {81}},\ \bibinfo {pages} {032317} (\bibinfo {year} {2010})}\BibitemShut {NoStop}%
\bibitem [{\citenamefont {Yang}\ \emph {et~al.}(2016)\citenamefont {Yang}, \citenamefont {Chiribella},\ and\ \citenamefont {Ebler}}]{yangEfficientQuantumCompression2016}%
  \BibitemOpen
  \bibfield  {author} {\bibinfo {author} {\bibfnamefont {Y.}~\bibnamefont {Yang}}, \bibinfo {author} {\bibfnamefont {G.}~\bibnamefont {Chiribella}},\ and\ \bibinfo {author} {\bibfnamefont {D.}~\bibnamefont {Ebler}},\ }\bibfield  {title} {\bibinfo {title} {Efficient {{Quantum Compression}} for {{Ensembles}} of {{Identically Prepared Mixed States}}},\ }\href {https://doi.org/10.1103/PhysRevLett.116.080501} {\bibfield  {journal} {\bibinfo  {journal} {Physical Review Letters}\ }\textbf {\bibinfo {volume} {116}},\ \bibinfo {pages} {080501} (\bibinfo {year} {2016})}\BibitemShut {NoStop}%
\bibitem [{\citenamefont {Hayashi}(2017)}]{hayashiQuantumInformationTheory2017}%
  \BibitemOpen
  \bibfield  {author} {\bibinfo {author} {\bibfnamefont {M.}~\bibnamefont {Hayashi}},\ }\href {https://doi.org/10.1007/978-3-662-49725-8} {\emph {\bibinfo {title} {Quantum {{Information Theory}}: {{Mathematical Foundation}}}}},\ Graduate {{Texts}} in {{Physics}}\ (\bibinfo  {publisher} {Springer},\ \bibinfo {address} {Berlin, Heidelberg},\ \bibinfo {year} {2017})\BibitemShut {NoStop}%
\bibitem [{\citenamefont {Bennett}\ \emph {et~al.}(2006)\citenamefont {Bennett}, \citenamefont {Harrow},\ and\ \citenamefont {Lloyd}}]{bennett2006universal}%
  \BibitemOpen
  \bibfield  {author} {\bibinfo {author} {\bibfnamefont {C.~H.}\ \bibnamefont {Bennett}}, \bibinfo {author} {\bibfnamefont {A.~W.}\ \bibnamefont {Harrow}},\ and\ \bibinfo {author} {\bibfnamefont {S.}~\bibnamefont {Lloyd}},\ }\bibfield  {title} {\bibinfo {title} {Universal quantum data compression via nondestructive tomography},\ }\href@noop {} {\bibfield  {journal} {\bibinfo  {journal} {Physical Review A—Atomic, Molecular, and Optical Physics}\ }\textbf {\bibinfo {volume} {73}},\ \bibinfo {pages} {032336} (\bibinfo {year} {2006})}\BibitemShut {NoStop}%
\bibitem [{Note2()}]{Note2}%
  \BibitemOpen
  \bibinfo {note} {Note that the time efficiency claims in these previous works refer to a runtime of ${\protect \ensuremath {\protect \mathsf {poly}}}(N, D)$, where $D=2^n$ is the Hilbert space dimension of a single copy that grows exponentially with the system size.}\BibitemShut {Stop}%
\bibitem [{\citenamefont {Alicki}\ \emph {et~al.}(2004)\citenamefont {Alicki}, \citenamefont {Horodecki}, \citenamefont {Horodecki},\ and\ \citenamefont {Horodecki}}]{alickiThermodynamicsQuantumInformational2004}%
  \BibitemOpen
  \bibfield  {author} {\bibinfo {author} {\bibfnamefont {R.}~\bibnamefont {Alicki}}, \bibinfo {author} {\bibfnamefont {M.}~\bibnamefont {Horodecki}}, \bibinfo {author} {\bibfnamefont {P.}~\bibnamefont {Horodecki}},\ and\ \bibinfo {author} {\bibfnamefont {R.}~\bibnamefont {Horodecki}},\ }\href {https://doi.org/10.48550/arXiv.quant-ph/0402012} {\bibinfo {title} {Thermodynamics of quantum informational systems - {{Hamiltonian}} description}} (\bibinfo {year} {2004}),\ \Eprint {https://arxiv.org/abs/quant-ph/0402012} {arXiv:quant-ph/0402012} \BibitemShut {NoStop}%
\bibitem [{\citenamefont {Faist}\ \emph {et~al.}(2015)\citenamefont {Faist}, \citenamefont {Dupuis}, \citenamefont {Oppenheim},\ and\ \citenamefont {Renner}}]{faistMinimalWorkCost2015}%
  \BibitemOpen
  \bibfield  {author} {\bibinfo {author} {\bibfnamefont {P.}~\bibnamefont {Faist}}, \bibinfo {author} {\bibfnamefont {F.}~\bibnamefont {Dupuis}}, \bibinfo {author} {\bibfnamefont {J.}~\bibnamefont {Oppenheim}},\ and\ \bibinfo {author} {\bibfnamefont {R.}~\bibnamefont {Renner}},\ }\bibfield  {title} {\bibinfo {title} {The minimal work cost of information processing},\ }\href {https://doi.org/10.1038/ncomms8669} {\bibfield  {journal} {\bibinfo  {journal} {Nature Communications}\ }\textbf {\bibinfo {volume} {6}},\ \bibinfo {pages} {7669} (\bibinfo {year} {2015})}\BibitemShut {NoStop}%
\bibitem [{Note3()}]{Note3}%
  \BibitemOpen
  \bibinfo {note} {In this work, we use the standard notations for asymptotics. For two positive functions $f(n)$ and $g(n)$, $f(n) = O(g(n))$ if $\exists n_0, C>0$ such that $\forall n>n_0, f(n)\leq Cg(n)$. $f(n)=\Omega (g(n))$ if $g(n)=O(f(n))$. $f(n)=\Theta (g(n))$ if $f(n)=O(g(n))$ and $f(n)=\Omega (g(n))$.}\BibitemShut {Stop}%
\bibitem [{\citenamefont {B{\u a}descu}\ and\ \citenamefont {O'Donnell}(2024)}]{badescuImprovedQuantumData2024}%
  \BibitemOpen
  \bibfield  {author} {\bibinfo {author} {\bibfnamefont {C.}~\bibnamefont {B{\u a}descu}}\ and\ \bibinfo {author} {\bibfnamefont {R.}~\bibnamefont {O'Donnell}},\ }\bibfield  {title} {\bibinfo {title} {Improved quantum data analysis},\ }\href {https://doi.org/10.46298/theoretics.24.7} {\bibfield  {journal} {\bibinfo  {journal} {TheoretiCS}\ }\textbf {\bibinfo {volume} {Volume 3}},\ \bibinfo {pages} {10924} (\bibinfo {year} {2024})}\BibitemShut {NoStop}%
\bibitem [{\citenamefont {Huang}\ \emph {et~al.}(2020)\citenamefont {Huang}, \citenamefont {Kueng},\ and\ \citenamefont {Preskill}}]{huangPredictingManyProperties2020}%
  \BibitemOpen
  \bibfield  {author} {\bibinfo {author} {\bibfnamefont {H.-Y.}\ \bibnamefont {Huang}}, \bibinfo {author} {\bibfnamefont {R.}~\bibnamefont {Kueng}},\ and\ \bibinfo {author} {\bibfnamefont {J.}~\bibnamefont {Preskill}},\ }\bibfield  {title} {\bibinfo {title} {Predicting many properties of a quantum system from very few measurements},\ }\href {https://doi.org/10.1038/s41567-020-0932-7} {\bibfield  {journal} {\bibinfo  {journal} {Nature Physics}\ }\textbf {\bibinfo {volume} {16}},\ \bibinfo {pages} {1050} (\bibinfo {year} {2020})}\BibitemShut {NoStop}%
\bibitem [{\citenamefont {Helstrom}(1969)}]{helstromQuantumDetectionEstimation1969}%
  \BibitemOpen
  \bibfield  {author} {\bibinfo {author} {\bibfnamefont {C.~W.}\ \bibnamefont {Helstrom}},\ }\bibfield  {title} {\bibinfo {title} {Quantum detection and estimation theory},\ }\href {https://doi.org/10.1007/BF01007479} {\bibfield  {journal} {\bibinfo  {journal} {Journal of Statistical Physics}\ }\textbf {\bibinfo {volume} {1}},\ \bibinfo {pages} {231} (\bibinfo {year} {1969})}\BibitemShut {NoStop}%
\bibitem [{\citenamefont {Fredkin}\ and\ \citenamefont {Toffoli}(1982)}]{fredkinConservativeLogic1982}%
  \BibitemOpen
  \bibfield  {author} {\bibinfo {author} {\bibfnamefont {E.}~\bibnamefont {Fredkin}}\ and\ \bibinfo {author} {\bibfnamefont {T.}~\bibnamefont {Toffoli}},\ }\bibfield  {title} {\bibinfo {title} {Conservative logic},\ }\href {https://doi.org/10.1007/BF01857727} {\bibfield  {journal} {\bibinfo  {journal} {International Journal of Theoretical Physics}\ }\textbf {\bibinfo {volume} {21}},\ \bibinfo {pages} {219} (\bibinfo {year} {1982})}\BibitemShut {NoStop}%
\bibitem [{\citenamefont {Nielsen}\ and\ \citenamefont {Chuang}(2010)}]{nielsenQuantumComputationQuantum2010}%
  \BibitemOpen
  \bibfield  {author} {\bibinfo {author} {\bibfnamefont {M.~A.}\ \bibnamefont {Nielsen}}\ and\ \bibinfo {author} {\bibfnamefont {I.~L.}\ \bibnamefont {Chuang}},\ }\href {https://books.google.com/books?hl=zh-CN&lr=&id=-s4DEy7o-a0C&oi=fnd&pg=PR17&dq=Nielsen+chuang&ots=NJ5FdmvxWo&sig=mALXBRaW2lvAUhu70yB4hzH7ePQ} {\emph {\bibinfo {title} {Quantum Computation and Quantum Information}}}\ (\bibinfo  {publisher} {Cambridge university press},\ \bibinfo {year} {2010})\BibitemShut {NoStop}%
\bibitem [{\citenamefont {{\v{S}}afr{\'a}nek}\ \emph {et~al.}(2023)\citenamefont {{\v{S}}afr{\'a}nek}, \citenamefont {Rosa},\ and\ \citenamefont {Binder}}]{vsafranek2023work}%
  \BibitemOpen
  \bibfield  {author} {\bibinfo {author} {\bibfnamefont {D.}~\bibnamefont {{\v{S}}afr{\'a}nek}}, \bibinfo {author} {\bibfnamefont {D.}~\bibnamefont {Rosa}},\ and\ \bibinfo {author} {\bibfnamefont {F.~C.}\ \bibnamefont {Binder}},\ }\bibfield  {title} {\bibinfo {title} {Work extraction from unknown quantum sources},\ }\href@noop {} {\bibfield  {journal} {\bibinfo  {journal} {Physical Review Letters}\ }\textbf {\bibinfo {volume} {130}},\ \bibinfo {pages} {210401} (\bibinfo {year} {2023})}\BibitemShut {NoStop}%
\bibitem [{\citenamefont {Xuereb}\ \emph {et~al.}(2024)\citenamefont {Xuereb}, \citenamefont {Junior}, \citenamefont {Clivaz}, \citenamefont {Bakhshinezhad},\ and\ \citenamefont {Huber}}]{xuereb2024resources}%
  \BibitemOpen
  \bibfield  {author} {\bibinfo {author} {\bibfnamefont {J.}~\bibnamefont {Xuereb}}, \bibinfo {author} {\bibfnamefont {A.}~\bibnamefont {Junior}}, \bibinfo {author} {\bibfnamefont {F.}~\bibnamefont {Clivaz}}, \bibinfo {author} {\bibfnamefont {P.}~\bibnamefont {Bakhshinezhad}},\ and\ \bibinfo {author} {\bibfnamefont {M.}~\bibnamefont {Huber}},\ }\bibfield  {title} {\bibinfo {title} {What resources do agents need to acquire knowledge in quantum thermodynamics?},\ }\href@noop {} {\bibfield  {journal} {\bibinfo  {journal} {arXiv preprint arXiv:2410.18167}\ } (\bibinfo {year} {2024})}\BibitemShut {NoStop}%
\bibitem [{\citenamefont {Renner}\ and\ \citenamefont {Wolf}(2004)}]{rennerSmoothRenyiEntropy2004}%
  \BibitemOpen
  \bibfield  {author} {\bibinfo {author} {\bibfnamefont {R.}~\bibnamefont {Renner}}\ and\ \bibinfo {author} {\bibfnamefont {S.}~\bibnamefont {Wolf}},\ }\bibfield  {title} {\bibinfo {title} {Smooth {{Renyi}} entropy and applications},\ }in\ \href {https://doi.org/10.1109/ISIT.2004.1365269} {\emph {\bibinfo {booktitle} {International {{Symposium onInformation Theory}}, 2004. {{ISIT}} 2004. {{Proceedings}}.}}}\ (\bibinfo {year} {2004})\ pp.\ \bibinfo {pages} {233--}\BibitemShut {NoStop}%
\bibitem [{\citenamefont {Konig}\ \emph {et~al.}(2009)\citenamefont {Konig}, \citenamefont {Renner},\ and\ \citenamefont {Schaffner}}]{konigOperationalMeaningMin2009}%
  \BibitemOpen
  \bibfield  {author} {\bibinfo {author} {\bibfnamefont {R.}~\bibnamefont {Konig}}, \bibinfo {author} {\bibfnamefont {R.}~\bibnamefont {Renner}},\ and\ \bibinfo {author} {\bibfnamefont {C.}~\bibnamefont {Schaffner}},\ }\bibfield  {title} {\bibinfo {title} {The {{Operational Meaning}} of {{Min-}} and {{Max-Entropy}}},\ }\href {https://doi.org/10.1109/TIT.2009.2025545} {\bibfield  {journal} {\bibinfo  {journal} {IEEE Transactions on Information Theory}\ }\textbf {\bibinfo {volume} {55}},\ \bibinfo {pages} {4337} (\bibinfo {year} {2009})}\BibitemShut {NoStop}%
\bibitem [{\citenamefont {Kim}\ \emph {et~al.}(2024)\citenamefont {Kim}, \citenamefont {Kim},\ and\ \citenamefont {Ranard}}]{kimLearningStatePreparation2024}%
  \BibitemOpen
  \bibfield  {author} {\bibinfo {author} {\bibfnamefont {H.-S.}\ \bibnamefont {Kim}}, \bibinfo {author} {\bibfnamefont {I.~H.}\ \bibnamefont {Kim}},\ and\ \bibinfo {author} {\bibfnamefont {D.}~\bibnamefont {Ranard}},\ }\href {https://doi.org/10.48550/arXiv.2410.23544} {\bibinfo {title} {Learning {{State Preparation Circuits}} for {{Quantum Phases}} of {{Matter}}}} (\bibinfo {year} {2024}),\ \Eprint {https://arxiv.org/abs/2410.23544} {arXiv:2410.23544} \BibitemShut {NoStop}%
\bibitem [{\citenamefont {Poulin}\ \emph {et~al.}(2011)\citenamefont {Poulin}, \citenamefont {Qarry}, \citenamefont {Somma},\ and\ \citenamefont {Verstraete}}]{poulinQuantumSimulationTimedependent2011}%
  \BibitemOpen
  \bibfield  {author} {\bibinfo {author} {\bibfnamefont {D.}~\bibnamefont {Poulin}}, \bibinfo {author} {\bibfnamefont {A.}~\bibnamefont {Qarry}}, \bibinfo {author} {\bibfnamefont {R.}~\bibnamefont {Somma}},\ and\ \bibinfo {author} {\bibfnamefont {F.}~\bibnamefont {Verstraete}},\ }\bibfield  {title} {\bibinfo {title} {Quantum simulation of time-dependent hamiltonians and the convenient illusion of hilbert space},\ }\href {https://doi.org/10.1103/PhysRevLett.106.170501} {\bibfield  {journal} {\bibinfo  {journal} {Physical Review Letters}\ }\textbf {\bibinfo {volume} {106}},\ \bibinfo {pages} {170501} (\bibinfo {year} {2011})}\BibitemShut {NoStop}%
\bibitem [{\citenamefont {Preskill}(2018)}]{preskillQuantumComputingNISQ2018}%
  \BibitemOpen
  \bibfield  {author} {\bibinfo {author} {\bibfnamefont {J.}~\bibnamefont {Preskill}},\ }\bibfield  {title} {\bibinfo {title} {Quantum computing in the {{NISQ}} era and beyond},\ }\href {https://doi.org/10.22331/q-2018-08-06-79} {\bibfield  {journal} {\bibinfo  {journal} {Quantum}\ }\textbf {\bibinfo {volume} {2}},\ \bibinfo {pages} {79} (\bibinfo {year} {2018})},\ \Eprint {https://arxiv.org/abs/1801.00862} {arXiv:1801.00862 [quant-ph]} \BibitemShut {NoStop}%
\bibitem [{\citenamefont {Terhal}\ and\ \citenamefont {DiVincenzo}(2004)}]{terhalAdaptiveQuantumComputation2004}%
  \BibitemOpen
  \bibfield  {author} {\bibinfo {author} {\bibfnamefont {B.~M.}\ \bibnamefont {Terhal}}\ and\ \bibinfo {author} {\bibfnamefont {D.~P.}\ \bibnamefont {DiVincenzo}},\ }\href {https://doi.org/10.48550/arXiv.quant-ph/0205133} {\bibinfo {title} {Adaptive {{Quantum Computation}}, {{Constant Depth Quantum Circuits}} and {{Arthur-Merlin Games}}}} (\bibinfo {year} {2004}),\ \Eprint {https://arxiv.org/abs/quant-ph/0205133} {arXiv:quant-ph/0205133} \BibitemShut {NoStop}%
\bibitem [{\citenamefont {Gao}\ \emph {et~al.}(2017)\citenamefont {Gao}, \citenamefont {Wang},\ and\ \citenamefont {Duan}}]{gaoQuantumSupremacySimulating2017}%
  \BibitemOpen
  \bibfield  {author} {\bibinfo {author} {\bibfnamefont {X.}~\bibnamefont {Gao}}, \bibinfo {author} {\bibfnamefont {S.-T.}\ \bibnamefont {Wang}},\ and\ \bibinfo {author} {\bibfnamefont {L.-M.}\ \bibnamefont {Duan}},\ }\bibfield  {title} {\bibinfo {title} {Quantum {{Supremacy}} for {{Simulating}} a {{Translation-Invariant Ising Spin Model}}},\ }\href {https://doi.org/10.1103/PhysRevLett.118.040502} {\bibfield  {journal} {\bibinfo  {journal} {Physical Review Letters}\ }\textbf {\bibinfo {volume} {118}},\ \bibinfo {pages} {040502} (\bibinfo {year} {2017})}\BibitemShut {NoStop}%
\bibitem [{\citenamefont {Giovannetti}\ \emph {et~al.}(2008)\citenamefont {Giovannetti}, \citenamefont {Lloyd},\ and\ \citenamefont {Maccone}}]{giovannettiArchitecturesQuantumRandom2008}%
  \BibitemOpen
  \bibfield  {author} {\bibinfo {author} {\bibfnamefont {V.}~\bibnamefont {Giovannetti}}, \bibinfo {author} {\bibfnamefont {S.}~\bibnamefont {Lloyd}},\ and\ \bibinfo {author} {\bibfnamefont {L.}~\bibnamefont {Maccone}},\ }\bibfield  {title} {\bibinfo {title} {Architectures for a quantum random access memory},\ }\href {https://doi.org/10.1103/PhysRevA.78.052310} {\bibfield  {journal} {\bibinfo  {journal} {Physical Review A}\ }\textbf {\bibinfo {volume} {78}},\ \bibinfo {pages} {52310} (\bibinfo {year} {2008})}\BibitemShut {NoStop}%
\bibitem [{\citenamefont {Haferkamp}\ \emph {et~al.}(2020)\citenamefont {Haferkamp}, \citenamefont {Hangleiter}, \citenamefont {Bouland}, \citenamefont {Fefferman}, \citenamefont {Eisert},\ and\ \citenamefont {{Bermejo-Vega}}}]{haferkampClosingGapsQuantum2020}%
  \BibitemOpen
  \bibfield  {author} {\bibinfo {author} {\bibfnamefont {J.}~\bibnamefont {Haferkamp}}, \bibinfo {author} {\bibfnamefont {D.}~\bibnamefont {Hangleiter}}, \bibinfo {author} {\bibfnamefont {A.}~\bibnamefont {Bouland}}, \bibinfo {author} {\bibfnamefont {B.}~\bibnamefont {Fefferman}}, \bibinfo {author} {\bibfnamefont {J.}~\bibnamefont {Eisert}},\ and\ \bibinfo {author} {\bibfnamefont {J.}~\bibnamefont {{Bermejo-Vega}}},\ }\bibfield  {title} {\bibinfo {title} {Closing {{Gaps}} of a {{Quantum Advantage}} with {{Short-Time Hamiltonian Dynamics}}},\ }\href {https://doi.org/10.1103/PhysRevLett.125.250501} {\bibfield  {journal} {\bibinfo  {journal} {Physical Review Letters}\ }\textbf {\bibinfo {volume} {125}},\ \bibinfo {pages} {250501} (\bibinfo {year} {2020})}\BibitemShut {NoStop}%
\bibitem [{\citenamefont {Hangleiter}\ and\ \citenamefont {Eisert}(2023)}]{hangleiterComputationalAdvantageQuantum2023}%
  \BibitemOpen
  \bibfield  {author} {\bibinfo {author} {\bibfnamefont {D.}~\bibnamefont {Hangleiter}}\ and\ \bibinfo {author} {\bibfnamefont {J.}~\bibnamefont {Eisert}},\ }\bibfield  {title} {\bibinfo {title} {Computational advantage of quantum random sampling},\ }\href {https://doi.org/10.1103/RevModPhys.95.035001} {\bibfield  {journal} {\bibinfo  {journal} {Reviews of Modern Physics}\ }\textbf {\bibinfo {volume} {95}},\ \bibinfo {pages} {035001} (\bibinfo {year} {2023})}\BibitemShut {NoStop}%
\bibitem [{\citenamefont {Bravyi}\ \emph {et~al.}(2018)\citenamefont {Bravyi}, \citenamefont {Gosset},\ and\ \citenamefont {K{\"o}nig}}]{bravyiQuantumAdvantageShallow2018}%
  \BibitemOpen
  \bibfield  {author} {\bibinfo {author} {\bibfnamefont {S.}~\bibnamefont {Bravyi}}, \bibinfo {author} {\bibfnamefont {D.}~\bibnamefont {Gosset}},\ and\ \bibinfo {author} {\bibfnamefont {R.}~\bibnamefont {K{\"o}nig}},\ }\bibfield  {title} {\bibinfo {title} {Quantum advantage with shallow circuits},\ }\href {https://doi.org/10.1126/science.aar3106} {\bibfield  {journal} {\bibinfo  {journal} {Science}\ }\textbf {\bibinfo {volume} {362}},\ \bibinfo {pages} {308} (\bibinfo {year} {2018})}\BibitemShut {NoStop}%
\bibitem [{\citenamefont {Liu}\ and\ \citenamefont {Winter}(2022)}]{liuManyBodyQuantumMagic2022}%
  \BibitemOpen
  \bibfield  {author} {\bibinfo {author} {\bibfnamefont {Z.-W.}\ \bibnamefont {Liu}}\ and\ \bibinfo {author} {\bibfnamefont {A.}~\bibnamefont {Winter}},\ }\bibfield  {title} {\bibinfo {title} {Many-{{Body Quantum Magic}}},\ }\href {https://doi.org/10.1103/PRXQuantum.3.020333} {\bibfield  {journal} {\bibinfo  {journal} {PRX Quantum}\ }\textbf {\bibinfo {volume} {3}},\ \bibinfo {pages} {020333} (\bibinfo {year} {2022})}\BibitemShut {NoStop}%
\bibitem [{\citenamefont {Shi}(2002)}]{shiBothToffoliControlledNOT2002}%
  \BibitemOpen
  \bibfield  {author} {\bibinfo {author} {\bibfnamefont {Y.}~\bibnamefont {Shi}},\ }\href {https://doi.org/10.48550/arXiv.quant-ph/0205115} {\bibinfo {title} {Both {{Toffoli}} and {{Controlled-NOT}} need little help to do universal quantum computation}} (\bibinfo {year} {2002}),\ \Eprint {https://arxiv.org/abs/quant-ph/0205115} {arXiv:quant-ph/0205115} \BibitemShut {NoStop}%
\bibitem [{\citenamefont {Bravyi}\ and\ \citenamefont {Kitaev}(2005)}]{bravyiUniversalQuantumComputation2005}%
  \BibitemOpen
  \bibfield  {author} {\bibinfo {author} {\bibfnamefont {S.}~\bibnamefont {Bravyi}}\ and\ \bibinfo {author} {\bibfnamefont {A.}~\bibnamefont {Kitaev}},\ }\bibfield  {title} {\bibinfo {title} {Universal quantum computation with ideal {{Clifford}} gates and noisy ancillas},\ }\href {https://doi.org/10.1103/PhysRevA.71.022316} {\bibfield  {journal} {\bibinfo  {journal} {Physical Review A}\ }\textbf {\bibinfo {volume} {71}},\ \bibinfo {pages} {022316} (\bibinfo {year} {2005})}\BibitemShut {NoStop}%
\bibitem [{\citenamefont {Aaronson}\ and\ \citenamefont {Gottesman}(2008)}]{aaronsonImprovedSimulationStabilizer2008}%
  \BibitemOpen
  \bibfield  {author} {\bibinfo {author} {\bibfnamefont {S.}~\bibnamefont {Aaronson}}\ and\ \bibinfo {author} {\bibfnamefont {D.}~\bibnamefont {Gottesman}},\ }\href {https://doi.org/10.48550/arXiv.quant-ph/0406196} {\bibinfo {title} {Improved {{Simulation}} of {{Stabilizer Circuits}}}} (\bibinfo {year} {2008}),\ \Eprint {https://arxiv.org/abs/quant-ph/0406196} {arXiv:quant-ph/0406196} \BibitemShut {NoStop}%
\bibitem [{\citenamefont {Bravyi}\ and\ \citenamefont {Gosset}(2016)}]{bravyiImprovedClassicalSimulation2016}%
  \BibitemOpen
  \bibfield  {author} {\bibinfo {author} {\bibfnamefont {S.}~\bibnamefont {Bravyi}}\ and\ \bibinfo {author} {\bibfnamefont {D.}~\bibnamefont {Gosset}},\ }\bibfield  {title} {\bibinfo {title} {Improved {{Classical Simulation}} of {{Quantum Circuits Dominated}} by {{Clifford Gates}}},\ }\href {https://doi.org/10.1103/PhysRevLett.116.250501} {\bibfield  {journal} {\bibinfo  {journal} {Physical Review Letters}\ }\textbf {\bibinfo {volume} {116}},\ \bibinfo {pages} {250501} (\bibinfo {year} {2016})}\BibitemShut {NoStop}%
\bibitem [{\citenamefont {Bravyi}\ \emph {et~al.}(2019)\citenamefont {Bravyi}, \citenamefont {Browne}, \citenamefont {Calpin}, \citenamefont {Campbell}, \citenamefont {Gosset},\ and\ \citenamefont {Howard}}]{bravyiSimulationQuantumCircuits2019}%
  \BibitemOpen
  \bibfield  {author} {\bibinfo {author} {\bibfnamefont {S.}~\bibnamefont {Bravyi}}, \bibinfo {author} {\bibfnamefont {D.}~\bibnamefont {Browne}}, \bibinfo {author} {\bibfnamefont {P.}~\bibnamefont {Calpin}}, \bibinfo {author} {\bibfnamefont {E.}~\bibnamefont {Campbell}}, \bibinfo {author} {\bibfnamefont {D.}~\bibnamefont {Gosset}},\ and\ \bibinfo {author} {\bibfnamefont {M.}~\bibnamefont {Howard}},\ }\bibfield  {title} {\bibinfo {title} {Simulation of quantum circuits by low-rank stabilizer decompositions},\ }\href {https://doi.org/10.22331/q-2019-09-02-181} {\bibfield  {journal} {\bibinfo  {journal} {Quantum}\ }\textbf {\bibinfo {volume} {3}},\ \bibinfo {pages} {181} (\bibinfo {year} {2019})},\ \Eprint {https://arxiv.org/abs/1808.00128} {arXiv:1808.00128 [quant-ph]} \BibitemShut {NoStop}%
\bibitem [{\citenamefont {Knill}(2004)}]{knillFaultTolerantPostselectedQuantum2004}%
  \BibitemOpen
  \bibfield  {author} {\bibinfo {author} {\bibfnamefont {E.}~\bibnamefont {Knill}},\ }\href {https://doi.org/10.48550/arXiv.quant-ph/0402171} {\bibinfo {title} {Fault-{{Tolerant Postselected Quantum Computation}}: {{Schemes}}}} (\bibinfo {year} {2004}),\ \Eprint {https://arxiv.org/abs/quant-ph/0402171} {arXiv:quant-ph/0402171} \BibitemShut {NoStop}%
\bibitem [{\citenamefont {Knill}\ \emph {et~al.}(2000)\citenamefont {Knill}, \citenamefont {Laflamme},\ and\ \citenamefont {Viola}}]{knillTheoryQuantumError2000}%
  \BibitemOpen
  \bibfield  {author} {\bibinfo {author} {\bibfnamefont {E.}~\bibnamefont {Knill}}, \bibinfo {author} {\bibfnamefont {R.}~\bibnamefont {Laflamme}},\ and\ \bibinfo {author} {\bibfnamefont {L.}~\bibnamefont {Viola}},\ }\bibfield  {title} {\bibinfo {title} {Theory of quantum error correction for general noise},\ }\href {https://doi.org/10.1103/PhysRevLett.84.2525} {\bibfield  {journal} {\bibinfo  {journal} {Physical Review Letters}\ }\textbf {\bibinfo {volume} {84}},\ \bibinfo {pages} {2525} (\bibinfo {year} {2000})}\BibitemShut {NoStop}%
\bibitem [{\citenamefont {Montanaro}(2017)}]{montanaroLearningStabilizerStates2017}%
  \BibitemOpen
  \bibfield  {author} {\bibinfo {author} {\bibfnamefont {A.}~\bibnamefont {Montanaro}},\ }\href {https://doi.org/10.48550/arXiv.1707.04012} {\bibinfo {title} {Learning stabilizer states by {{Bell}} sampling}} (\bibinfo {year} {2017}),\ \Eprint {https://arxiv.org/abs/1707.04012} {arXiv:1707.04012 [quant-ph]} \BibitemShut {NoStop}%
\bibitem [{\citenamefont {Hastings}(2018)}]{hastingsAreaLawOne2018}%
  \BibitemOpen
  \bibfield  {author} {\bibinfo {author} {\bibfnamefont {M.~B.}\ \bibnamefont {Hastings}},\ }\href {https://doi.org/10.48550/arXiv.0705.2024} {\bibinfo {title} {An {{Area Law}} for {{One Dimensional Quantum Systems}}}} (\bibinfo {year} {2018}),\ \Eprint {https://arxiv.org/abs/0705.2024} {arXiv:0705.2024} \BibitemShut {NoStop}%
\bibitem [{\citenamefont {Zeng}\ \emph {et~al.}(2019)\citenamefont {Zeng}, \citenamefont {Chen}, \citenamefont {Zhou},\ and\ \citenamefont {Wen}}]{zengQuantumInformationMeets2019}%
  \BibitemOpen
  \bibfield  {author} {\bibinfo {author} {\bibfnamefont {B.}~\bibnamefont {Zeng}}, \bibinfo {author} {\bibfnamefont {X.}~\bibnamefont {Chen}}, \bibinfo {author} {\bibfnamefont {D.-L.}\ \bibnamefont {Zhou}},\ and\ \bibinfo {author} {\bibfnamefont {X.-G.}\ \bibnamefont {Wen}},\ }\href {https://doi.org/10.1007/978-1-4939-9084-9} {\emph {\bibinfo {title} {Quantum {{Information Meets Quantum Matter}}: {{From Quantum Entanglement}} to {{Topological Phases}} of {{Many-Body Systems}}}}},\ Quantum {{Science}} and {{Technology}}\ (\bibinfo  {publisher} {Springer},\ \bibinfo {address} {New York, NY},\ \bibinfo {year} {2019})\BibitemShut {NoStop}%
\bibitem [{\citenamefont {Affleck}\ \emph {et~al.}(1987)\citenamefont {Affleck}, \citenamefont {Kennedy}, \citenamefont {Lieb},\ and\ \citenamefont {Tasaki}}]{affleckRigorousResultsValencebond1987}%
  \BibitemOpen
  \bibfield  {author} {\bibinfo {author} {\bibfnamefont {I.}~\bibnamefont {Affleck}}, \bibinfo {author} {\bibfnamefont {T.}~\bibnamefont {Kennedy}}, \bibinfo {author} {\bibfnamefont {E.~H.}\ \bibnamefont {Lieb}},\ and\ \bibinfo {author} {\bibfnamefont {H.}~\bibnamefont {Tasaki}},\ }\bibfield  {title} {\bibinfo {title} {Rigorous results on valence-bond ground states in antiferromagnets},\ }\href {https://doi.org/10.1103/PhysRevLett.59.799} {\bibfield  {journal} {\bibinfo  {journal} {Physical Review Letters}\ }\textbf {\bibinfo {volume} {59}},\ \bibinfo {pages} {799} (\bibinfo {year} {1987})}\BibitemShut {NoStop}%
\bibitem [{\citenamefont {Degen}\ \emph {et~al.}(2017)\citenamefont {Degen}, \citenamefont {Reinhard},\ and\ \citenamefont {Cappellaro}}]{degenQuantumSensing2017}%
  \BibitemOpen
  \bibfield  {author} {\bibinfo {author} {\bibfnamefont {C.~L.}\ \bibnamefont {Degen}}, \bibinfo {author} {\bibfnamefont {F.}~\bibnamefont {Reinhard}},\ and\ \bibinfo {author} {\bibfnamefont {P.}~\bibnamefont {Cappellaro}},\ }\bibfield  {title} {\bibinfo {title} {Quantum sensing},\ }\href {https://doi.org/10.1103/RevModPhys.89.035002} {\bibfield  {journal} {\bibinfo  {journal} {Reviews of Modern Physics}\ }\textbf {\bibinfo {volume} {89}},\ \bibinfo {pages} {035002} (\bibinfo {year} {2017})}\BibitemShut {NoStop}%
\bibitem [{\citenamefont {Jozsa}(2005)}]{jozsaIntroductionMeasurementBased2005}%
  \BibitemOpen
  \bibfield  {author} {\bibinfo {author} {\bibfnamefont {R.}~\bibnamefont {Jozsa}},\ }\href {http://arxiv.org/abs/quant-ph/0508124} {\bibinfo {title} {An introduction to measurement based quantum computation}} (\bibinfo {year} {2005}),\ \Eprint {https://arxiv.org/abs/quant-ph/0508124} {arXiv:quant-ph/0508124} \BibitemShut {NoStop}%
\bibitem [{\citenamefont {Hu}\ \emph {et~al.}(2021)\citenamefont {Hu}, \citenamefont {Liang},\ and\ \citenamefont {Calderbank}}]{hu2021climbing}%
  \BibitemOpen
  \bibfield  {author} {\bibinfo {author} {\bibfnamefont {J.}~\bibnamefont {Hu}}, \bibinfo {author} {\bibfnamefont {Q.}~\bibnamefont {Liang}},\ and\ \bibinfo {author} {\bibfnamefont {R.}~\bibnamefont {Calderbank}},\ }\bibfield  {title} {\bibinfo {title} {Climbing the diagonal clifford hierarchy},\ }\href@noop {} {\bibfield  {journal} {\bibinfo  {journal} {arXiv preprint arXiv:2110.11923}\ } (\bibinfo {year} {2021})}\BibitemShut {NoStop}%
\bibitem [{\citenamefont {Raussendorf}(2013)}]{raussendorf2013contextuality}%
  \BibitemOpen
  \bibfield  {author} {\bibinfo {author} {\bibfnamefont {R.}~\bibnamefont {Raussendorf}},\ }\bibfield  {title} {\bibinfo {title} {Contextuality in measurement-based quantum computation},\ }\href@noop {} {\bibfield  {journal} {\bibinfo  {journal} {Physical Review A—Atomic, Molecular, and Optical Physics}\ }\textbf {\bibinfo {volume} {88}},\ \bibinfo {pages} {022322} (\bibinfo {year} {2013})}\BibitemShut {NoStop}%
\bibitem [{\citenamefont {Raussendorf}\ and\ \citenamefont {Briegel}(2001)}]{raussendorfOneWayQuantumComputer2001}%
  \BibitemOpen
  \bibfield  {author} {\bibinfo {author} {\bibfnamefont {R.}~\bibnamefont {Raussendorf}}\ and\ \bibinfo {author} {\bibfnamefont {H.~J.}\ \bibnamefont {Briegel}},\ }\bibfield  {title} {\bibinfo {title} {A {{One-Way Quantum Computer}}},\ }\href {https://doi.org/10.1103/PhysRevLett.86.5188} {\bibfield  {journal} {\bibinfo  {journal} {Physical Review Letters}\ }\textbf {\bibinfo {volume} {86}},\ \bibinfo {pages} {5188} (\bibinfo {year} {2001})}\BibitemShut {NoStop}%
\bibitem [{\citenamefont {Bremner}\ \emph {et~al.}(2010)\citenamefont {Bremner}, \citenamefont {Jozsa},\ and\ \citenamefont {Shepherd}}]{bremnerClassicalSimulationCommuting2010}%
  \BibitemOpen
  \bibfield  {author} {\bibinfo {author} {\bibfnamefont {M.~J.}\ \bibnamefont {Bremner}}, \bibinfo {author} {\bibfnamefont {R.}~\bibnamefont {Jozsa}},\ and\ \bibinfo {author} {\bibfnamefont {D.~J.}\ \bibnamefont {Shepherd}},\ }\bibfield  {title} {\bibinfo {title} {Classical simulation of commuting quantum computations implies collapse of the polynomial hierarchy},\ }\href {https://doi.org/10.1098/rspa.2010.0301} {\bibfield  {journal} {\bibinfo  {journal} {Proceedings of the Royal Society A: Mathematical, Physical and Engineering Sciences}\ }\textbf {\bibinfo {volume} {467}},\ \bibinfo {pages} {459} (\bibinfo {year} {2010})}\BibitemShut {NoStop}%
\bibitem [{\citenamefont {Bremner}\ \emph {et~al.}(2017)\citenamefont {Bremner}, \citenamefont {Montanaro},\ and\ \citenamefont {Shepherd}}]{bremnerAchievingQuantumSupremacy2017}%
  \BibitemOpen
  \bibfield  {author} {\bibinfo {author} {\bibfnamefont {M.~J.}\ \bibnamefont {Bremner}}, \bibinfo {author} {\bibfnamefont {A.}~\bibnamefont {Montanaro}},\ and\ \bibinfo {author} {\bibfnamefont {D.~J.}\ \bibnamefont {Shepherd}},\ }\bibfield  {title} {\bibinfo {title} {Achieving quantum supremacy with sparse and noisy commuting quantum computations},\ }\href {https://doi.org/10.22331/q-2017-04-25-8} {\bibfield  {journal} {\bibinfo  {journal} {Quantum}\ }\textbf {\bibinfo {volume} {1}},\ \bibinfo {pages} {8} (\bibinfo {year} {2017})}\BibitemShut {NoStop}%
\bibitem [{\citenamefont {Gidney}(2015)}]{gidneyConstructingLargeControlled2015}%
  \BibitemOpen
  \bibfield  {author} {\bibinfo {author} {\bibfnamefont {C.}~\bibnamefont {Gidney}},\ }\href {https://algassert.com/circuits/2015/06/05/Constructing-Large-Controlled-Nots.html} {\bibinfo {title} {Constructing {{Large Controlled Nots}}}} (\bibinfo {year} {2015})\BibitemShut {NoStop}%
\bibitem [{\citenamefont {Ji}\ \emph {et~al.}(2018{\natexlab{a}})\citenamefont {Ji}, \citenamefont {Liu},\ and\ \citenamefont {Song}}]{jiPseudorandomQuantumStates2018}%
  \BibitemOpen
  \bibfield  {author} {\bibinfo {author} {\bibfnamefont {Z.}~\bibnamefont {Ji}}, \bibinfo {author} {\bibfnamefont {Y.-K.}\ \bibnamefont {Liu}},\ and\ \bibinfo {author} {\bibfnamefont {F.}~\bibnamefont {Song}},\ }\bibfield  {title} {\bibinfo {title} {Pseudorandom {{Quantum States}}},\ }in\ \href {https://doi.org/10.1007/978-3-319-96878-0_5} {\emph {\bibinfo {booktitle} {Advances in {{Cryptology}} -- {{CRYPTO}} 2018}}},\ \bibinfo {editor} {edited by\ \bibinfo {editor} {\bibfnamefont {H.}~\bibnamefont {Shacham}}\ and\ \bibinfo {editor} {\bibfnamefont {A.}~\bibnamefont {Boldyreva}}}\ (\bibinfo  {publisher} {Springer International Publishing},\ \bibinfo {address} {Cham},\ \bibinfo {year} {2018})\ pp.\ \bibinfo {pages} {126--152}\BibitemShut {NoStop}%
\bibitem [{\citenamefont {Ma}\ and\ \citenamefont {Huang}(2024)}]{maHowConstructRandom2024}%
  \BibitemOpen
  \bibfield  {author} {\bibinfo {author} {\bibfnamefont {F.}~\bibnamefont {Ma}}\ and\ \bibinfo {author} {\bibfnamefont {H.-Y.}\ \bibnamefont {Huang}},\ }\href {http://arxiv.org/abs/2410.10116} {\bibinfo {title} {How to construct random unitaries}} (\bibinfo {year} {2024}),\ \Eprint {https://arxiv.org/abs/2410.10116} {arXiv:2410.10116 [quant-ph]} \BibitemShut {NoStop}%
\bibitem [{\citenamefont {Regev}(2009)}]{regevLatticesLearningErrors2009}%
  \BibitemOpen
  \bibfield  {author} {\bibinfo {author} {\bibfnamefont {O.}~\bibnamefont {Regev}},\ }\bibfield  {title} {\bibinfo {title} {On lattices, learning with errors, random linear codes, and cryptography},\ }\href {https://doi.org/10.1145/1568318.1568324} {\bibfield  {journal} {\bibinfo  {journal} {J. ACM}\ }\textbf {\bibinfo {volume} {56}},\ \bibinfo {pages} {34:1} (\bibinfo {year} {2009})}\BibitemShut {NoStop}%
\bibitem [{Note4()}]{Note4}%
  \BibitemOpen
  \bibinfo {note} {Here, we assume that the energy source is classical and can be measured deterministically. If the energy source is quantum, we can always measure the work cost to $1/{\protect \ensuremath {\protect \mathsf {poly}}}(n)$ accuracy in polynomial time. Thus the work cost lower bound in the no-go result will change to $W_{\protect \mathrm {Haar}}-\protect \mathrm {negl}(n)$ where $\protect \mathrm {negl}(n)$ is a function that decays faster than any polynomial of $n$. In \protect \Cref {sec:pseudorandom}, we provide the proof details of the fully quantum case}\BibitemShut {NoStop}%
\bibitem [{Note5()}]{Note5}%
  \BibitemOpen
  \bibinfo {note} {We consider the unattainability formulation of the third law of thermodynamics, which applies not only to the cooling of thermal states but also the erasure/purification of arbitrary states. See e.g., \cite {masanesGeneralDerivationQuantification2017}.}\BibitemShut {Stop}%
\bibitem [{\citenamefont {Yang}(2025)}]{yangCompressionQuantumShallowcircuit2025}%
  \BibitemOpen
  \bibfield  {author} {\bibinfo {author} {\bibfnamefont {Y.}~\bibnamefont {Yang}},\ }\bibfield  {title} {\bibinfo {title} {Compression of quantum shallow-circuit states},\ }\href {https://doi.org/10.1103/PhysRevLett.134.010603} {\bibfield  {journal} {\bibinfo  {journal} {Physical Review Letters}\ }\textbf {\bibinfo {volume} {134}},\ \bibinfo {pages} {010603} (\bibinfo {year} {2025})},\ \Eprint {https://arxiv.org/abs/2404.11177} {arXiv:2404.11177 [quant-ph]} \BibitemShut {NoStop}%
\bibitem [{\citenamefont {Watanabe}\ and\ \citenamefont {Takagi}(2024)}]{watanabe2024black}%
  \BibitemOpen
  \bibfield  {author} {\bibinfo {author} {\bibfnamefont {K.}~\bibnamefont {Watanabe}}\ and\ \bibinfo {author} {\bibfnamefont {R.}~\bibnamefont {Takagi}},\ }\bibfield  {title} {\bibinfo {title} {Black box work extraction and composite hypothesis testing},\ }\href@noop {} {\bibfield  {journal} {\bibinfo  {journal} {Physical Review Letters}\ }\textbf {\bibinfo {volume} {133}},\ \bibinfo {pages} {250401} (\bibinfo {year} {2024})}\BibitemShut {NoStop}%
\bibitem [{\citenamefont {Chakraborty}\ \emph {et~al.}(2024)\citenamefont {Chakraborty}, \citenamefont {Das}, \citenamefont {Ghorui}, \citenamefont {Hazra},\ and\ \citenamefont {Singh}}]{chakraborty2024sample}%
  \BibitemOpen
  \bibfield  {author} {\bibinfo {author} {\bibfnamefont {S.}~\bibnamefont {Chakraborty}}, \bibinfo {author} {\bibfnamefont {S.}~\bibnamefont {Das}}, \bibinfo {author} {\bibfnamefont {A.}~\bibnamefont {Ghorui}}, \bibinfo {author} {\bibfnamefont {S.}~\bibnamefont {Hazra}},\ and\ \bibinfo {author} {\bibfnamefont {U.}~\bibnamefont {Singh}},\ }\bibfield  {title} {\bibinfo {title} {Sample complexity of black box work extraction},\ }\href@noop {} {\bibfield  {journal} {\bibinfo  {journal} {arXiv preprint arXiv:2412.02673}\ } (\bibinfo {year} {2024})}\BibitemShut {NoStop}%
\bibitem [{\citenamefont {B{\'e}rut}\ \emph {et~al.}(2012)\citenamefont {B{\'e}rut}, \citenamefont {Arakelyan}, \citenamefont {Petrosyan}, \citenamefont {Ciliberto}, \citenamefont {Dillenschneider},\ and\ \citenamefont {Lutz}}]{berutExperimentalVerificationLandauers2012}%
  \BibitemOpen
  \bibfield  {author} {\bibinfo {author} {\bibfnamefont {A.}~\bibnamefont {B{\'e}rut}}, \bibinfo {author} {\bibfnamefont {A.}~\bibnamefont {Arakelyan}}, \bibinfo {author} {\bibfnamefont {A.}~\bibnamefont {Petrosyan}}, \bibinfo {author} {\bibfnamefont {S.}~\bibnamefont {Ciliberto}}, \bibinfo {author} {\bibfnamefont {R.}~\bibnamefont {Dillenschneider}},\ and\ \bibinfo {author} {\bibfnamefont {E.}~\bibnamefont {Lutz}},\ }\bibfield  {title} {\bibinfo {title} {Experimental verification of {{Landauer}}'s principle linking information and thermodynamics},\ }\href {https://doi.org/10.1038/nature10872} {\bibfield  {journal} {\bibinfo  {journal} {Nature}\ }\textbf {\bibinfo {volume} {483}},\ \bibinfo {pages} {187} (\bibinfo {year} {2012})}\BibitemShut {NoStop}%
\bibitem [{\citenamefont {Gaudenzi}\ \emph {et~al.}(2018)\citenamefont {Gaudenzi}, \citenamefont {Burzur{\'i}}, \citenamefont {Maegawa}, \citenamefont {{van der Zant}},\ and\ \citenamefont {Luis}}]{gaudenziQuantumLandauerErasure2018}%
  \BibitemOpen
  \bibfield  {author} {\bibinfo {author} {\bibfnamefont {R.}~\bibnamefont {Gaudenzi}}, \bibinfo {author} {\bibfnamefont {E.}~\bibnamefont {Burzur{\'i}}}, \bibinfo {author} {\bibfnamefont {S.}~\bibnamefont {Maegawa}}, \bibinfo {author} {\bibfnamefont {H.~S.~J.}\ \bibnamefont {{van der Zant}}},\ and\ \bibinfo {author} {\bibfnamefont {F.}~\bibnamefont {Luis}},\ }\bibfield  {title} {\bibinfo {title} {Quantum {{Landauer}} erasure with a molecular nanomagnet},\ }\href {https://doi.org/10.1038/s41567-018-0070-7} {\bibfield  {journal} {\bibinfo  {journal} {Nature Physics}\ }\textbf {\bibinfo {volume} {14}},\ \bibinfo {pages} {565} (\bibinfo {year} {2018})}\BibitemShut {NoStop}%
\bibitem [{\citenamefont {Scandi}\ \emph {et~al.}(2022)\citenamefont {Scandi}, \citenamefont {Barker}, \citenamefont {Lehmann}, \citenamefont {Dick}, \citenamefont {Maisi},\ and\ \citenamefont {{Perarnau-Llobet}}}]{scandiMinimallyDissipativeInformation2022}%
  \BibitemOpen
  \bibfield  {author} {\bibinfo {author} {\bibfnamefont {M.}~\bibnamefont {Scandi}}, \bibinfo {author} {\bibfnamefont {D.}~\bibnamefont {Barker}}, \bibinfo {author} {\bibfnamefont {S.}~\bibnamefont {Lehmann}}, \bibinfo {author} {\bibfnamefont {K.~A.}\ \bibnamefont {Dick}}, \bibinfo {author} {\bibfnamefont {V.~F.}\ \bibnamefont {Maisi}},\ and\ \bibinfo {author} {\bibfnamefont {M.}~\bibnamefont {{Perarnau-Llobet}}},\ }\bibfield  {title} {\bibinfo {title} {Minimally {{Dissipative Information Erasure}} in a {{Quantum Dot}} via {{Thermodynamic Length}}},\ }\href {https://doi.org/10.1103/PhysRevLett.129.270601} {\bibfield  {journal} {\bibinfo  {journal} {Physical Review Letters}\ }\textbf {\bibinfo {volume} {129}},\ \bibinfo {pages} {270601} (\bibinfo {year} {2022})}\BibitemShut {NoStop}%
\bibitem [{\citenamefont {Koski}\ \emph {et~al.}(2014)\citenamefont {Koski}, \citenamefont {Maisi}, \citenamefont {Pekola},\ and\ \citenamefont {Averin}}]{koskiExperimentalRealizationSzilard2014}%
  \BibitemOpen
  \bibfield  {author} {\bibinfo {author} {\bibfnamefont {J.~V.}\ \bibnamefont {Koski}}, \bibinfo {author} {\bibfnamefont {V.~F.}\ \bibnamefont {Maisi}}, \bibinfo {author} {\bibfnamefont {J.~P.}\ \bibnamefont {Pekola}},\ and\ \bibinfo {author} {\bibfnamefont {D.~V.}\ \bibnamefont {Averin}},\ }\bibfield  {title} {\bibinfo {title} {Experimental realization of a {{Szilard}} engine with a single electron},\ }\href {https://doi.org/10.1073/pnas.1406966111} {\bibfield  {journal} {\bibinfo  {journal} {Proceedings of the National Academy of Sciences}\ }\textbf {\bibinfo {volume} {111}},\ \bibinfo {pages} {13786} (\bibinfo {year} {2014})}\BibitemShut {NoStop}%
\bibitem [{\citenamefont {Vidrighin}\ \emph {et~al.}(2016)\citenamefont {Vidrighin}, \citenamefont {Dahlsten}, \citenamefont {Barbieri}, \citenamefont {Kim}, \citenamefont {Vedral},\ and\ \citenamefont {Walmsley}}]{vidrighinPhotonicMaxwellsDemon2016}%
  \BibitemOpen
  \bibfield  {author} {\bibinfo {author} {\bibfnamefont {M.~D.}\ \bibnamefont {Vidrighin}}, \bibinfo {author} {\bibfnamefont {O.}~\bibnamefont {Dahlsten}}, \bibinfo {author} {\bibfnamefont {M.}~\bibnamefont {Barbieri}}, \bibinfo {author} {\bibfnamefont {M.~S.}\ \bibnamefont {Kim}}, \bibinfo {author} {\bibfnamefont {V.}~\bibnamefont {Vedral}},\ and\ \bibinfo {author} {\bibfnamefont {I.~A.}\ \bibnamefont {Walmsley}},\ }\bibfield  {title} {\bibinfo {title} {Photonic {{Maxwell}}'s {{Demon}}},\ }\href {https://doi.org/10.1103/PhysRevLett.116.050401} {\bibfield  {journal} {\bibinfo  {journal} {Physical Review Letters}\ }\textbf {\bibinfo {volume} {116}},\ \bibinfo {pages} {050401} (\bibinfo {year} {2016})}\BibitemShut {NoStop}%
\bibitem [{\citenamefont {Camati}\ \emph {et~al.}(2016)\citenamefont {Camati}, \citenamefont {Peterson}, \citenamefont {Batalh{\~a}o}, \citenamefont {Micadei}, \citenamefont {Souza}, \citenamefont {Sarthour}, \citenamefont {Oliveira},\ and\ \citenamefont {Serra}}]{camatiExperimentalRectificationEntropy2016}%
  \BibitemOpen
  \bibfield  {author} {\bibinfo {author} {\bibfnamefont {P.~A.}\ \bibnamefont {Camati}}, \bibinfo {author} {\bibfnamefont {J.~P.~S.}\ \bibnamefont {Peterson}}, \bibinfo {author} {\bibfnamefont {T.~B.}\ \bibnamefont {Batalh{\~a}o}}, \bibinfo {author} {\bibfnamefont {K.}~\bibnamefont {Micadei}}, \bibinfo {author} {\bibfnamefont {A.~M.}\ \bibnamefont {Souza}}, \bibinfo {author} {\bibfnamefont {R.~S.}\ \bibnamefont {Sarthour}}, \bibinfo {author} {\bibfnamefont {I.~S.}\ \bibnamefont {Oliveira}},\ and\ \bibinfo {author} {\bibfnamefont {R.~M.}\ \bibnamefont {Serra}},\ }\bibfield  {title} {\bibinfo {title} {Experimental {{Rectification}} of {{Entropy Production}} by {{Maxwell}}'s {{Demon}} in a {{Quantum System}}},\ }\href {https://doi.org/10.1103/PhysRevLett.117.240502} {\bibfield  {journal} {\bibinfo  {journal} {Physical Review Letters}\ }\textbf {\bibinfo {volume} {117}},\ \bibinfo {pages} {240502} (\bibinfo {year} {2016})}\BibitemShut {NoStop}%
\bibitem [{\citenamefont {Cottet}\ \emph {et~al.}(2017)\citenamefont {Cottet}, \citenamefont {Jezouin}, \citenamefont {Bretheau}, \citenamefont {{Campagne-Ibarcq}}, \citenamefont {Ficheux}, \citenamefont {Anders}, \citenamefont {Auff{\`e}ves}, \citenamefont {Azouit}, \citenamefont {Rouchon},\ and\ \citenamefont {Huard}}]{cottetObservingQuantumMaxwell2017}%
  \BibitemOpen
  \bibfield  {author} {\bibinfo {author} {\bibfnamefont {N.}~\bibnamefont {Cottet}}, \bibinfo {author} {\bibfnamefont {S.}~\bibnamefont {Jezouin}}, \bibinfo {author} {\bibfnamefont {L.}~\bibnamefont {Bretheau}}, \bibinfo {author} {\bibfnamefont {P.}~\bibnamefont {{Campagne-Ibarcq}}}, \bibinfo {author} {\bibfnamefont {Q.}~\bibnamefont {Ficheux}}, \bibinfo {author} {\bibfnamefont {J.}~\bibnamefont {Anders}}, \bibinfo {author} {\bibfnamefont {A.}~\bibnamefont {Auff{\`e}ves}}, \bibinfo {author} {\bibfnamefont {R.}~\bibnamefont {Azouit}}, \bibinfo {author} {\bibfnamefont {P.}~\bibnamefont {Rouchon}},\ and\ \bibinfo {author} {\bibfnamefont {B.}~\bibnamefont {Huard}},\ }\bibfield  {title} {\bibinfo {title} {Observing a quantum {{Maxwell}} demon at work},\ }\href {https://doi.org/10.1073/pnas.1704827114} {\bibfield  {journal} {\bibinfo  {journal} {Proceedings of the National Academy of Sciences}\ }\textbf {\bibinfo {volume} {114}},\ \bibinfo {pages} {7561} (\bibinfo {year} {2017})}\BibitemShut {NoStop}%
\bibitem [{\citenamefont {Auff{\`e}ves}(2022)}]{auffevesQuantumTechnologiesNeed2022}%
  \BibitemOpen
  \bibfield  {author} {\bibinfo {author} {\bibfnamefont {A.}~\bibnamefont {Auff{\`e}ves}},\ }\bibfield  {title} {\bibinfo {title} {Quantum {{Technologies Need}} a {{Quantum Energy Initiative}}},\ }\href {https://doi.org/10.1103/PRXQuantum.3.020101} {\bibfield  {journal} {\bibinfo  {journal} {PRX Quantum}\ }\textbf {\bibinfo {volume} {3}},\ \bibinfo {pages} {020101} (\bibinfo {year} {2022})}\BibitemShut {NoStop}%
\bibitem [{\citenamefont {Munson}\ \emph {et~al.}(2024)\citenamefont {Munson}, \citenamefont {Kothakonda}, \citenamefont {Haferkamp}, \citenamefont {Halpern}, \citenamefont {Eisert},\ and\ \citenamefont {Faist}}]{munsonComplexityconstrainedQuantumThermodynamics2024}%
  \BibitemOpen
  \bibfield  {author} {\bibinfo {author} {\bibfnamefont {A.}~\bibnamefont {Munson}}, \bibinfo {author} {\bibfnamefont {N.~B.~T.}\ \bibnamefont {Kothakonda}}, \bibinfo {author} {\bibfnamefont {J.}~\bibnamefont {Haferkamp}}, \bibinfo {author} {\bibfnamefont {N.~Y.}\ \bibnamefont {Halpern}}, \bibinfo {author} {\bibfnamefont {J.}~\bibnamefont {Eisert}},\ and\ \bibinfo {author} {\bibfnamefont {P.}~\bibnamefont {Faist}},\ }\href {http://arxiv.org/abs/2403.04828} {\bibinfo {title} {Complexity-constrained quantum thermodynamics}} (\bibinfo {year} {2024}),\ \Eprint {https://arxiv.org/abs/2403.04828} {arXiv:2403.04828 [quant-ph]} \BibitemShut {NoStop}%
\bibitem [{\citenamefont {Leone}\ \emph {et~al.}(2025)\citenamefont {Leone}, \citenamefont {Rizzo}, \citenamefont {Eisert},\ and\ \citenamefont {Jerbi}}]{leoneEntanglementTheoryLimited2025}%
  \BibitemOpen
  \bibfield  {author} {\bibinfo {author} {\bibfnamefont {L.}~\bibnamefont {Leone}}, \bibinfo {author} {\bibfnamefont {J.}~\bibnamefont {Rizzo}}, \bibinfo {author} {\bibfnamefont {J.}~\bibnamefont {Eisert}},\ and\ \bibinfo {author} {\bibfnamefont {S.}~\bibnamefont {Jerbi}},\ }\href {https://doi.org/10.48550/arXiv.2502.12284} {\bibinfo {title} {Entanglement theory with limited computational resources}} (\bibinfo {year} {2025}),\ \Eprint {https://arxiv.org/abs/2502.12284} {arXiv:2502.12284 [quant-ph]} \BibitemShut {NoStop}%
\bibitem [{\citenamefont {Bakshi}\ \emph {et~al.}(2024)\citenamefont {Bakshi}, \citenamefont {Liu}, \citenamefont {Moitra},\ and\ \citenamefont {Tang}}]{bakshi2024learning}%
  \BibitemOpen
  \bibfield  {author} {\bibinfo {author} {\bibfnamefont {A.}~\bibnamefont {Bakshi}}, \bibinfo {author} {\bibfnamefont {A.}~\bibnamefont {Liu}}, \bibinfo {author} {\bibfnamefont {A.}~\bibnamefont {Moitra}},\ and\ \bibinfo {author} {\bibfnamefont {E.}~\bibnamefont {Tang}},\ }\bibfield  {title} {\bibinfo {title} {Learning quantum hamiltonians at any temperature in polynomial time},\ }in\ \href@noop {} {\emph {\bibinfo {booktitle} {Proceedings of the 56th Annual ACM Symposium on Theory of Computing}}}\ (\bibinfo {year} {2024})\ pp.\ \bibinfo {pages} {1470--1477}\BibitemShut {NoStop}%
\bibitem [{\citenamefont {Chen}\ \emph {et~al.}(2025)\citenamefont {Chen}, \citenamefont {Anshu},\ and\ \citenamefont {Nguyen}}]{chen2025learning}%
  \BibitemOpen
  \bibfield  {author} {\bibinfo {author} {\bibfnamefont {C.-F.}\ \bibnamefont {Chen}}, \bibinfo {author} {\bibfnamefont {A.}~\bibnamefont {Anshu}},\ and\ \bibinfo {author} {\bibfnamefont {Q.~T.}\ \bibnamefont {Nguyen}},\ }\bibfield  {title} {\bibinfo {title} {Learning quantum gibbs states locally and efficiently},\ }\href@noop {} {\bibfield  {journal} {\bibinfo  {journal} {arXiv preprint arXiv:2504.02706}\ } (\bibinfo {year} {2025})}\BibitemShut {NoStop}%
\bibitem [{\citenamefont {Mele}\ and\ \citenamefont {Herasymenko}(2025)}]{meleEfficientLearningQuantum2025}%
  \BibitemOpen
  \bibfield  {author} {\bibinfo {author} {\bibfnamefont {A.~A.}\ \bibnamefont {Mele}}\ and\ \bibinfo {author} {\bibfnamefont {Y.}~\bibnamefont {Herasymenko}},\ }\bibfield  {title} {\bibinfo {title} {Efficient {{Learning}} of {{Quantum States Prepared With Few Fermionic Non-Gaussian Gates}}},\ }\href {https://doi.org/10.1103/PRXQuantum.6.010319} {\bibfield  {journal} {\bibinfo  {journal} {PRX Quantum}\ }\textbf {\bibinfo {volume} {6}},\ \bibinfo {pages} {010319} (\bibinfo {year} {2025})}\BibitemShut {NoStop}%
\bibitem [{\citenamefont {Mele}\ \emph {et~al.}(2024)\citenamefont {Mele}, \citenamefont {Mele}, \citenamefont {Bittel}, \citenamefont {Eisert}, \citenamefont {Giovannetti}, \citenamefont {Lami}, \citenamefont {Leone},\ and\ \citenamefont {Oliviero}}]{meleLearningQuantumStates2024}%
  \BibitemOpen
  \bibfield  {author} {\bibinfo {author} {\bibfnamefont {F.~A.}\ \bibnamefont {Mele}}, \bibinfo {author} {\bibfnamefont {A.~A.}\ \bibnamefont {Mele}}, \bibinfo {author} {\bibfnamefont {L.}~\bibnamefont {Bittel}}, \bibinfo {author} {\bibfnamefont {J.}~\bibnamefont {Eisert}}, \bibinfo {author} {\bibfnamefont {V.}~\bibnamefont {Giovannetti}}, \bibinfo {author} {\bibfnamefont {L.}~\bibnamefont {Lami}}, \bibinfo {author} {\bibfnamefont {L.}~\bibnamefont {Leone}},\ and\ \bibinfo {author} {\bibfnamefont {S.~F.~E.}\ \bibnamefont {Oliviero}},\ }\href {https://doi.org/10.48550/arXiv.2405.01431} {\bibinfo {title} {Learning quantum states of continuous variable systems}} (\bibinfo {year} {2024}),\ \Eprint {https://arxiv.org/abs/2405.01431} {arXiv:2405.01431 [quant-ph]} \BibitemShut {NoStop}%
\bibitem [{\citenamefont {Hayden}\ and\ \citenamefont {Preskill}(2007)}]{haydenBlackHolesMirrors2007}%
  \BibitemOpen
  \bibfield  {author} {\bibinfo {author} {\bibfnamefont {P.}~\bibnamefont {Hayden}}\ and\ \bibinfo {author} {\bibfnamefont {J.}~\bibnamefont {Preskill}},\ }\bibfield  {title} {\bibinfo {title} {Black holes as mirrors: Quantum information in random subsystems},\ }\href {https://doi.org/10.1088/1126-6708/2007/09/120} {\bibfield  {journal} {\bibinfo  {journal} {Journal of High Energy Physics}\ }\textbf {\bibinfo {volume} {2007}},\ \bibinfo {pages} {120} (\bibinfo {year} {2007})}\BibitemShut {NoStop}%
\bibitem [{\citenamefont {Bouland}\ \emph {et~al.}(2020)\citenamefont {Bouland}, \citenamefont {Fefferman},\ and\ \citenamefont {Vazirani}}]{boulandComputationalPseudorandomnessWormhole2020}%
  \BibitemOpen
  \bibfield  {author} {\bibinfo {author} {\bibfnamefont {A.}~\bibnamefont {Bouland}}, \bibinfo {author} {\bibfnamefont {B.}~\bibnamefont {Fefferman}},\ and\ \bibinfo {author} {\bibfnamefont {U.}~\bibnamefont {Vazirani}},\ }\bibfield  {title} {\bibinfo {title} {Computational {{Pseudorandomness}}, the {{Wormhole Growth Paradox}}, and {{Constraints}} on the {{AdS}}/{{CFT Duality}}},\ }in\ \href {https://doi.org/10.4230/LIPIcs.ITCS.2020.63} {\emph {\bibinfo {booktitle} {11th {{Innovations}} in {{Theoretical Computer Science Conference}} ({{ITCS}} 2020)}}},\ \bibinfo {series} {Leibniz {{International Proceedings}} in {{Informatics}} ({{LIPIcs}})}, Vol.\ \bibinfo {volume} {151},\ \bibinfo {editor} {edited by\ \bibinfo {editor} {\bibfnamefont {T.}~\bibnamefont {Vidick}}}\ (\bibinfo  {publisher} {Schloss Dagstuhl -- Leibniz-Zentrum f{\"u}r Informatik},\ \bibinfo {address} {Dagstuhl, Germany},\ \bibinfo {year} {2020})\ pp.\ \bibinfo {pages} {63:1--63:2}\BibitemShut {NoStop}%
\bibitem [{\citenamefont {Akers}\ \emph {et~al.}(2024)\citenamefont {Akers}, \citenamefont {Bouland}, \citenamefont {Chen}, \citenamefont {Kohler}, \citenamefont {Metger},\ and\ \citenamefont {Vazirani}}]{akersHolographicPseudoentanglementComplexity2024}%
  \BibitemOpen
  \bibfield  {author} {\bibinfo {author} {\bibfnamefont {C.}~\bibnamefont {Akers}}, \bibinfo {author} {\bibfnamefont {A.}~\bibnamefont {Bouland}}, \bibinfo {author} {\bibfnamefont {L.}~\bibnamefont {Chen}}, \bibinfo {author} {\bibfnamefont {T.}~\bibnamefont {Kohler}}, \bibinfo {author} {\bibfnamefont {T.}~\bibnamefont {Metger}},\ and\ \bibinfo {author} {\bibfnamefont {U.}~\bibnamefont {Vazirani}},\ }\href {https://doi.org/10.48550/arXiv.2411.04978} {\bibinfo {title} {Holographic pseudoentanglement and the complexity of the {{AdS}}/{{CFT}} dictionary}} (\bibinfo {year} {2024}),\ \Eprint {https://arxiv.org/abs/2411.04978} {arXiv:2411.04978 [hep-th]} \BibitemShut {NoStop}%
\bibitem [{\citenamefont {Swingle}\ \emph {et~al.}(2016)\citenamefont {Swingle}, \citenamefont {Bentsen}, \citenamefont {{Schleier-Smith}},\ and\ \citenamefont {Hayden}}]{swingleMeasuringScramblingQuantum2016}%
  \BibitemOpen
  \bibfield  {author} {\bibinfo {author} {\bibfnamefont {B.}~\bibnamefont {Swingle}}, \bibinfo {author} {\bibfnamefont {G.}~\bibnamefont {Bentsen}}, \bibinfo {author} {\bibfnamefont {M.}~\bibnamefont {{Schleier-Smith}}},\ and\ \bibinfo {author} {\bibfnamefont {P.}~\bibnamefont {Hayden}},\ }\bibfield  {title} {\bibinfo {title} {Measuring the scrambling of quantum information},\ }\href {https://doi.org/10.1103/PhysRevA.94.040302} {\bibfield  {journal} {\bibinfo  {journal} {Physical Review A}\ }\textbf {\bibinfo {volume} {94}},\ \bibinfo {pages} {040302} (\bibinfo {year} {2016})}\BibitemShut {NoStop}%
\bibitem [{\citenamefont {Zdeborov{\'a}}\ and\ \citenamefont {Krzakala}(2016)}]{zdeborovaStatisticalPhysicsInference2016}%
  \BibitemOpen
  \bibfield  {author} {\bibinfo {author} {\bibfnamefont {L.}~\bibnamefont {Zdeborov{\'a}}}\ and\ \bibinfo {author} {\bibfnamefont {F.}~\bibnamefont {Krzakala}},\ }\bibfield  {title} {\bibinfo {title} {Statistical physics of inference: Thresholds and algorithms},\ }\href {https://doi.org/10.1080/00018732.2016.1211393} {\bibfield  {journal} {\bibinfo  {journal} {Advances in Physics}\ }\textbf {\bibinfo {volume} {65}},\ \bibinfo {pages} {453} (\bibinfo {year} {2016})}\BibitemShut {NoStop}%
\bibitem [{\citenamefont {Watanabe}\ and\ \citenamefont {Takagi}(2025)}]{watanabe2025universal}%
  \BibitemOpen
  \bibfield  {author} {\bibinfo {author} {\bibfnamefont {K.}~\bibnamefont {Watanabe}}\ and\ \bibinfo {author} {\bibfnamefont {R.}~\bibnamefont {Takagi}},\ }\bibfield  {title} {\bibinfo {title} {Universal work extraction in quantum thermodynamics},\ }\href@noop {} {\bibfield  {journal} {\bibinfo  {journal} {arXiv preprint arXiv:2504.12373}\ } (\bibinfo {year} {2025})}\BibitemShut {NoStop}%
\bibitem [{\citenamefont {Lumbreras}\ \emph {et~al.}(2025)\citenamefont {Lumbreras}, \citenamefont {Huang}, \citenamefont {Hu}, \citenamefont {Gu},\ and\ \citenamefont {Tomamichel}}]{lumbreras2025quantum}%
  \BibitemOpen
  \bibfield  {author} {\bibinfo {author} {\bibfnamefont {J.}~\bibnamefont {Lumbreras}}, \bibinfo {author} {\bibfnamefont {R.~C.}\ \bibnamefont {Huang}}, \bibinfo {author} {\bibfnamefont {Y.}~\bibnamefont {Hu}}, \bibinfo {author} {\bibfnamefont {M.}~\bibnamefont {Gu}},\ and\ \bibinfo {author} {\bibfnamefont {M.}~\bibnamefont {Tomamichel}},\ }\bibfield  {title} {\bibinfo {title} {Quantum state-agnostic work extraction (almost) without dissipation},\ }\href@noop {} {\bibfield  {journal} {\bibinfo  {journal} {arXiv preprint arXiv:2505.09456}\ } (\bibinfo {year} {2025})}\BibitemShut {NoStop}%
\bibitem [{\citenamefont {Wilde}(2013)}]{wildeQuantumInformationTheory2013}%
  \BibitemOpen
  \bibfield  {author} {\bibinfo {author} {\bibfnamefont {M.~M.}\ \bibnamefont {Wilde}},\ }\href {https://doi.org/10.1017/CBO9781139525343} {\emph {\bibinfo {title} {Quantum {{Information Theory}}}}}\ (\bibinfo  {publisher} {Cambridge University Press},\ \bibinfo {address} {Cambridge},\ \bibinfo {year} {2013})\BibitemShut {NoStop}%
\bibitem [{\citenamefont {Vershynin}(2018)}]{vershyninHighDimensionalProbabilityIntroduction2018}%
  \BibitemOpen
  \bibfield  {author} {\bibinfo {author} {\bibfnamefont {R.}~\bibnamefont {Vershynin}},\ }\href {https://doi.org/10.1017/9781108231596} {\emph {\bibinfo {title} {High-{{Dimensional Probability}}: {{An Introduction}} with {{Applications}} in {{Data Science}}}}},\ Cambridge {{Series}} in {{Statistical}} and {{Probabilistic Mathematics}}\ (\bibinfo  {publisher} {Cambridge University Press},\ \bibinfo {address} {Cambridge},\ \bibinfo {year} {2018})\BibitemShut {NoStop}%
\bibitem [{\citenamefont {Sch{\"o}n}\ \emph {et~al.}(2005)\citenamefont {Sch{\"o}n}, \citenamefont {Solano}, \citenamefont {Verstraete}, \citenamefont {Cirac},\ and\ \citenamefont {Wolf}}]{schonSequentialGenerationEntangled2005}%
  \BibitemOpen
  \bibfield  {author} {\bibinfo {author} {\bibfnamefont {C.}~\bibnamefont {Sch{\"o}n}}, \bibinfo {author} {\bibfnamefont {E.}~\bibnamefont {Solano}}, \bibinfo {author} {\bibfnamefont {F.}~\bibnamefont {Verstraete}}, \bibinfo {author} {\bibfnamefont {J.~I.}\ \bibnamefont {Cirac}},\ and\ \bibinfo {author} {\bibfnamefont {M.~M.}\ \bibnamefont {Wolf}},\ }\bibfield  {title} {\bibinfo {title} {Sequential generation of entangled multiqubit states},\ }\href {https://doi.org/10.1103/PhysRevLett.95.110503} {\bibfield  {journal} {\bibinfo  {journal} {Physical Review Letters}\ }\textbf {\bibinfo {volume} {95}},\ \bibinfo {pages} {110503} (\bibinfo {year} {2005})}\BibitemShut {NoStop}%
\bibitem [{\citenamefont {Yoshida}\ \emph {et~al.}(2023)\citenamefont {Yoshida}, \citenamefont {Soeda},\ and\ \citenamefont {Murao}}]{yoshida2023reversing}%
  \BibitemOpen
  \bibfield  {author} {\bibinfo {author} {\bibfnamefont {S.}~\bibnamefont {Yoshida}}, \bibinfo {author} {\bibfnamefont {A.}~\bibnamefont {Soeda}},\ and\ \bibinfo {author} {\bibfnamefont {M.}~\bibnamefont {Murao}},\ }\bibfield  {title} {\bibinfo {title} {Reversing unknown qubit-unitary operation, deterministically and exactly},\ }\href@noop {} {\bibfield  {journal} {\bibinfo  {journal} {Physical Review Letters}\ }\textbf {\bibinfo {volume} {131}},\ \bibinfo {pages} {120602} (\bibinfo {year} {2023})}\BibitemShut {NoStop}%
\bibitem [{\citenamefont {Jozsa}\ \emph {et~al.}(1998)\citenamefont {Jozsa}, \citenamefont {Horodecki}, \citenamefont {Horodecki},\ and\ \citenamefont {Horodecki}}]{jozsa1998universal}%
  \BibitemOpen
  \bibfield  {author} {\bibinfo {author} {\bibfnamefont {R.}~\bibnamefont {Jozsa}}, \bibinfo {author} {\bibfnamefont {M.}~\bibnamefont {Horodecki}}, \bibinfo {author} {\bibfnamefont {P.}~\bibnamefont {Horodecki}},\ and\ \bibinfo {author} {\bibfnamefont {R.}~\bibnamefont {Horodecki}},\ }\bibfield  {title} {\bibinfo {title} {Universal quantum information compression},\ }\href@noop {} {\bibfield  {journal} {\bibinfo  {journal} {Physical review letters}\ }\textbf {\bibinfo {volume} {81}},\ \bibinfo {pages} {1714} (\bibinfo {year} {1998})}\BibitemShut {NoStop}%
\bibitem [{\citenamefont {Jozsa}\ and\ \citenamefont {Presnell}(2003)}]{jozsa2003universal}%
  \BibitemOpen
  \bibfield  {author} {\bibinfo {author} {\bibfnamefont {R.}~\bibnamefont {Jozsa}}\ and\ \bibinfo {author} {\bibfnamefont {S.}~\bibnamefont {Presnell}},\ }\bibfield  {title} {\bibinfo {title} {Universal quantum information compression and degrees of prior knowledge},\ }\href@noop {} {\bibfield  {journal} {\bibinfo  {journal} {Proceedings of the Royal Society of London. Series A: Mathematical, Physical and Engineering Sciences}\ }\textbf {\bibinfo {volume} {459}},\ \bibinfo {pages} {3061} (\bibinfo {year} {2003})}\BibitemShut {NoStop}%
\bibitem [{\citenamefont {Hayashi}\ and\ \citenamefont {Matsumoto}(2002)}]{hayashi2002quantum}%
  \BibitemOpen
  \bibfield  {author} {\bibinfo {author} {\bibfnamefont {M.}~\bibnamefont {Hayashi}}\ and\ \bibinfo {author} {\bibfnamefont {K.}~\bibnamefont {Matsumoto}},\ }\bibfield  {title} {\bibinfo {title} {Quantum universal variable-length source coding},\ }\href@noop {} {\bibfield  {journal} {\bibinfo  {journal} {Physical Review A}\ }\textbf {\bibinfo {volume} {66}},\ \bibinfo {pages} {022311} (\bibinfo {year} {2002})}\BibitemShut {NoStop}%
\bibitem [{\citenamefont {Haah}\ \emph {et~al.}(2017)\citenamefont {Haah}, \citenamefont {Harrow}, \citenamefont {Ji}, \citenamefont {Wu},\ and\ \citenamefont {Yu}}]{haahSampleoptimalTomographyQuantum2017}%
  \BibitemOpen
  \bibfield  {author} {\bibinfo {author} {\bibfnamefont {J.}~\bibnamefont {Haah}}, \bibinfo {author} {\bibfnamefont {A.~W.}\ \bibnamefont {Harrow}}, \bibinfo {author} {\bibfnamefont {Z.}~\bibnamefont {Ji}}, \bibinfo {author} {\bibfnamefont {X.}~\bibnamefont {Wu}},\ and\ \bibinfo {author} {\bibfnamefont {N.}~\bibnamefont {Yu}},\ }\bibfield  {title} {\bibinfo {title} {Sample-optimal tomography of quantum states},\ }\href {https://doi.org/10.1109/TIT.2017.2719044} {\bibfield  {journal} {\bibinfo  {journal} {IEEE Transactions on Information Theory}\ ,\ \bibinfo {pages} {1}} (\bibinfo {year} {2017})}\BibitemShut {NoStop}%
\bibitem [{\citenamefont {O'Donnell}\ and\ \citenamefont {Wright}(2015)}]{odonnellEfficientQuantumTomography2015}%
  \BibitemOpen
  \bibfield  {author} {\bibinfo {author} {\bibfnamefont {R.}~\bibnamefont {O'Donnell}}\ and\ \bibinfo {author} {\bibfnamefont {J.}~\bibnamefont {Wright}},\ }\href {http://arxiv.org/abs/1508.01907} {\bibinfo {title} {Efficient quantum tomography}} (\bibinfo {year} {2015}),\ \Eprint {https://arxiv.org/abs/1508.01907} {arXiv:1508.01907 [quant-ph]} \BibitemShut {NoStop}%
\bibitem [{\citenamefont {Ji}\ \emph {et~al.}(2018{\natexlab{b}})\citenamefont {Ji}, \citenamefont {Liu},\ and\ \citenamefont {Song}}]{ji2018pseudorandom}%
  \BibitemOpen
  \bibfield  {author} {\bibinfo {author} {\bibfnamefont {Z.}~\bibnamefont {Ji}}, \bibinfo {author} {\bibfnamefont {Y.-K.}\ \bibnamefont {Liu}},\ and\ \bibinfo {author} {\bibfnamefont {F.}~\bibnamefont {Song}},\ }\bibfield  {title} {\bibinfo {title} {Pseudorandom quantum states},\ }in\ \href@noop {} {\emph {\bibinfo {booktitle} {Annual International Cryptology Conference}}}\ (\bibinfo {organization} {Springer},\ \bibinfo {year} {2018})\ pp.\ \bibinfo {pages} {126--152}\BibitemShut {NoStop}%
\bibitem [{\citenamefont {Arunachalam}\ \emph {et~al.}(2021)\citenamefont {Arunachalam}, \citenamefont {Grilo},\ and\ \citenamefont {Sundaram}}]{arunachalam2021quantum}%
  \BibitemOpen
  \bibfield  {author} {\bibinfo {author} {\bibfnamefont {S.}~\bibnamefont {Arunachalam}}, \bibinfo {author} {\bibfnamefont {A.~B.}\ \bibnamefont {Grilo}},\ and\ \bibinfo {author} {\bibfnamefont {A.}~\bibnamefont {Sundaram}},\ }\bibfield  {title} {\bibinfo {title} {Quantum hardness of learning shallow classical circuits},\ }\href@noop {} {\bibfield  {journal} {\bibinfo  {journal} {SIAM Journal on Computing}\ }\textbf {\bibinfo {volume} {50}},\ \bibinfo {pages} {972} (\bibinfo {year} {2021})}\BibitemShut {NoStop}%
\bibitem [{\citenamefont {Gu}\ \emph {et~al.}(2024)\citenamefont {Gu}, \citenamefont {Quek}, \citenamefont {Yelin}, \citenamefont {Eisert},\ and\ \citenamefont {Leone}}]{guSimulatingQuantumChaos2024}%
  \BibitemOpen
  \bibfield  {author} {\bibinfo {author} {\bibfnamefont {A.}~\bibnamefont {Gu}}, \bibinfo {author} {\bibfnamefont {Y.}~\bibnamefont {Quek}}, \bibinfo {author} {\bibfnamefont {S.}~\bibnamefont {Yelin}}, \bibinfo {author} {\bibfnamefont {J.}~\bibnamefont {Eisert}},\ and\ \bibinfo {author} {\bibfnamefont {L.}~\bibnamefont {Leone}},\ }\href {http://arxiv.org/abs/2410.18196} {\bibinfo {title} {Simulating quantum chaos without chaos}} (\bibinfo {year} {2024}),\ \Eprint {https://arxiv.org/abs/2410.18196} {arXiv:2410.18196 [quant-ph]} \BibitemShut {NoStop}%
\bibitem [{\citenamefont {Sagawa}\ and\ \citenamefont {Ueda}(2009)}]{sagawaMinimalEnergyCost2009}%
  \BibitemOpen
  \bibfield  {author} {\bibinfo {author} {\bibfnamefont {T.}~\bibnamefont {Sagawa}}\ and\ \bibinfo {author} {\bibfnamefont {M.}~\bibnamefont {Ueda}},\ }\bibfield  {title} {\bibinfo {title} {Minimal {{Energy Cost}} for {{Thermodynamic Information Processing}}: {{Measurement}} and {{Information Erasure}}},\ }\href {https://doi.org/10.1103/PhysRevLett.102.250602} {\bibfield  {journal} {\bibinfo  {journal} {Physical Review Letters}\ }\textbf {\bibinfo {volume} {102}},\ \bibinfo {pages} {250602} (\bibinfo {year} {2009})}\BibitemShut {NoStop}%
\bibitem [{\citenamefont {Schulman}\ \emph {et~al.}(2005)\citenamefont {Schulman}, \citenamefont {Mor},\ and\ \citenamefont {Weinstein}}]{schulman2005physical}%
  \BibitemOpen
  \bibfield  {author} {\bibinfo {author} {\bibfnamefont {L.~J.}\ \bibnamefont {Schulman}}, \bibinfo {author} {\bibfnamefont {T.}~\bibnamefont {Mor}},\ and\ \bibinfo {author} {\bibfnamefont {Y.}~\bibnamefont {Weinstein}},\ }\bibfield  {title} {\bibinfo {title} {Physical limits of heat-bath algorithmic cooling},\ }\href@noop {} {\bibfield  {journal} {\bibinfo  {journal} {Physical review letters}\ }\textbf {\bibinfo {volume} {94}},\ \bibinfo {pages} {120501} (\bibinfo {year} {2005})}\BibitemShut {NoStop}%
\bibitem [{\citenamefont {Chattopadhyay}\ \emph {et~al.}(2025)\citenamefont {Chattopadhyay}, \citenamefont {Misra}, \citenamefont {Pandit},\ and\ \citenamefont {Paul}}]{chattopadhyay2025landauer}%
  \BibitemOpen
  \bibfield  {author} {\bibinfo {author} {\bibfnamefont {P.}~\bibnamefont {Chattopadhyay}}, \bibinfo {author} {\bibfnamefont {A.}~\bibnamefont {Misra}}, \bibinfo {author} {\bibfnamefont {T.}~\bibnamefont {Pandit}},\ and\ \bibinfo {author} {\bibfnamefont {G.}~\bibnamefont {Paul}},\ }\bibfield  {title} {\bibinfo {title} {Landauer principle and thermodynamics of computation},\ }\href@noop {} {\bibfield  {journal} {\bibinfo  {journal} {Reports on Progress in Physics}\ } (\bibinfo {year} {2025})}\BibitemShut {NoStop}%
\bibitem [{\citenamefont {Skrzypczyk}\ \emph {et~al.}(2014)\citenamefont {Skrzypczyk}, \citenamefont {Short},\ and\ \citenamefont {Popescu}}]{skrzypczyk2014work}%
  \BibitemOpen
  \bibfield  {author} {\bibinfo {author} {\bibfnamefont {P.}~\bibnamefont {Skrzypczyk}}, \bibinfo {author} {\bibfnamefont {A.~J.}\ \bibnamefont {Short}},\ and\ \bibinfo {author} {\bibfnamefont {S.}~\bibnamefont {Popescu}},\ }\bibfield  {title} {\bibinfo {title} {Work extraction and thermodynamics for individual quantum systems},\ }\href@noop {} {\bibfield  {journal} {\bibinfo  {journal} {Nature communications}\ }\textbf {\bibinfo {volume} {5}},\ \bibinfo {pages} {4185} (\bibinfo {year} {2014})}\BibitemShut {NoStop}%
\end{thebibliography}%

\clearpage
\onecolumngrid

\appendix
\appendixpage
\tableofcontents

\vspace{1em}

The appendices are structured as follows.
In \Cref{sec:rel-work}, we discuss how this work relates to the existing and concurrent literature.
In \Cref{sec:prelim}, we provide a formal and pedagogical introduction to the thermodynamic framework and the quantum information tools that we use in this work.
In \Cref{sec:correctness}, we give the detailed proof of the correctness of the learning-to-erase protocol in the general formulation.
In \Cref{sec:pseudorandom}, we provide the detailed proof of the computational hardness result for erasure.

\section{Related work}
\label{sec:rel-work}

\textbf{Computational complexity of thermodynamic tasks.}
The relation between energy dissipation and information processing has been known since the foundational work of Landauer \cite{landauerIrreversibilityHeatGeneration1961,landauerInformationPhysical1991} and Bennett \cite{bennettLogicalReversibilityComputation1973,bennettThermodynamicsComputationReview1982}.
It has been further formalized and generalized to quantum systems and finite-shot scenarios in a series of works \cite{reebImprovedLandauerPrinciple2014,dahlstenInadequacyNeumannEntropy2011,feynmanFeynmanLecturesComputation2023,faistMinimalWorkCost2015,rioThermodynamicMeaningNegative2011}.
These works establish information-theoretic quantities, including various notions of entropy, as the key metric that characterizes the energy dissipation in information processing tasks, among which erasure and work extraction are the canonical examples.
These connections to entropic quantities allow one to use algorithmic primitives from compression and source coding \cite{nielsenQuantumComputationQuantum2010} to build thermodynamic protocols that achieve the optimal energy cost \cite{faistMinimalWorkCost2015}.
For example, in the erasure task of erasing $N$ copies of an $n$-qubit state, the information-theoretically optimal protocol proceeds in a compress-to-erase fashion \cite{dahlstenInadequacyNeumannEntropy2011,rioThermodynamicMeaningNegative2011}: we first compress the states to a minimal number of qubits, and then erase them one by one.
However, such compression protocols, developed in the early days of quantum information, focus on the case where each copy has a constant Hilbert space dimension $D=2^n$ and have runtime polynomial in $N$ and $D$.
This means that the runtime of the resulting thermodynamic protocols grows exponentially with the system size $n$, making them prohibitively expensive for many-body systems.
We note that these earlier compression works define ``efficiency'' to be polynomial in $N, D=2^n$, which is not efficient (i.e., polynomial in $N, n$) in the modern language of quantum computation.

Only recently have people started to explore the efficiency of thermodynamic and compression tasks.
\cite{munsonComplexityconstrainedQuantumThermodynamics2024} proposed a general framework using complexity-constrained entropic quantities to characterize the energy cost when time complexity is limited.
Such quantities are hard to calculate, as they involve an optimization over all possible circuit configurations with bounded gate complexity.
As a result, it is not clear when efficiency can be achieved beyond simple examples.
In \cite{yangCompressionQuantumShallowcircuit2025}, an efficient compression protocol was developed for shallow circuit states.
Our work contributes to this growing literature that studies the computational complexity of thermodynamics, by providing provably-efficient and energy-optimal thermodynamic protocols for a wide range of physically-relevant states, and proving no-go theorems showing that inefficiency is inevitable for slightly more complex states.

\textbf{Energy cost of acquiring knowledge.}
The energy cost of erasure relates to the entropy of the state, which measures our ignorance.
This illustrates that the energy cost is fundamentally tied to our knowledge of the system.
However, entropic estimates of the energy cost do not take into account the potential cost of acquiring such knowledge.
Indeed, the compression primitives used in existing thermodynamic protocols rely on having complete knowledge of the state that we want to erase \cite{dahlstenInadequacyNeumannEntropy2011,rioThermodynamicMeaningNegative2011}.
Replacing the compression step by universal quantum compression protocols \cite{jozsa1998universal,jozsa2003universal,hayashi2002quantum,bennett2006universal} removes the need for complete knowledge.
But these protocols cause irreversible disturbance to the states, thus incurring additional energy cost. 

Several recent works \cite{vsafranek2023work,xuereb2024resources,watanabe2024black,chakraborty2024sample} formalized this intuition and explicitly asked whether acquiring knowledge has a fundamental energy cost.
In particular, \cite{vsafranek2023work,watanabe2024black,chakraborty2024sample} remarked that acquiring knowledge of the state through tomography (i.e., quantum learning algorithms) requires an enormous number of additional copies and is therefore not suitable for thermodynamic tasks.
\cite{xuereb2024resources} analyzed the energy cost of acquiring knowledge by calculating the energy cost of the measurements in standard tomography algorithms and regarded that as the energy cost of gaining knowledge.
The two works \cite{watanabe2025universal,lumbreras2025quantum} released subsequent to ours start from the same premise that standard tomography destroys many copies and uses measurements that may incur an additional energy cost; they also develop work extraction protocols that do not use standard tomography.
In contrast, our work shows that learning itself has no fundamental energy cost.
We prove this by showing that every learning algorithm can be made fully reversible, an idea that dates back to Bennett \cite{bennettThermodynamicsComputationReview1982}, thereby circumventing the thermodynamic analysis of measurements. 
In fact, we show that learning-to-erase protocols can be energy optimal and provably efficient, hence reviving the use of tomography for thermodynamic tasks. 

\textbf{Quantum learning theory.}
Recent years have seen a surge in efforts formalizing the process of acquiring knowledge in a quantum world via information theory and complexity theory, creating the field of quantum learning theory \cite{anshuSurveyComplexityLearning2024}.
After realizing that general state learning is exponentially expensive \cite{haahSampleoptimalTomographyQuantum2017,odonnellEfficientQuantumTomography2015}, people have developed efficient learning algorithms for more structured and physically relevant states, including matrix product states \cite{cramerEfficientQuantumState2010}, shallow circuit states \cite{huangLearningShallowQuantum2024,landauLearningQuantumStates2024,kimLearningStatePreparation2024}, states with low gate complexity \cite{zhaoLearningQuantumStates2024} and low magic \cite{montanaroLearningStabilizerStates2017,leoneLearningTdopedStabilizer2024,grewalEfficientLearningQuantum2024}, low-degree phase states \cite{arunachalamOptimalAlgorithmsLearning2023}, etc.
Our work opens up a new avenue by applying these efficient quantum learning algorithms: we show that any efficient learning algorithm can be turned into an efficient and energy-optimal thermodynamic protocol for erasure or work extraction.
Our work also adds a new dimension to the resource analysis of learning algorithms.
Apart from time, space, and sample complexity, our work establishes energy cost as another physical property of learning algorithms that we can analyze.

Tools from quantum cryptography, in particular pseudorandom states and unitaries \cite{ji2018pseudorandom,maHowConstructRandom2024,schuster2025random}, have been extensively used in proving the hardness of learning \cite{arunachalam2021quantum,zhaoLearningQuantumStates2024,schuster2025random}.
These have also found use in proving the computational hardness of some physical tasks, including classifying phases of matter \cite{schuster2025random}, recognizing causal structures \cite{schuster2025random}, recognizing signatures of chaos \cite{guSimulatingQuantumChaos2024}, performing entanglement distillation and dilution \cite{leoneEntanglementTheoryLimited2025}, etc.
Our no-go theorem contributes to this growing understanding of fundamental hardness in physical tasks, by establishing the computational hardness of performing thermodynamic tasks.

\section{Preliminaries}
\label{sec:prelim}

\subsection{Thermodynamics}

We begin with a formal description of the standard (one-shot) Landauer erasure setting (see e.g., \cite{reebImprovedLandauerPrinciple2014,dahlstenInadequacyNeumannEntropy2011,rioThermodynamicMeaningNegative2011,faistMinimalWorkCost2015}).
We will then state the (one-shot) Landauer principle that gives a lower bound on the work cost of erasure.
After that, we will give a formal description of the standard Landauer erasure protocol and a rigorous proof that it erases a single qubit with the optimal energy cost.
Building upon this, we explain the compress-to-erase protocol that erases an arbitrary state at the optimal work cost.
It is simple and information-theoretically optimal, but its computational complexity is generally exponential.
Finally, we conclude by explaining the work extraction task, a different thermodynamic task that is dual to erasure, the work cost of which is determined by the work cost of the erasure task.

\begin{definition}[Landauer erasure \cite{reebImprovedLandauerPrinciple2014}]
\label{def:erasure}
    Let $S$ be an $n$-qubit quantum system described by the Hilbert space $\mathbb{C}^{2^n}$.
    Let $E$ be the environment described by the Hilbert space $\mathbb{C}^{2^{|E|}}$, where the size of the environment $|E|$ can be infinitely large.
    The environment is associated with a local Hamiltonian $\mathcal{H}_E$ and a temperature $T\geq 0$.
    We use $k_B$ to denote the Boltzmann constant and $\beta = 1/(k_B T)$ to denote the inverse temperature.
    Initially, the system $S$ is in a given state $\rho \in \mathbb{C}^{2^n\times 2^n}$, the environment $E$ is in the thermal Gibbs state
    \begin{equation}
        \rho_E = \frac{e^{-\beta \mathcal{H}_E}}{\tr(e^{-\beta \mathcal{H}_E})},
    \end{equation}
    and the overall state is $\rho_{SE} = \rho \otimes \rho_E$.
    The task of Landauer erasure of $\rho$ under temperature $T$ is to apply a unitary $U$ acting on both the system and the environment that transforms the initial state into the final state 
    \begin{equation}
        \rho'_{SE} = U(\rho\otimes \rho_E)U^\dagger
    \end{equation}
    such that the reduced density matrix of the system becomes $\ket{0}$: $\rho'_S = \tr_E(\rho'_{SE}) = \ket{0^n}\bra{0^n}$.
    The work cost is defined as the energy change of the environment
    \begin{equation}
    \label{def:work}
        W = \tr(\mathcal{H}_E \rho'_{E}) - \tr(\mathcal{H}_E \rho_E),
    \end{equation}
    where $\rho'_E = \tr_S(\rho'_{SE})$ is the reduced density matrix of the environment after erasure.
    An erasure protocol is specified by the tuple $(U, \mathcal{H}_E)$, which we can strategically choose to minimize the work cost $W$.
\end{definition}

Several remarks help to clarify the physical consideration in \Cref{def:erasure}.

\begin{remark}[System Hamiltonian]
    In \Cref{def:erasure}, the Hamiltonian of the system $\mathcal{H}_S$ does not need to be specified \cite{reebImprovedLandauerPrinciple2014}.
    Only the Hamiltonian of the environment $\mathcal{H}_E$ is specified.
    This is because what we care about is the irreversible energy loss caused by heat dissipation that we have to pay, rather than the internal energy change of the system (see e.g., \cite[Sec. 2.2]{reebImprovedLandauerPrinciple2014}).
    This is equivalent to choosing a completely degenerate Hamiltonian $\mathcal{H}_S=0$ for the system, as in \cite{faistMinimalWorkCost2015,rioThermodynamicMeaningNegative2011,munsonComplexityconstrainedQuantumThermodynamics2024}.
    To make statements about the actual work done, we need to take into account the system's Hamiltonian in general \cite{sagawaMinimalEnergyCost2009}, but when $\mathcal{H}_S=0$, the actual work done is the same as the heat dissipation defined in \Cref{def:work}.
    Therefore, from now on, we will take the standard convention of $\mathcal{H}_S=0$ and use the following three words interchangeably: work cost, energy cost, and heat dissipation.
\end{remark}

\begin{remark}[System state]
    The initial state of the system $\rho$ (i.e., the state to be erased) can be arbitrary \cite{reebImprovedLandauerPrinciple2014}.
    In particular, it does not need to be a thermal Gibbs state of some Hamiltonian.
    This allows the study of Landauer erasure to be applicable to a wider range of scenarios in information processing, where the states that we prepare may contain useful information and are not necessarily in thermal equilibrium.
    There is a different line of work (sometimes referred to as algorithmic cooling, or simply cooling) that focuses on the cooling of physical systems to their ground states.
   In that setting, it is often assumed that the state of the system is initially in the thermal state of a given Hamiltonian (see e.g., \cite{tarantoLandauerNernstWhat2023,schulman2005physical}).
    This assumption is natural in a cooling setting, but limits the applicability of the results in general information processing scenarios.
\end{remark}

\begin{remark}[Ancilla qubits]
    Ancilla qubits that are in the zero state $\ket{0}$ can be freely introduced and used as long as they are erased at the end (i.e., put back to $\ket{0}$) \cite{faistMinimalWorkCost2015}.
    The requirement that the final state of the ancilla qubits is the same as the initial state ensures that one cannot hide entropy or work cost in the ancilla qubits.
    
    One may be interested in a more restricted setting where we only have access to ancilla qubits in arbitrary mixed states.
    This can be reduced to the setting with perfect ancilla, because we can first perform erasure on each ancilla to some arbitrarily small error, use it as if it is a perfect $\ket{0}$ ancilla, and at the end perform work extraction to restore it back to the initial mixed state.
    This procedure incurs no addition energy cost and it is computationally efficient, since the energy cost of the initial erasure step cancels with the final work yield exactly.
    Therefore, we can assume without loss of generality that we have access to perfect ancilla qubits for free as long as they are brought back to their initial states at the end.
\end{remark}

\begin{remark}[Finite-size effect]
    Throughout this work, we ignore finite-size effects by allowing the environment to have infinite size $|E|\to \infty$.
    This is because in this work we are focusing on studying Landauer erasure in the many-body regime (i.e., when the system size $n$ is large).
    It is well studied that a finite-size heat bath necessarily introduces small corrections to Landauer's principle that are on the scale of $O(1/2^{|E|})$ and independent of $n$ \cite{reebImprovedLandauerPrinciple2014}.
    Allowing infinite-size environment eliminates these distractions and enables us to focus on the scaling with respect to $n$.
    For similar reasons, we also ignore finite-error corrections that are on the order of $O(\log(1/\epsilon))$ for some constant error $\epsilon$ \cite{rioThermodynamicMeaningNegative2011,faistMinimalWorkCost2015}.
    These finite-error effects are still under active study (often for few-body systems) and are not the focus of this work.
    We refer interested readers to a recent survey \cite{chattopadhyay2025landauer} for more details.
\end{remark}

Under the formal setting described in \Cref{def:erasure}, Landauer's principle has been rigorously established \cite{dahlstenInadequacyNeumannEntropy2011,reebImprovedLandauerPrinciple2014,faistMinimalWorkCost2015}.
We state the one-shot version that best fits into our problem setting, where we are aiming at erasing a \textit{single} copy of the total state $\rho = \sum_{x=1}^m p_x (\ketbra{\psi_x})^{\otimes N}$.

\begin{lemma}[One-shot Landauer's principle \cite{faistMinimalWorkCost2015}]
\label{lem:landauer}
    The work cost of erasing a state $\rho$ under temperature $T$ using any protocol is lower bounded by
    \begin{equation}
        W\geq H_{\max}(\rho) k_B T \ln 2,
    \end{equation}
    where $H_{\max} = \log_2 \mathrm{rank}(\rho)$ is the max-entropy of $\rho$.
\end{lemma}

On the other hand, Landauer's principle can be achieved for a single qubit when there is an infinitely large environment \cite{rioThermodynamicMeaningNegative2011,skrzypczyk2014work}.
Below, we give a formal description of this protocol and rigorously prove that its work cost is indeed $k_BT\ln 2$, following the idea of \cite{skrzypczyk2014work}.
This formalizes the usually hand-wavy description of the standard Landauer erasure protocol: couple the system with the heat bath, slowly raise the energy level of $\ket{1}$ to infinity, decouple from the heat bath, and bring the energy level back to $0$.
For any $d$-dimensional quantum state $\rho$, a similar proof applies and shows that there is an erasure protocol with work cost $W=\log_2(d)k_BT \ln2$.

\begin{lemma}[Standard Landauer erasure]
\label{lem:standard-erasure}
    There is a protocol $(U, \mathcal{H}_E)$ that erases any single-qubit state $\rho$ under temperature $T$ with work cost
    \begin{equation}
        W = k_BT\ln 2.
    \end{equation}
\end{lemma}

\begin{proof}[Proof of \Cref{lem:standard-erasure}]
    We take an environment that has $|E|$ qubits and label them as qubit $1, \ldots, |E|$.
    For each qubit $i=1, \ldots, |E|$, we assign a single-qubit Hamiltonian
    \begin{equation}
        \mathcal{H}_i = \Delta_i\ketbra{1} = \frac{i-1}{\sqrt{|E|}}\ketbra{1}.
    \end{equation}
    In other words, for qubit $i$, the state $\ket{0}$ has energy $0$ and the state $\ket{1}$ has energy $\Delta_i=(i-1)/\sqrt{|E|}$ that goes to infinity as $|E|\to\infty$.
    The Hamiltonian of the environment is
    \begin{equation}
        \mathcal{H}_E = \sum_{i=1}^{|E|} \mathcal{H}_i.
    \end{equation}
    Initially, the state of the environment at inverse temperature $\beta=1/{k_BT}$ is the thermal Gibbs state $\rho_E = \bigotimes_{i=1}^{|E|}\rho_{i}$ where
    \begin{equation}
        \rho_{i} = \frac{e^{-\beta \mathcal{H}_i}}{\tr(e^{-\beta \mathcal{H}_i})} = \frac{1}{1+e^{-\beta\Delta_i}}\ketbra{0} + \frac{e^{-\beta \Delta_i}}{1+e^{-\beta\Delta_i}}\ketbra{1}.
    \end{equation}
    The initial joint state can be written as
    \begin{equation}
        \underbrace{\rho}_{\text{system}} \otimes \underbrace{\rho_1\otimes \rho_2 \otimes \cdots \otimes \rho_{|E|-1}\otimes  \rho_{|E|}}_{\text{environment qubit $1, \ldots, |E|$}}.
    \end{equation}
    The unitary that we execute is as follows.
    We swap the system qubit with environment qubit $1$, and then swap the system qubit with $2$, and then swap the system qubit with $3$, and so on, until we swap the system qubit with $|E|$.
    The resulting state is
    \begin{equation}
        \underbrace{\rho_{|E|}}_{\text{system}} \otimes \underbrace{\rho\otimes \rho_1 \otimes \cdots \otimes \rho_{|E-2|} \otimes\rho_{|E-1|}}_{\text{environment qubit $1, \ldots, |E|$}}.
    \end{equation}
    To see that this is a valid erasure protocol, we bound the trace distance between $\ket{0}$ the final system state $\rho_{|E|}$:
    \begin{equation}
        \dtr(\rho_{|E|}, \ketbra{0}) = \frac{1}{2}\left\|\frac{-e^{-\beta \Delta_{|E|}}}{1+e^{-\beta \Delta_{|E|}}}\ketbra{0} + \frac{e^{-\beta \Delta_{|E|}}}{1+e^{-\beta \Delta_{|E|}}}\ketbra{1}\right\|_1\leq e^{-\beta \left(\sqrt{|E|}-1/\sqrt{|E|}\right)} \to 0
    \end{equation}
    as $|E|\to \infty$.
    To calculate the work cost, we note that the first environment qubit does not have energy change since it has a degenerate Hamiltonian.
    Taking into account all the other qubits, we have
    \begin{equation}
    \begin{split}
        W &= \sum_{i=2}^{|E|}\tr(\mathcal{H}_i \rho_{i-1}) - \sum_{i=1}^{|E|}\tr(\mathcal{H}_i \rho_{i}) \\
        &= \sum_{i=2}^{|E|}\tr(\mathcal{H}_i \rho_{i-1}) - \sum_{i=2}^{|E|}\tr(\mathcal{H}_{i-1} \rho_{i-1}) - \tr(\mathcal{H}_{|E|} \rho_{|E|})\\
        &=\sum_{i=2}^{|E|} (\Delta_i-\Delta_{i-1})\frac{e^{-\beta \Delta_{i-1}}}{1+e^{-\beta \Delta_{i-1}}} - \Delta_{|E|}\frac{e^{-\beta \Delta_{|E|}}}{1+e^{-\beta \Delta_{|E|}}} \\
        &\to \int_{0}^{\infty} \mathrm{d}\Delta \frac{e^{-\beta \Delta}}{1+e^{-\beta \Delta}} - 0 \\
        &=\frac{1}{\beta}\ln 2 = k_BT\ln 2
    \end{split}
    \end{equation}
    as $|E|\to \infty$.
    This completes the proof of \Cref{lem:standard-erasure}.
\end{proof}

\begin{remark}[Computational complexity]
    We remark that in this standard single-qubit erasure protocol, the gate complexity of $U$ is already infinitely large (it grows as $O(|E|)\to \infty$).
    This is necessary and in accordance with the third law of thermodynamics \cite{reebImprovedLandauerPrinciple2014,tarantoLandauerNernstWhat2023}.
    That said, as we mentioned above, we are primarily interested in the many-body regime (i.e., the scaling behavior of the work cost and the computational complexity of erasure with respect to the system size $n$).
    Therefore, we regard the single-qubit Landauer erasure as an elementary step that only contributes unit computational complexity, similar to a two-qubit gate.
    More formally, we define the \textit{computational complexity} of an erasure protocol as the minimal number of two-qubit gates and single-qubit Landauer erasures that implements the protocol.
    This way of calculating computational complexity was also used in \cite{munsonComplexityconstrainedQuantumThermodynamics2024}.
\end{remark}

\begin{remark}[Unknown pure states]
    Erasure of a known pure state does not cost energy because the entropy of a pure state is zero.
    In this work, we consider the erasure of many copies of an unknown pure state.
    That means the state we need to erase is a mixed state of the form $\rho = \sum_{x=1}^m p_x (\ket{\psi_x}\bra{\psi_x})^{\otimes N}$, where $\{p_x\}$ are the prior probabilities. 
    It is this lack of knowledge that gives rise to the energy cost of erasure.
\end{remark}

For a general $n$-qubit state $\rho$, the erasure protocol that achieves the optimal work cost $W=H_{\max}(\rho)k_B T\ln 2$ is well established \cite{dahlstenInadequacyNeumannEntropy2011}, which we call compress-to-erase protocols.

\begin{lemma}[Compress-to-erase protocol \cite{dahlstenInadequacyNeumannEntropy2011}]
\label{lem:compress-to-erase}
    Let $\rho$ be an $n$-qubit state with max-entropy $H_{\max}(\rho)$.
    There is a protocol that erases $\rho$ under temperature $T$ with work cost
    \begin{equation}
        W = H_{\max}(\rho)k_BT\ln 2.
    \end{equation}
\end{lemma}
\begin{proof}[Proof of \Cref{lem:compress-to-erase}]
    For simplicity, we assume that $H_{\max}(\rho)$ is an integer.
    The proof extends to the general situation straightforwardly.
    The idea of this compress-to-erase protocol is very simple: first compress the $n$-qubit state $\rho$ to $H_{\max}(\rho)$ qubits, and then erase these qubits.
    More concretely, from the definition of max-entropy, we have
    \begin{equation}
        \mathrm{rank}(\rho) = 2^{H_{\max}(\rho)}.
    \end{equation}
    This means that the state $\rho$ has an eigen-decomposition
    \begin{equation}
        \rho = \sum_{i=1}^{2^{H_{\max}(\rho)}} \lambda_i \ketbra{\psi_i}
    \end{equation}
    with $2^{H_{\max}(\rho)}$ orthonormal eigenvectors $\ket{\psi_i}$.
    This defines a unitary $U$ satisfying
    \begin{equation}
        U\ket{\psi_i} = \ket{\mathrm{bin}(i)}\ket{0^{n-H_{\max}(\rho)}}, \quad \forall i=1, \ldots, 2^{H_{\max}(\rho)},
    \end{equation}
    where $\mathrm{bin}(i)$ is the length-$H_{\max}(\rho)$ binary representation of $i$.
    After applying this unitary $U$, the state is compressed to $H_{\max}(\rho)$ qubits with the rest qubits already erased, and we apply the standard single-qubit erasure protocol for each of these qubits.
    \Cref{lem:standard-erasure} implies that the work cost of this erasure protocol is
    \begin{equation}
        W = H_{\max}(\rho)k_B T\ln 2.
    \end{equation}
    This completes the proof of \Cref{lem:compress-to-erase}.
\end{proof}

\begin{remark}[Computational complexity of compress-to-erase]
    We note that the simple protocol of compress-to-erase in \Cref{lem:compress-to-erase} is already information-theoretic optimal (i.e., it achieves the optimal work cost $W=H_{\max}(\rho)k_BT\ln 2$).
    However, its computational complexity may be huge.
    Typically, building the unitary $U$ given the description of $\rho$ already takes exponentially long time.
    Moreover, the resulting unitary $U$ is also likely to be exponentially complex.
    In fact, our computational hardness result in \Cref{sec:pseudorandom} proves that such general-purpose compress-to-erase protocols is bound to have exponential computational complexity because of the existence of pseudorandom states.
\end{remark}

Lastly, a different thermodynamic task that is closely related to erasure is work extraction.
In a work extraction task with an $n$-qubit initial state $\rho$, the goal is to maximize the work yield no matter what state the system ends up being in.
It is closely related to erasure because any erasure protocol can be used to build a work extraction protocol and vice versa.
This is because, similar to \Cref{lem:standard-erasure}, one can show that with a single-qubit pure state (say $\ket{0}$), there is a work extraction protocol that transforms it into a maximally mixed state and produces work $k_BT\ln 2$.
Therefore, given an erasure protocol with work cost $W_{\mathrm{erase}}(\rho)$, we can use it to erase the state $\rho$ and then apply the single-qubit work extraction protocol to each of the qubits.
The total work yield is $W_{\mathrm{extract}}(\rho)=nk_B T\ln 2-W_{\mathrm{erase}}(\rho)$.
On the other hand, given any work extraction protocol with work yield $W_{\mathrm{extract}}(\rho)$, we can use it to transform the state $\rho$ into the maximally mixed state and then apply the single-qubit erasure protocol to each of the qubits.
The total work cost is $W_{\mathrm{erase}}(\rho)  =nk_B T\ln 2-W_{\mathrm{extract}}(\rho)$.
Using \Cref{lem:standard-erasure,lem:compress-to-erase}, we see that the optimal work yield from an $n$-qubit state $\rho$ is
\begin{equation}
    W = (n-H_{\max}(\rho))k_B T.
\end{equation}
This shows that the two tasks are dual to each other.
In particular, the results that we prove for erasure directly carry over to work extraction.

\subsection{Quantum Learning Theory}
\label{sec:learning}

Now we move on to give a formal definition of a quantum state learning algorithm that we will use in this work.
We refer interested readers to a recent survey \cite{anshuSurveyComplexityLearning2024} for more discussion.

\begin{definition}[Quantum state learning algorithm]
\label{def:learning}
    Let $\mathcal{C}$ be a set of $n$-qubit quantum states.
    Let $\epsilon\in [0, 1]$ be the learning accuracy and $\delta\in [0, 1]$ be the failure probability.
    Let $s$ be a positive integer.
    A quantum state learning algorithm $\mathcal{L}$ for the class of states $\mathcal{C}$ with sample complexity $s$, accuracy $\epsilon$, and failure probability $\delta$ is a quantum algorithm that takes $s$ copies of a state $\rho\in \mathcal{C}$ and outputs a circuit description of a state $\hat{\rho}$ such that
    \begin{equation}
        \dtr(\rho, \hat{\rho})\leq \epsilon
    \end{equation}
    with probability at least $1-\delta$ for all $\rho\in \mathcal{C}$.
\end{definition}

\begin{remark}[Circuit description as the output]
    Here, we require quantum state learning algorithms to output circuit descriptions of the learned quantum state, rather than any other classical description of the quantum state.
    This is crucial in defining the computational complexity of a quantum learning algorithm.
    Otherwise one can hide a huge amount of computational complexity in the translation between different classical descriptions of the same quantum state.
    To see this, consider the scenario where the states we want to learn are all ground states of different Hamiltonians.
    If the learning algorithm only outputs the Hamiltonian of the learned state, which is a valid classical description, figuring out the circuit description and preparing the state may still take exponential time.
    Requiring the learner to output the circuit descriptions ensures that the output can indeed be used to prepare copies of the unknown state.
    See \cite[Appendix B]{zhaoLearningQuantumStates2024} for more discussion on this choice.
    We note that the learning algorithms we use in this work (i.e., for shallow circuit states, $t$-doped stabilizer states, matrix product states, and low-degree phase states) all fits into \Cref{def:learning}.
\end{remark}

\begin{remark}[Efficient learning ensures efficient state (un)preparation]
    Under \Cref{def:learning}, if there is a polynomial-time quantum learning algorithm for a class of states $\mathcal{C}$, then every state in this class must have a polynomial-time quantum circuit that prepares the state up to $\epsilon$ error.
    This circuit is the output of the learning algorithm.
    Otherwise, writing down the learning outcome would take super-polynomial time and the learning algorithm cannot be efficient.
    Moreover, given the learned circuit description, one can easily construct the circuit that unprepares the state: simply reverse the order of the gates in the learned circuit and take the complex conjugates of them.
\end{remark}

\subsection{Quantum Cryptography}

Here, we review pseudorandom states, the cryptographic object that we will use to prove computational hardness results for erasure.
Informally speaking, pseudorandom states are ensembles of quantum states that cannot be distinguished from Haar random states by any polynomial-time quantum algorithm.
This notion has been of great interest in quantum cryptography, complexity theory, and learning theory.

\begin{definition}[Pseudorandom quantum states (PRS) \cite{ji2018pseudorandom}]
    \label{def:PRS}
  Let $n$ denote the security parameter. 
  Let $\mathcal{K} = \{\mathcal{K}_n\}_{n \in \mathbb{N}}$ be the key space. 
  A keyed family of pure quantum states $\left\{\left|\phi_k\right\rangle\right\}_{k \in \mathcal{K}_n}$ is \emph{pseudorandom against $t(n)$ adversaries} if the following two conditions hold:
\begin{enumerate}
    \item (Efficient generation). There is a polynomial-time quantum algorithm $\mathsf{Gen}$ that generates state $\left|\phi_k\right\rangle$ on input $k$. That is, for all $n \in \mathbb{N}$ and for all $k \in \mathcal{K}_n, \mathsf{Gen}(1^n, k)=\left|\phi_k\right\rangle$.
    \item (Pseudorandomness). Any polynomial number of copies of $\left|\phi_k\right\rangle$ with the same random $k \in \mathcal{K}_n$ are computationally indistinguishable from the same number of copies of a Haar-random state. More precisely, for any $t(n)$-time quantum algorithm $\mathcal{D}$ and any copy number $N = \poly(n)$, there exists a negligible function $\mathrm{negl}(\cdot)$ (i.e., a function that grows more slowly than any inverse polynomial of its argument) such that for all $n \in \mathbb{N}$,
    \begin{equation}
        \left|\Pr_{k \leftarrow \mathcal{K}_n}\left[\mathcal{D}\left(\left|\phi_k\right\rangle^{\otimes N}\right)=1\right]-\Pr_{|\psi\rangle \leftarrow \mu}\left[\mathcal{D}\left(|\psi\rangle^{\otimes N}\right)=1\right] \right|\leq\operatorname{negl}(n),
    \end{equation}
    where $\mu$ is the Haar measure over pure states on $n$ qubits.
\end{enumerate}
When $t(n)=\poly(n)$, we simply say that the states are \emph{pseudorandom}.
\end{definition}

Intuitively, one can think of $\mathcal{D}$ as a distinguisher trying to tell pseudorandom states from Haar random states.
Its output value is binary: 0 means that it thinks the state is a PRS and 1 stands for a Haar random state.
Then the above definition states that the probability it thinks the state is Haar random is roughly the same no matter whether the input is pseudorandom or Haar random.
This formalizes the notion of indistinguishably.

In a recent breakthrough \cite{schuster2025random}, it has been shown that pseudorandom states can be constructed with extremely low depth (and therefore low magic, entanglement entropy, etc.), assuming the standard cryptographic conjecture that learning with error problems (LWE) are sub-exponentially hard \cite{regevLatticesLearningErrors2009}.

\begin{lemma}[Low-depth pseudorandom states {\cite[Corollary 2]{schuster2025random}}]
\label{lem:low-depth-prs}
    Under the conjecture that no subexponential-time quantum algorithm can solve LWE, pseudorandom states over $n$ qubits can be formed using one-dimensional circuits of depth $\polylog(n)$.
\end{lemma}

\section{Correctness of the learning-to-erase protocol}
\label{sec:correctness}

In this section, we give a detailed proof for the correctness of the learning-to-erase protocol in the general formulation.
In the whole protocol, there are in total five registers:
\begin{enumerate}
    \item the sample register $S$, $ns$ qubits that stores $s$ copies of the $n$-qubit state $\ket{\psi_x}$ that we will feed into the learning algorithm; 
    \item the remaining-copies register $R$, $n(N-s)$ qubits that stores the remaining $N-s$ copies of the state $\ket{\psi_x}$;
    \item the memory register $M$, $\log_2 m$ memory qubits that stores the learning outcome $x'\in [m]$;
    \item the ancilla register $A$, ancilla qubits that we will use as work space of the learning algorithm; and
    \item the additional memory register $M'$, additional $\log_2 m$ memory qubits that are used to store the learning outcome during the uncomputation.
\end{enumerate}
Registers $M, A, M'$ are initialized in the zero state $\ket{0}_M\ket{0}_A\ket{0}_M'$ while registers $S, R$ are initialized in the state $\rho=\sum_{x=1}^m p_x (\ketbra{\psi_x})^{\otimes N}$.
In the following, we will first consider the case where the initial state is $\ket{\psi_x}^{\otimes N}$ for any $x$, and then add in the probability mixture $\{p_x\}$.

We have seen in the main text that we can build a reversible learning algorithm $\mathcal{L}$ (i.e., a unitary) that satisfies
\begin{equation}
\label{eq:learning-alg}
    \mathcal{L}\ket{\psi_x}^{\otimes s}_S\ket{0}_M\ket{0}_A=\sum_{x'=1}^m c_{x'|x}\ket{x'}_M\ket{\mathrm{junk}_{x'}}_{S,A},
\end{equation}
where the coefficients $c_{x'|x} \in\mathbb{C}$ are determined by the probability of predicting $x'$ given the state $\ket{\psi_x}$.
In particular, $|c_{x|x}|^2\geq p_{\mathrm{succ}}$ is close to one for any $x$.

We first give a detailed recap on how we build the erasure protocol $\mathcal{E}(\mathcal{L})$ using the reversible learning algorithm $\mathcal{L}$.
The erasure protocol can be decomposed in five steps (see \Cref{fig:overview}(a)):
\begin{itemize}
    \item Step 1 (learn), where we execute the learning algorithm $\mathcal{L}$;
    \item Step 2 (copy), where we copy out the learning outcome;
    \item Step 3 (uncompute), where we uncompute the learning algorithm by executing $\mathcal{L}^\dagger$;
    \item Step 4 (unprepare), where we apply the state unpreparation unitary on all the copies stored in $S, R$; and 
    \item Step 5 (erase learning outcome), where we erase the $M'$ register that records the learning outcome.
\end{itemize}
The initial state of the system is
\begin{equation}
    \ket{\phi_0^x}=\ket{\psi_x}^{\otimes s}_S \ket{\psi_x}^{\otimes (N-s)}_R \ket{0}_{M}\ket{0}_A \ket{0}_{M'}.
\end{equation}

In Step 1 (learn), we run the learning algorithm $\mathcal{L}$ on the sample register $S$, the memory register $M$, and the ancilla work space $A$.
This gives us the following state
\begin{equation}
    \ket{\phi_1^x}=\ket{0}_{M'}\ket{\psi_x}^{\otimes (N-s)}_R \left(\mathcal{L}\ket{\psi_x}^{\otimes s}_S\ket{0}_{M}\ket{0}_A\right) = \ket{0}_{M'}\ket{\psi_x}^{\otimes (N-s)}_R \left(\sum_{x'=1}^m c_{x'|x}\ket{x'}_M\ket{\mathrm{junk}_{x'}}_{S,A} \right),
\end{equation}
where we have used the property of the learning algorithm.

In Step 2 (copy), we apply CNOT gates on $M'$ controlled by $M$ (i.e., apply CNOT on the $i$-th qubit of $M'$ controlled by the $i$-th qubit of $M$).
In this way, we copy the content of $M$ and write it into $M'$.
This is because the action of the CNOT gate satisfies
\begin{equation}
    \mathrm{CNOT}\ket{a}\ket{0} = \mathrm{CNOT}\ket{a}\ket{a}, \quad \forall a\in \{0, 1\}.
\end{equation}
Therefore, the state after Step 2 is
\begin{equation}
    \ket{\phi_2^x}=\ket{\psi_x}^{\otimes (N-s)}_R \left(\sum_{x'=1}^m c_{x'|x}\ket{x'}_{M'}\ket{x'}_M\ket{\mathrm{junk}_{x'}}_{S,A} \right).
\end{equation}

In Step 3 (uncompute), we uncompute the learning algorithm by applying $\mathcal{L}^\dagger$ to the sample register $S$, the memory register $M$, and the ancilla work space $A$.
The resulting state is
\begin{equation}
    \ket{\phi_3^x} = \ket{\psi_x}^{\otimes (N-s)}_R \mathcal{L}^\dagger \left(\sum_{x'=1}^m c_{x'|x}\ket{x'}_{M'}\ket{x'}_M\ket{\mathrm{junk}_{x'}}_{S,A} \right) = \ket{\psi_x}^{\otimes (N-s)}_R \ket{\eta_x}_{S, A, M, M'},
\end{equation}
where we have defined the state 
\begin{equation}
    \ket{\eta_x}_{S, A, M, M'}=\mathcal{L}^\dagger \left(\sum_{x'=1}^m c_{x'|x}\ket{x'}_{M'}\ket{x'}_M\ket{\mathrm{junk}_{x'}}_{S,A} \right).
\end{equation}
Now we show that $\ket{\eta_x}_{S, A, M, M'}$ is close to the following ideal state (the state when the learning algorithm is perfect $c_{x|x}=1$)
\begin{equation}
    \ket{\tilde{\eta}_x}_{S, A, M, M'} = \ket{\psi_x}^{\otimes s}_S \ket{0}_A \ket{0}_M \ket{x}_{M'}.
\end{equation}
To this end, we calculate the overlap between $\ket{\tilde{\eta}_x}_{S, A, M, M'}$ and $\ket{\eta_x}_{S, A, M, M'}$:
\begin{equation}
\begin{split}
    \braket{\tilde{\eta}_x}{\eta_x}_{S, A, M, M'} &= \left(\bra{x}_{M'}\bra{\psi_x}^{\otimes s}_S \bra{0}_{M,A} \right) \left(\mathcal{L}^\dagger \sum_{x'}c_{x'|x}\ket{x'}_{M'}\ket{x'}_M\ket{\mathrm{junk}_{x'}}_{S,A}\right) \\
    &=\sum_{x'}c_{x'|x}\braket{x}{x'}_{M'}\left(\bra{\psi_x}^{\otimes s}_S \bra{0}_{M,A} \mathcal{L}^\dagger \right) \left(\ket{x'}_M\ket{\mathrm{junk}_{x'}}_{S,A}\right) \\
    &= \sum_{x'}c_{x'|x}\braket{x}{x'}_{M'}\left(\mathcal{L}\ket{\psi_x}^{\otimes s}_S \ket{0}_{M,A}  \right)^\dagger \left(\ket{x'}_M\ket{\mathrm{junk}_{x'}}_{S,A}\right)\\
    &=c_{x|x}\left(\sum_{x''}c_{x''|x}\ket{x''}_M\ket{\mathrm{junk}_{x''}}_{S,A}\right)^\dagger\ket{x}_M\ket{\mathrm{junk}_x}_{S,A} \\
    &=c_{x|x}\left(\sum_{x''}c^*_{x''|x}\bra{x''}_M\bra{\mathrm{junk}_{x''}}_{S,A}\right)\ket{x}_M\ket{\mathrm{junk}_x}_{S,A} \\
    &=c_{x|x}\sum_{x''}c^*_{x''|x}\braket{x''}{x}_M\bra{\mathrm{junk}_{x''}}_{S,A}\ket{\mathrm{junk}_x}_{S,A} \\
    &=|c_{x|x}|^2\geq p_{\mathrm{succ}},
\end{split} 
\end{equation}
where we have used the action of $\mathcal{L}$ \Cref{eq:learning-alg}, the basis $\ket{x}$ being orthonormal, and the success probability guarantee $|c_{x|x}|^2\geq p_{\mathrm{succ}}, \forall x$.
This implies that the trace distance between $\ket{\eta_x}$ and $\ket{\tilde{\eta}_x}$ is small:
\begin{equation}
    \dtr(\ket{\eta_x}, \ket{\tilde{\eta}_x}) = \sqrt{1-|\braket{\eta_x}{\tilde{\eta}_x}|^2} = \sqrt{1-|c_{x|x}|^4}\leq \sqrt{1-p_{\mathrm{succ}}^2}.
\end{equation}
Therefore, at the end of Step 3 (uncompute), the state $\ket{\phi^x_3} = \ket{\psi_x}^{\otimes (N-s)}_R \ket{\eta_x}_{S, A, M, M'}$ is close to the following ideal state
\begin{equation}
    \ket{\tilde{\phi}_3^{x}} = \ket{\psi_x}^{\otimes (N-s)}_R \ket{\tilde{\eta}_x}_{S, A, M, M'} = \ket{\psi_x}^{\otimes (N-s)}_R\ket{\psi_x}^{\otimes s}_S\ket{0}_A\ket{0}_M \ket{x}_{M'}.
\end{equation}
In particular, we have
\begin{equation}
    \dtr(\ket{\phi^x_3}, \ket{\tilde{\phi}^x_3}) = \dtr(\ket{\eta_x}, \ket{\tilde{\eta}_x})\leq \sqrt{1-p_{\mathrm{succ}}^2}.
\end{equation}
From now on, we will keep track of both $\ket{\phi^x_3}$ and $\ket{\tilde{\phi}^x_3}$.

In Step 4 (unprepare), we apply the state unpreparation unitaries to all the copies $S, R$ controlled by the learning outcome in $M'$.
More formally, we apply a unitary $U_{\mathrm{unprep}}$ that satisfies
\begin{equation}
    U_{\mathrm{unprep}} \ket{x}_{M'} \ket{\psi_x}^{\otimes N}_{S, R} = \ket{x}_{M'}\ket{0}_{S, R}.
\end{equation}
The resulting state is
\begin{equation}
    \ket{\phi^x_4} = U_{\mathrm{unprep}}\ket{\phi^x_3}
\end{equation}
and 
\begin{equation}
    \ket{\tilde{\phi}^x_4} = U_{\mathrm{unprep}}\ket{\tilde{\phi}^x_3} = U_{\mathrm{unprep}}\ket{\psi_x}^{\otimes (N-s)}_R\ket{\psi_x}^{\otimes s}_S\ket{0}_A\ket{0}_M \ket{x}_{M'} = \ket{0}_{S,R}\ket{0}_A\ket{0}_M \ket{x}_{M'}.
\end{equation}

Finally, in Step 5 (erase learning outcome), we erase the $M'$ register that records the learning outcome using \Cref{lem:standard-erasure} on each qubit of $M'$.
This amounts to applying the erasure channel
\begin{equation}
    \mathcal{M}: \rho_{M'} \to \tr_{M'}(\rho) \ketbra{0}_{M'}, \quad \forall \rho.
\end{equation}
The resulting state is
\begin{equation}
    \rho_5^x = \mathcal{M}(\ketbra{\phi^x_4}) = \mathcal{M}(U_{\mathrm{unprep}}\ketbra{\phi^x_3}U_{\mathrm{unprep}}^\dagger)
\end{equation}
and
\begin{equation}
    \ket{\tilde{\phi}^x_5} = \ket{0}_{S, R}\ket{0}_A\ket{0}_M\ket{0}_{M'}.
\end{equation}
This means that in the ideal case the state is indeed in the all zero state $\ket{\tilde{\phi}^5_x}$.
The distance from the ideal case is bounded by
\begin{equation}
    \dtr(\rho^x_5, \ketbra{0}) = \dtr(\rho^x_5, \ketbra{\tilde{\phi}^x_5}) \leq \dtr(\ket{\phi_3^x}, \ket{\tilde{\phi}^x_3}) \leq \sqrt{1-p_{\mathrm{succ}}^2},
\end{equation}
where we have used the data processing inequality of trace distance that states the trace distance cannot increase under the channel $\mathcal{M}(U_{\mathrm{unprep}}(\cdot)U_{\mathrm{unprep}}^\dagger)$.
This proves that when the initial state is $\ket{\psi_x}^{\otimes N}$, the erasure protocol $\mathcal{E}(\mathcal{L})$ indeed erases the state.

When the initial state is $\rho = \sum_x p_x (\ketbra{\psi_x})^{\otimes N}$, we have
\begin{equation}
\begin{split}
    \dtr(\mathcal{E}(\mathcal{L})[\rho], \ketbra{0}) & = \dtr\left(\mathcal{E}(\mathcal{L})\left[\sum_x p_x (\ketbra{\psi_x})^{\otimes N}\right], \ketbra{0}\right)\\
    &= \dtr\left(\sum_x p_x \mathcal{E}(\mathcal{L})\left[(\ketbra{\psi_x})^{\otimes N}\right], \ketbra{0}\right) \\
    &= \frac{1}{2}\left\|\sum_x p_x \mathcal{E}(\mathcal{L})\left[(\ketbra{\psi_x})^{\otimes N}\right]-\ketbra{0}\right\|_1 \\
    &=\frac{1}{2}\left\|\sum_x p_x \left(\mathcal{E}(\mathcal{L})\left[(\ketbra{\psi_x})^{\otimes N}\right]-\ketbra{0}\right)\right\|_1 \\
    &\leq \sum_x p_x\frac{1}{2}\left\|\mathcal{E}(\mathcal{L})\left[(\ketbra{\psi_x})^{\otimes N}\right]-\ketbra{0}\right\|_1 \\
    &=\sum_x p_x \dtr(\rho^x_5, \ketbra{0}) \\
    &\leq \sqrt{1-p_{\mathrm{succ}}^2},
\end{split}
\end{equation}
where we have used the linearity of quantum channels, the definition of trace distance, $\sum_x p_x=1$, the triangle inequality of $\|\cdot\|_1$, and $\dtr(\rho^x_5, \ketbra{0})\leq \sqrt{1-p_{\mathrm{succ}}^2}$.
Hence, we have $\dtr(\mathcal{E}(\mathcal{L})[\rho], \ketbra{0})\to 0$ as $p_{\mathrm{succ}}\to 1$ as desired.
Meanwhile, the only part that costs energy is the erasure of learning outcome, where we used \Cref{lem:standard-erasure} to erase $\log_2 m$ qubits.
\Cref{lem:standard-erasure} then implies that the work cost of the learning-to-erase protocol $\mathcal{E}(\mathcal{L})$ is
\begin{equation}
    W = \log_2(m) k_BT\ln 2
\end{equation}
as claimed.
This completes the proof for the correctness of the learning-to-erase protocol.

\section{Computational hardness of erasure}
\label{sec:pseudorandom}

In this section, we give a detailed proof for the computational hardness of erasing pseudorandom states.
This implies the computational hardness of erasing general classes of states as long as the class includes pseudorandom states (e.g., the class of shallow circuit states with depth $d=\polylog(n)$).

In particular, as described in the main text, we prove that any polynomial-time quantum algorithm that erases $N=\poly(n)$ copies of pseudorandom states $\mathcal{C}=\{\ket{\psi_x}\}_{x=1}^m$ must have an energy cost that is at least
\begin{equation}
    W_{\mathrm{PRS}} \geq W_{\mathrm{Haar}} - \mathrm{negl}(n)
\end{equation}
where
\begin{equation}
    W_{\mathrm{Haar}}= \log_2\binom{N+2^n-1}{N}k_BT\ln 2
\end{equation}
is the work cost for erasing $N$ copies of Haar random states and $\mathrm{negl}(n)$ denotes a function that decays faster than any polynomial of $n$.

This is dramatically different from the prediction of Landauer's principle, which says that the energy cost of erasing pseudorandom states (information-theoretically, i.e., without constraints on computational complexity) is
\begin{equation}
    W_{\mathrm{PRS}}^{\mathrm{info}} = \log_2(m) k_B T \ln 2 = O(nd)k_BT\ln 2 = O(n\polylog(n))k_BT\ln 2,
\end{equation}
independent of $N$.
Here, we have used the work cost formula for shallow circuit states with depth $d$ and \Cref{lem:low-depth-prs} which states that pseudorandom states can be implemented in $d=\polylog(n)$ depth.
The difference between $W_{\mathrm{PRS}}\approx W_{\mathrm{Haar}}$ that scales with $N$ and the prediction of Landauer's principle $W_{\mathrm{PRS}}^{\mathrm{info}}$ that is independent of $N$ manifests the drastic change of fundamental physical laws when we put realistic computational complexity constraints on what we are allowed to do.

In the following, we follow the main text and give a detailed proof of this result by contradiction.
Recall that in \Cref{def:PRS}, we see that for any polynomial-time quantum algorithm $\mathcal{D}$, the distinguishability of pseudorandom states from Haar random states is negligible:
\begin{equation}
\label{eq:prs-indist}
    \left|\Pr_{\ket{\psi_x} \leftarrow \mathcal{C}}\left[\mathcal{D}\left(\left|\psi_x\right\rangle^{\otimes N}\right)=1\right]-\Pr_{|\psi\rangle \leftarrow \mu}\left[\mathcal{D}\left(|\psi\rangle^{\otimes N}\right)=1\right] \right|\leq\operatorname{negl}(n),
\end{equation}
where $\mu$ is the Haar measure over pure states on $n$ qubits.
Intuitively, one can think of $\mathcal{D}$ as a distinguisher trying to tell pseudorandom states from Haar random states.
Its output value is binary: 0 means that it thinks the state is a PRS and 1 stands for a Haar random state.
Then \Cref{eq:prs-indist} states that the probability it thinks the state is Haar random is roughly the same no matter whether the input is PRS or Haar random.

To prove the result, we assume for the sake of contradiction that there exists a polynomial-time erasure protocol $\mathcal{E}$ that erases PRS with work cost
\begin{equation}
    W_{\mathcal{E}} < W_{\mathrm{Haar}} - 1/q(n)
\end{equation}
for some polynomial function $q(n)$.
Now we construct a distinguisher $\mathcal{D}$ based on the erasure protocol $\mathcal{E}$ that violates \Cref{eq:prs-indist}.

The distinguisher is the following.
We first execute $\mathcal{E}$ on the given states and conduct two tests
\begin{enumerate}
    \item determine if the erasure is successful by measuring the projector $\ketbra{0}$; and
    \item test if the work cost is less than $W_{\mathrm{Haar}}$ by measuring the energy change in the heat bath.
\end{enumerate}
If both tests pass, $\mathcal{D}$ outputs $0$ (i.e., we think the state is PRS).
Otherwise, we think the state is Haar random and output $1$.

In this detailed proof we consider the heat bath to be quantum in general and the energy measurement is also quantum and thus stochastic.
But the proof extends straightforwardly when we assume the energy source is classical in which case the energy measurement is deterministic and easier.
More concretely, in the second test, we note that the energy expectation value of the heat bath before and after the erasure can be determined to $\epsilon$ error with failure probability at most $1/\poly(n)$ in time $\poly(n, 1/\epsilon)$ (e.g., via quantum phase estimation or simply measure each of the $\poly(n)$ terms in the Hamiltonian \cite{nielsenQuantumComputationQuantum2010}).
Here, we consider the heat bath to have $\poly(n)$ size because we can only change the state of a $\poly(n)$-size system when we run the $\poly(n)$-time erasure protocol, and that is the only part of the heat bath for which we need to measure the energy.
In particular, we choose $\epsilon=1/(2q(n))$ and conduct the energy measurement to obtain an $\epsilon$-estimate of the work cost $\hat{W}$ that satisfies
\begin{equation}
    |\hat{W}-W_{\mathcal{E}}|\leq \epsilon=\frac{1}{2q(n)}
\end{equation}
with probability at least $1-1/\poly(n)$.
This only takes $\poly(n)$ more copies of the input states.
The time needed for this energy measurement is $\poly(n, 1/\epsilon)=\poly(n, 2q(n))=\poly(n)$ because $q(n)$ is a polynomial of $n$.
The second test is considered passed if $\hat{W}$ is smaller than $W_{\mathrm{Haar}}$.

Now we analyze the performance of this distinguisher.
When the input state is a PRS $\ket{\psi_x}^{\otimes N}, \ket{\psi_x}\in \mathcal{C}$, the erasure protocol $\mathcal{E}$ guarantees that the state is erased to $\ket{0}$.
Hence the first test is guaranteed to pass.
For the second test, we note that by assumption the work cost of $\mathcal{E}$ is
\begin{equation}
    W_{\mathcal{E}} < W_{\mathrm{Haar}} - 1/q(n).
\end{equation}
This means that the work cost estimate is
\begin{equation}
    \hat{W} \leq W_{\mathcal{E}}+\frac{1}{2q(n)}<W_{\mathrm{Haar}} - \frac{1}{q(n)} + \frac{1}{2q(n)} = W_{\mathrm{Haar}} - \frac{1}{2q(n)} < W_{\mathrm{Haar}}
\end{equation}
with probability at least $1-1/\poly(n)$.
Therefore, when the input state is indeed PRS, both tests pass and the distinguisher $\mathcal{D}$ outputs $0$ with probability at least $1-1/\poly(n)$.
This implies that
\begin{equation}
    \Pr_{\ket{\psi_x} \leftarrow \mathcal{C}}\left[\mathcal{D}\left(\left|\psi_x\right\rangle^{\otimes N}\right)=1\right] = 1- \Pr_{\ket{\psi_x} \leftarrow \mathcal{C}}\left[\mathcal{D}\left(\left|\psi_x\right\rangle^{\otimes N}\right)=0\right]<1/\poly(n).
\end{equation}

On the other hand, when the input states are Haar random $\ket{\psi}^{\otimes N}, \ket{\psi}\leftarrow \mu$, we note that Landauer's principle (i.e., \Cref{lem:landauer}) asserts that the energy cost of erasure must be at least $W_{\mathrm{Haar}}$.
Meanwhile, the event that the distinguisher $\mathcal{D}$ outputs $0$ (i.e., both tests pass) by construction is the event that Haar random states are erased successfully with energy cost less than $W_{\mathrm{Haar}}$, which is forbidden by Landauer's principle.
This means that
\begin{equation}
    \Pr_{\ket{\psi_x} \leftarrow \mathcal{C}}\left[\mathcal{D}\left(\left|\psi_x\right\rangle^{\otimes N}\right)=0\right]=0
\end{equation}
and hence
\begin{equation}
    \Pr_{\ket{\psi_x} \leftarrow \mathcal{C}}\left[\mathcal{D}\left(\left|\psi_x\right\rangle^{\otimes N}\right)=1\right]=1-\Pr_{\ket{\psi_x} \leftarrow \mathcal{C}}\left[\mathcal{D}\left(\left|\psi_x\right\rangle^{\otimes N}\right)=0\right]=1.
\end{equation}
Therefore, we have
\begin{equation}
\begin{split}
    \left|\Pr_{\ket{\psi_x} \leftarrow \mathcal{C}}\left[\mathcal{D}\left(\left|\psi_x\right\rangle^{\otimes N}\right)=1\right]-\Pr_{|\psi\rangle \leftarrow \mu}\left[\mathcal{D}\left(|\psi\rangle^{\otimes N}\right)=1\right] \right| &\geq \Pr_{|\psi\rangle \leftarrow \mu}\left[\mathcal{D}\left(|\psi\rangle^{\otimes N}\right)=1\right] - \Pr_{\ket{\psi_x} \leftarrow \mathcal{C}}\left[\mathcal{D}\left(\left|\psi_x\right\rangle^{\otimes N}\right)=1\right] \\
    &>1-1/\poly(n)
\end{split}
\end{equation}
close to one, contradicting \Cref{eq:prs-indist}.
This means that the there is no polynomial-time erasure protocol that erases PRS with work cost 
\begin{equation}
    W_{\mathcal{E}} < W_{\mathrm{Haar}} - 1/q(n)
\end{equation}
for any polynomial function $q(n)$.
Thus any polynomial-time quantum algorithm that erases $N=\poly(n)$ copies of pseudorandom states $\mathcal{C}=\{\ket{\psi_x}\}_{x=1}^m$ must have an energy cost that is at least
\begin{equation}
    W_{\mathrm{PRS}} \geq W_{\mathrm{Haar}} - \mathrm{negl}(n).
\end{equation}
This concludes the proof of the computational hardness of erasing PRS.

\end{document}